\documentclass[11pt]{article}
 \usepackage{justin}
\usepackage{Maryams_Macros}

\usepackage{project}
\def\thickhline{\noalign{\hrule height 1.1pt}}

\title{How fast can you find a good hypothesis?\thanks{Accepted for presentation at the Conference on Learning Theory (COLT) 2026.}}
\author{
Anders Aamand  \\
BARC, University of Copenhagen \\
\texttt{aa@di.ku.dk}
\and
Maryam Aliakbarpour\thanks{Department of Computer Science and
Ken Kennedy Institute}
\\ Rice University 
\\ \texttt{maryama@rice.edu}
\and\and
Justin Y.\ Chen \\
MIT \\
\texttt{justc@mit.edu}
\and
Sandeep Silwal \\
University of Wisconsin-Madison\\
\texttt{silwal@cs.wisc.edu}
}
\date{}

\begin{document}

\maketitle

\begin{abstract}
In the hypothesis selection problem, we are given sample and query access to finite set of candidate distributions (hypotheses), $\mathcal{H} = \{H_1, \ldots, H_n\}$, and samples from an unknown distribution $P$, both over a (possibly infinite) domain $\mathcal{X}$. The goal is to output a distribution $Q$ whose distance to $P$ is comparable to that of the nearest hypothesis in $\mathcal{H}$. Specifically, if the minimum distance is $\mathsf{OPT}$, we aim to output $Q$ such that, with probability at least $1-\delta$, its total variation distance to $P$ is at most $C \cdot \mathsf{OPT} + \varepsilon$. 

The statistical complexity of this problem is essentially solved: the optimal approximation for proper algorithms (where $Q \in \mathcal{H}$) is $C=3$ using $\Theta(\log(n/\delta)/\varepsilon^2)$ samples from $P$ and for improper algorithms (where $Q$ is not necessarily in $\mathcal{H}$) is $C=2$ using $\tilde{\Theta}(\log(n/\delta)/\varepsilon^2)$ samples from $P$. However, much remains unknown about the computational landscape of algorithms achieving these statistically optimal results.

In the improper setting, the algorithm achieving $C=2$ [Bousquet, Braverman, Kol, Efremenko, Moran, FOCS 2021] runs in time which grows polynomially with $|\mathcal{X}|$---it does not run in finite time for real-valued distributions. A promising path towards improved runtime is to consider improper algorithms which output a mixture $Q$ of the hypotheses $\mathcal{H}$ as such a distribution can be represented in $n$ words of memory. Whether mixtures can achieve $C < 3$ was posed as an open question in several prior works. We fully resolve the mixture setting by showing (1) a lower bound that no algorithm which outputs a mixture can achieve approximation better than $C = 3-2/n$ unless the number of samples is polynomial in $|\mathcal{X}|$, as well as (2) an algorithm which runs in time $\text{poly}(n)$ and achieves the same approximation guarantee.

In the proper setting, the previous state-of-the art result [Aliakbarpour, Bun, Smith, NeurIPS 2024] provided a near-linear in $n$ time algorithm with no dependence on $|\mathcal{X}|$, but it has sub-optimal dependence on the other parameters, running in $\tilde{O}(n/(\delta^3\varepsilon^3))$ time. We improve this time complexity to $\tilde{O}(n/(\delta \varepsilon^2))$, significantly reducing the dependence on the confidence and error parameters. We also show statistically optimal algorithms in two related settings where we achieve subquadratic runtime and high success probability: Assuming the numerical value of \textit{$\mathsf{OPT}$ is known} in advance, we present such an algorithm which runs in time $O(n\log^3(n/\delta)/\varepsilon^2)$. Allowing a \textit{preprocessing} step on the hypothesis class $\mathcal{H}$ before observing samples from $P$, we give such an algorithm with polynomial preprocessing time and query runtime of $\tilde{O}(n^{2-\Omega(\varepsilon)}/\varepsilon^3)$. 
\end{abstract}

\thispagestyle{empty}
\newpage
\thispagestyle{empty}
\renewcommand{\baselinestretch}{0.5}\normalsize
{\footnotesize \tableofcontents}
\renewcommand{\baselinestretch}{1.0}\normalsize
\thispagestyle{empty}
\newpage
\setcounter{page}{1}

\section{Introduction}

In the hypothesis selection problem, there are $n$ known distributions (hypotheses) $\HH = (H_1, \ldots, H_n)$ over a domain $\calX$ from which we can draw samples and query the probability density (or mass) function.
There is an unknown (true) distribution $P$ also over $\calX$ from which we can only draw samples.
The goal is to output a distribution $Q$ approximating $P$ in total variation distance which is competitive with the approximation of the best hypothesis in the set $\HH$.
Letting $\OPT = \min_{i \in [n]} \dtv(P, H_i)$, the output distribution $Q$ should satisfy
\begin{equation}\label{eq:optimal-approx}
    \dtv(P, Q) \leq C \cdot \OPT + \eps,
\end{equation}
with probability $1-\delta$ for some approximation factor $C \ge 1$, error parameter $\eps > 0$, and confidence parameter $\delta \in (0,1)$.
We will focus on the agnostic setting where $P$ is not necessarily in $\HH$ and will study both the \emph{proper} setting where $Q \in \HH$ and the \emph{improper} setting where $Q$ may not be one of the hypotheses in $\HH$.

Hypothesis selection, also known as density estimation, is a central problem in statistical estimation and a fundamental primitive in learning theory.
The proper case corresponds to the scenario where we are trying to identify, from statistical observations, the best explanation of the samples among a set of candidate distributions. In the improper case, hypothesis selection corresponds to distribution learning over $\mathcal{X}$ with respect to a finite but otherwise unstructured hypothesis class, namely the set $\HH$.
This contrasts with (a) direct distribution learning up to small total variation distance where satisfying $\dtv(P, Q) \leq \eps$ requires $\Omega(|\calX|/\eps^2)$ samples and is thus intractable for distributions over real values or large discrete domains, and (b) structured distribution learning such as Gaussian density estimation where the hypothesis class corresponds to a specific parameteric model.
Due to its generality, hypothesis selection has widespread applications as a subroutine in agnostic learning of a parametric class of distributions such as mixtures of Gaussians \cite{suresh2014near, daskalakis2014faster, ashtiani2018sample, ashtiani2020near}, Poisson binomials~\cite{daskalakis2012learning}, and junta distributions~\cite{AliakbarpourBR16}.

Hypothesis selection admits a surprising statistical phenomena: with only $O\p{\log(n/\delta)/\eps^2}$ samples from $P$, a hypothesis $H_i \in \HH$ can be found with error $\eps$, confidence $\delta$, and approximation factor $C=3$~\cite{yatracos1985rates, devroye1996universally, devroye1997nonasymptotic, devroye2001combinatorial}.
Thus, the problem can be solved with a sample size which scales logarithmically with the number of hypotheses and \emph{not at all} with the size of the underlying domain $\calX$.
Furthermore, this is tight as at least $\Omega\p{\log (n)/\eps^2}$ samples are necessary for any approximation,\footnote{This follows from a folklore result in distribution learning, see \Cref{thm:C-lower-bound}.} and for proper algorithms, $C=3$ is the optimal approximation factor achievable with domain-independent sample size (as in, the number of samples does not grow with $|\calX|$).\footnote{We identify a subtle bug in the prior proof that $C=3$ is tight for proper algorithms given in \cite{bousquet2019optimal}. In this paper, we generalize the proof and correct it. We have contacted the authors, and they have updated their manuscript.}
A recent breakthrough showed that a better approximation of $C=2$ is achievable for improper algorithms with essentially the same number of samples~\cite{bousquet2019optimal, bousquet2022statistically}.
It is also known that this is tight as prior work which studied the specific setting where $\HH$ is the set of piecewise constant distributions showed that $C<2$ is impossible to achieve with domain-independent sample size~\cite{ChanDSS14}.

While the statistical aspects of this problem are essentially solved, there are significant gaps in our knowledge of best results achievable by efficient algorithms. In recent years, the interplay between the sample complexity, approximation factor, and runtime\footnote{In addition to standard operations, we consider taking a single sample from $P$ or any $H_i$ as well as a single query to the pdf of any $H_i$ to take $O(1)$ time.} (among other computational quantities) of algorithms for hypothesis selection has been an active area of study~\cite{devroye2001combinatorial, mahalanabis2007density, ChanDSS14, acharya2014sorting, AcharyaD0S17, acharya2018maximum, bousquet2019optimal, bun2019private, gopi2020locally, bousquet2022statistically, aamand2023data, aliakbarpour2023hypothesis, aamand2024statistical, aliakbarpour2024optimal, Aliakbarpour2025, kamath2025local}.
Our work is motivated by two central questions regarding the computational efficiency of algorithms achieving the best possible sample size and approximation factor.

\paragraph{Improper, Finite Time}
The only known improper algorithms achieving $C=2$ with domain-independent samples run in time $\poly(n, |\mathcal{X}|)$~\cite{bousquet2019optimal, bousquet2022statistically}. 
Thus for real-valued distributions, while finite samples suffice for hypothesis selection,  there is no known \emph{finite time} algorithm which achieves an approximation better than $C = 3$ (this is achieved by proper algorithms).
We consider the following question.
\begin{boxquestion}\label{q:improper}
    Does there exist an improper algorithm which takes $O\p{\log(n/\delta)/\eps^2}$ samples from $P$, has approximation ratio $C < 3$, and runs in time which does not grow with $|\calX|$?
\end{boxquestion}

\paragraph{Proper, Near-Linear Time}
In the proper setting, algorithms achieving $C=3$ in $\poly(n)$ time (with no dependence on $\calX$) have been known for decades~\cite{yatracos1985rates, devroye2001combinatorial}. A natural target runtime is linear, or near-linear, in the number of hypotheses $n$ as $\Omega(n)$ time is trivially required to examine each hypothesis. A recent line of work has culminated in the state-of-the-art runtime of $\tilde{O}\p{n/(\delta^3 \eps^3)}$~\cite{aliakbarpour2024optimal}.\footnote{We use the notation $\tilde{O}(f(n))$ to suppress $\polylog(f(n))$ factors.} While near-optimal in terms of the dependence on $n$ for constant failure probability, this algorithm quickly degenerates when $\delta$ is small and has a seemingly extraneous factor of $1/\eps$. Intriguingly, standard boosting techniques (to transform a constant success probability algorithm to one that succeeds with high probability) fail in the domain-independent sample regime---we later expound upon this \textbf{high probability conundrum}.
This leads to our second guiding question. 
\begin{boxquestion}\label{q:proper}
    Does there exist a proper algorithm which takes $O\p{\log(n/\delta)/\eps^2}$ samples from $P$, has approximation ratio $C = 3$, and runs in time $\tilde{O}\p{n\log(1/\delta)/\eps^2}$?
\end{boxquestion}

We make significant progress on both directions, solving open problems raised in prior work. In their full generality, these problems remain open, and our work in this paper highlights potential avenues (or dead ends) to fully resolving them.

\subsection{Our Contributions}

\Cref{tbl:results} summarizes algorithmic results for hypothesis selection in various settings including our upper bounds.
\begin{table*}[ht]
\centering
\small
{\renewcommand{\arraystretch}{1.5}
\begin{tabular}{c|c|c|c|c}
\textbf{Algorithm}        & \textbf{Approx.} & \textbf{Runtime}                                 & \textbf{\begin{tabular}[c]{@{}c@{}}Failure \\ Prob.\end{tabular}} & \textbf{Note}         \\ \thickhline
\makecell{Cutting-with-margin \\ \cite{bousquet2022statistically}} & $2$ & $\poly(|\mathcal{X}|, n)$ & $1/\poly(n)$ & Improper \\
% \cite{mahalanabis2007density} & $2$ & $O\p{\frac{\log (1/\delta)}{\eps^2}}$ & $\delta$ & Expected, $n=2$ \\
\Cref{thm:ev_ub} (\textbf{ours}) & $3-\frac{2}{n}$ & $O\p{\frac{n^2\log n}{\eps^2}}$ & $1/\poly(n)$ & Expected Apx
\\ \thickhline
\makecell{Min distance \\ \cite{yatracos1985rates}} & 3                & $O\p{\frac{n^3 \log n}{\eps^2}}$                                    & $1/\text{poly}(n)$           &                       \\
\makecell{Min loss-weight \\\cite{mahalanabis2007density}} & 3 & $O\left( \frac{n^2 \log n}{\eps^2} \right)$ & $1/\text{poly}(n)$ \\
\makecell{Approx min distance \\\cite{aliakbarpour2024optimal}}                   & 3                & $\tilde{O}\left(\frac{n}{\eps^3 \delta^3}\right)$   
                                              & $\delta$                     &                       \\
\Cref{thm:fast} (\textbf{ours})           & 3                & $\tilde{O}\p{\frac{n}{\eps^2\delta}}$                                                & $\delta$                     &                       \\ \thickhline
\makecell{Ladder \\\cite{aliakbarpour2023hypothesis}}                 & 3                & $O\left( \frac{n \log( n/\delta)}{\eps^2 \delta} \right)$                                             & $\delta$                     & Known OPT\\
\Cref{thm:knownOPT} (\textbf{ours})           &    3 & $O\p{\frac{n\log^3(n/\delta)}{\eps^2}}$                         & $\delta$                     & Known OPT\\ \thickhline 
\Cref{cor:preprocessing} (\textbf{ours}) & 3 & $\tilde{O}\p{\frac{n^{2-\Omega(\eps)} }{\eps^3}}$ & $1/\poly(n)$ &  Preprocessing
\\ \thickhline
\makecell{Seed \\ \cite{aliakbarpour2024optimal}} & 4 & $\tilde{O}\left(\frac{n\log(n/\delta)\log(1/\delta)}{\eps^2}\right)$ & $\delta$ 
\\ \makecell{Modified Knockout \\ \cite{aamand2023data}}        & 9                & $O\left( \frac{n \log(1/\delta)}{\eps^2} \right)$ & $\delta$                     & $\delta \ge n^{-1/5}$ 
\end{tabular}
}
\caption{Relevant algorithms for hypothesis selection in various settings. Algorithms are proper unless otherwise specified. All algorithms use $O(\log(n/\delta)/\eps^2)$ samples or $O(\log (n)/\eps^2)$ samples if $\delta$ is not specified other than that of \cite{bousquet2022statistically} which contains an additional $\log\log n$ factor.}
\label{tbl:results}
\end{table*}
 
\subsubsection{(Improper) Optimal Approximation of Mixture Distributions}

The improper algorithms achieving $C=2$ rely on finding a distribution $Q$ over $\calX$ that satisfies certain distance constraints to the hypotheses $\HH$~\cite{bousquet2019optimal, bousquet2022statistically}.
Their runtimes depend polynomially on $|\calX|$ and even require $\Omega(|\calX|)$ bits to specify the output distribution $Q$.

On the other hand, a curious result is known for the special case where $n=2$ (there are only two hypotheses in $\HH$).
In this setting, there exists an algorithm which outputs a distribution $Q$ which is a mixture of the two hypothesis in $\HH$ which achieves an approximation of $C=2$~\cite{mahalanabis2007density}.
The algorithm runs in time $O(\log(1/\delta)/\eps^2)$ and describing the distribution only requires two numbers corresponding to the mixture weights on the two hypotheses, even if the underlying domain is infinite.
As we have sample and pdf query access to the hypotheses, the mixture weights are sufficient for the same accesses to the output distribution.

In fact, \cite{mahalanabis2007density} give a slightly stronger result by considering the \emph{expected} approximation ratio of randomized proper algorithms.
Let $Q$ be a random variable for the hypothesis output by such an algorithm, then the guarantee is of the form:
\begin{equation}
    \E[Q]{\dtv(P, Q)} \leq C \cdot \OPT + \eps.
\end{equation}
When $n=2$, \cite{mahalanabis2007density} give a proper algorithm achieving expected approximation ratio $C=2$.
This immediately yields an improper algorithm with at least as good of an approximation by the triangle inequality by outputting the mixture of the hypotheses whose weights correspond to the output probabilities of the randomized algorithm.
Generalizing to larger hypothesis classes with $n > 2$ remained an open question posed in \cite{mahalanabis2007density, bousquet2019optimal}.
\begin{quote}
    \emph{
    For a set $\HH$ of $n$ hypotheses for arbitrary $n > 2$, can we output a distribution $Q$ in the convex hull of $\HH$ which achieves approximation factor $C < 3$?
    }
\end{quote}
Answering this question affirmatively would be a direct and promising way to answer \Cref{q:improper}.
Mixture distributions are the most natural improper class that can be sampled from in domain-independent time (via sample access to the hypotheses in $\HH$).

% On the other hand, a curious observation from \cite{mahalanabis2007density} considers a different perspective on the guarantee of hypothesis selection: instead of finding a hypothesis with a high-probability distance guarantee, the focus is on the expected distance of the (randomized) output hypothesis $H_i$. That is, the algorithm outputs a distribution over $\HH$ and aims for a guarantee of the form:
% $$\E{\dtv(P, H_i)} \leq C \cdot \OPT + \eps\,.$$
% One might naturally ask: does this weaker requirement on the average performance of the algorithm lead to substantial improvement? Mahalanabis et al.~\cite{mahalanabis2007density} show that for a hypothesis class of size $n=2$, $O(\log(n)/\eps^2) = O(1/\eps^2)$ samples are sufficient to guarantee $C=2$.
% (Note that the output is still in $\HH$ and competing with the best hypothesis in $\HH$.)
% This circumvents the required lower bound of $C=3$ when a single hypothesis must be output. This finding raises an open question, which was also posed in~\cite{mahalanabis2007density, bousquet2019optimal}:

% \begin{quote}
%     For an arbitrarily large finite class $\HH$, can we get an expected approximation factor $C < 3$?
% \end{quote}

We resolve this open question and provide a complete understanding of the optimal approximation factor for this expected approximation/improper mixture setting.
In these settings, the best possible approximation is $C=3-2/n$, and we give an algorithm which runs in $\poly(n)$ time and uses the optimal sample complexity which achieves this bound.
When $n=2$, we recover the result of \cite{mahalanabis2007density}, but for general $n$, this implies that there is no absolute constant approximation $C < 3$ achievable by an algorithm which outputs a mixture.
Our lower bound holds against algorithms which output a mixture of the hypotheses while the upper bound is a randomized proper algorithm, demonstrating that while \emph{a priori} improper mixtures could improve upon the expected approximation ratio of proper algorithms, they cannot in the worst-case.

To prove the lower bound, we construct a set of hypotheses and family of hard true distributions such that all but the optimal hypothesis are $3$-approximations.
The resulting approximation factor of $3-2/n = \frac{3(n-1) + 1}{n}$ essentially comes from showing that no algorithm with domain-independent sample size can do better than outputting a uniform mixture of the hypotheses.
Along the way, we identify and correct a subtle error in an existing proof of the weaker result that no algorithm with domain-independent sample size can distinguish between the optimal hypothesis and a single other $3$-approximation \cite{bousquet2019optimal}, thereby giving the first correct proof that $C=3$ is the best achievable approximation by proper algorithms (see \Cref{rem:lb-bug} for details).

\begin{theorem}[Informal version of \Cref{cor:proper-lb,cor:exp-lower-bound}]
Consider a sample size $s$ which is a function of $n, \eps, \delta$.
There exists a sufficiently large, finite domain size $|\XX|$ such that no randomized algorithm with $s$ samples can output a convex combination of the hypotheses with expected approximation factor less than $3 - \frac{2}{n} - o\left(\frac{1}{n}\right)$.
Furthermore, no proper algorithm with $s$ samples can output a hypothesis with approximation less than $3$ with probability more than $\frac{1}{n}\p{1 + o(1)}$.
\end{theorem}

% We also remark that $\Omega(\log(n)/\eps^2)$ samples are necessary to obtain any constant factor expected value approximation; see \Cref{thm:C-lower-bound}.

Our upper bound generalizes the algorithm of \cite{mahalanabis2007density} to more than two distributions and matches the $3-2/n$ bound.

\begin{restatable}{theorem}{thmEVUB}
\label{thm:ev_ub}
There exists an algorithm that takes $s = O(\log(n)/\eps^2)$ samples from $P$, runs in $\tilde{O}(n^2/\eps^2)$ time, and, with high probability in $n$, outputs an explicit distribution $\calD$ over $\HH$ such that:
\[
\E[H \sim \calD]{\dtv(H, P)} \leq \left(3-\frac{2}{n}\right)\OPT + \eps.
\]
\end{restatable}

Overall, our results in this setting, which resolve an open question in \cite{mahalanabis2007density, bousquet2019optimal}, yield mixed conclusions for the state of \Cref{q:improper}. For constant $n$, approximations less than $3$ are achieved by our algorithm which runs in domain-independent time. On the other hand, our lower bound shows that mixtures cannot achieve an absolute constant approximation better than $C=3$. Moving forward to answer \Cref{q:improper}, it is unclear even what candidate classes of output distributions beyond mixtures could allow for better approximation while having description length polynomial in $n$.

\subsubsection{(Proper) Fast Algorithms for the Moderate  Probability Regime}

The first proper algorithm for hypothesis selection which achieved $C=3$ has a runtime of $O(n^3\log(n/\delta)/\eps^2)$~\cite{yatracos1985rates, devroye2001combinatorial}.
The cubic dependence on $n$ is less than ideal and can make hypothesis selection a computational bottleneck in applications.
A significant body of recent work has focused on reducing the running time down to nearly-linear in $n$~\cite{mahalanabis2007density, daskalakis2014faster, acharya2014sorting, acharya2018maximum, aamand2023data, aliakbarpour2023hypothesis, aliakbarpour2024optimal}.
These attempts, however, often came at the cost of a worse approximation factor ($C \geq 4$) or required special assumptions. Recently, \cite{aliakbarpour2024optimal} broke this barrier, presenting an algorithm with the optimal approximation factor $C=3$, sample complexity of $O(\log(n/\delta)/\eps^2)$, and a nearly-linear total running time of $\tilde{O}\left( \frac{n}{\eps^3 \delta^3} \right)$.

Despite being linear in $n$, this complexity is sub-optimal with respect to $\delta$ and $\eps$. 
These results therefore only yield improvements over quadratic time in the constant or \emph{moderate} success probability regime where $1/\delta$ is significantly sublinear in $n$.
Interestingly, with the same number of samples, the dependence on $\delta$ and $\eps$ improves if we relax either the approximation factor $C$ or the linear-time requirement. On the one hand, the algorithm of~\cite{mahalanabis2007density} achieves the optimal $C=3$, but has a quadratic running time of $O(n^2\log(n/\delta)/\eps^2)$. On the other hand, there exist algorithms running in time $\tilde{O}(n\log(n/\delta)/\eps^2)$ or even $O(n\log(1/\delta)/\eps^2)$ if they aim for a larger approximation factor, such as $C = 4$ or $C=9$~\cite{aamand2023data,  aliakbarpour2024optimal}. 

The polynomial dependence on the failure probability $\delta$ is particularly unusual, as typical statistical tasks exhibit a milder $\log(1/\delta)$ dependency. This issue is present in several other algorithms that achieved $C < 9$ in near-linear time~\cite{aliakbarpour2023hypothesis, Aliakbarpour2025}. This sub-optimality stems from the failure of standard amplification techniques; for instance, repeating a process $O(\log(1/\delta))$ times that individually guarantees $C=3$ with constant probability and taking the ``best option'' among the returned hypotheses is itself a new hypothesis selection problem, so it only yields a combined guarantee of $C=9$ as the true best option can only be approximated up to a factor of $3$. Such amplification methods are ineffective for algorithms achieving the optimal approximation factor. 
This \textbf{high probability conundrum} stems from the fact that there is no way to validate the approximation quality of any hypothesis with domain-independent samples.
In fact, with $o(\XX/\log(\XX))$ samples from $P$, it is known that estimating the $\dtv(P, Q)$ for any fixed distribution $Q$ is impossible~\cite{valiant2017unseen}.
This phenomenon prompts a key question about improving the dependence on $\delta$ and $\eps$ for the $C=3$ case while maintaining near-linear dependence on $n$. This was posed as an open problem in~\cite{aliakbarpour2024optimal, Aliakbarpour2025} and is the key challenge in answering \Cref{q:proper}.

We develop an algorithm that significantly improves the running time of~\cite{aliakbarpour2024optimal} by a factor of $1/(\delta^2 \eps)$. Our main result is summarized in the following theorem.

\begin{restatable}{theorem}{thmFAST}
\label{thm:fast}
For any $\delta > 1/n$, there exists an algorithm for hypothesis selection that uses $O(\log(n/\delta)/\eps^2)$ samples, achieves the optimal approximation factor of $C=3$, and runs in time $\tilde{O}\p{\frac{n}{\eps^2 \delta}}$.
\end{restatable}

This result allows for sub-quadratic time hypothesis selection with optimal sample complexity and approximation factor for a much larger range of failure probabilities $\delta$ than prior work.
The failure probability dependence achieved by our algorithm is the best possible by algorithms following the high-level strategy of iteratively filtering out a $\delta$ fraction of ``bad'' hypotheses until either there are no remaining hypotheses or a $1-\delta$ fraction of the hypotheses are ``good''.
This filtering strategy, used by our algorithm and previously in~\cite{aliakbarpour2023hypothesis, aliakbarpour2024optimal}, requires $O(\log(n)/\delta)$ iterations.
It remains an extremely intriguing open question whether it is possible to get rid of the polynomial dependence on $1/\delta$ in the runtime while ensuring $C = 3$, which would resolve \Cref{q:proper}.
This would likely require new techniques to bypass this $1/\delta$ barrier.

\subsubsection{(Proper) Extra Information: Known Upper Bound on $\OPT$} Another relaxation of hypothesis selection studied in the literature involves fixing the \emph{quantity} $\OPT$. Interestingly, knowing an approximation to this numerical value, but not the underlying hypothesis, allows for faster algorithm or algorithms with better approximation factors~\cite{daskalakis2014faster, bun2019private, aliakbarpour2023hypothesis}. 
The primary advantage of such algorithms is that knowing $\OPT$ enables us to better filter the sub-optimal hypotheses. If we find evidence that a hypothesis is significantly farther from the true distribution than $\OPT$, we can remove it from consideration, a step that is generally not possible for standard algorithms. 
%much lower failure probability (\sandeep{is there a high level intuition of why? maybe expand on that in tech overview because this is quite non obvious}):
In particular, 
\cite{aliakbarpour2023hypothesis} gave an algorithm that achieves approximation factor $C = 3$ in $\tilde{O}(n/(\delta\eps^2))$ time. 
Naturally, we ask the following analogue of \Cref{q:proper}:
\begin{quote}
\emph{
Is there an algorithm which knows the value of $\OPT$ and achieves $C = 3$ using $s = O(\log (n/\delta)/\eps^2)$ samples and in runtime $\tilde{O}\p{n\log(1/\delta)/\eps^2}$?}
\end{quote}
We answer this question affirmatively up to logarithmic factors in $n$ and $1/\delta$: the sample complexity is optimal, and when $\delta=1/\poly(n)$, the running time is within $\polylog(n)$ factors of $n/\eps^2$. In particular, we have the following theorem.

\begin{restatable}{theorem}{thmKnownOPT}\label{thm:knownOPT}
        Suppose the algorithm is given $R \in [0,1]$ as part of the input with the guarantee that $R\geq\OPT$. Then, there exists an algorithm with sample complexity $O\p{\frac{\log (n/\delta)}{\eps^2}}$ and runtime $O\p{\frac{n\log^3 (n/\delta)}{\eps^2}}$, which with probability at least $1-\delta$ returns a hypothesis $H_i$ such that $\dtv(P, H_i) \leq 2 \cdot \OPT + R   + \eps$.
\end{restatable}

Note that the theorem only requires as input a value $R\geq \OPT$, but the approximation factor grows with $R$. In particular, with an estimate $\widetilde{OPT}$ such that $\OPT\leq \widetilde{OPT}\leq \alpha\cdot \OPT$, the algorithm is a $(2+\alpha)$-approximation (up to the additive $\eps$).

We note that \Cref{thm:knownOPT} is incomparable to \Cref{thm:fast}. While \Cref{thm:knownOPT} has logarithmic dependence on $1/\delta$, it uses knowledge of the value of $\OPT$ which cannot be easily estimated for the same reason as the aforementioned high probability conundrum. 
Indeed, we note that the sample complexity of learning even an estimate $\widetilde{\OPT}$ such that $\OPT\leq \widetilde{\OPT} \leq \alpha\cdot \OPT+\eps$ depends polynomially on the domain size. If $U$ is uniform over $\mathcal{X}$, it is known that the sample complexity of distinguishing between whether we receive samples from $U$ or a distribution $P$ with $d_{TV}(P,U)\geq 2\eps$ is $\Theta(\sqrt{|\mathcal{X}|}/\eps^2)$~\cite{paninski2008coincidence,diakonikolas2014testing,valiant2017automatic}. If $\mathcal{H}=\{U\}$, the desired estimate $\widetilde{\OPT}$ has $\widetilde{\OPT}\leq \eps$ if we receive samples from $U$, and $\widetilde{\OPT}\geq 2\eps$ if we receive samples from $P$ and would thus let distinguish between the two cases.

\subsubsection{(Proper) Impact of Preprocessing} An important extension to the standard hypothesis selection problem is the \textit{pre-processing setting}, which is motivated by applications where the selection task is performed repeatedly for a fixed hypothesis set $\mathcal{H}$ but against different unknown distributions $P$. In this model, an algorithm first analyzes $\mathcal{H}$ before receiving any samples and can spend time to \emph{preprocess} $\mathcal{H}$. This approach has been studied in several works~\cite{mahalanabis2007density, aamand2023data, aamand2024statistical}.
    
First, note that exponential preprocessing can trivialize the problem on a finite domain: for $s=O(\log(n/\delta)/\eps^2)$ samples, one can precompute the output under all possible length-$s$ sample sequences using $|\calX|^{O(\log(n/\delta)/\eps^2)}$ preprocessing time and then perform table lookup at query time. In the constant-$\eps$, inverse-polynomial failure probability regime, this is $|\calX|^{O(\log n)}$ preprocessing. \cite{mahalanabis2007density} showed that $2^{O(n)}$ time preprocessing can be used to pre-compute all possible computational paths for linear in $n$ query time. The exponential time complexity of these preprocessing procedures makes them prohibitively expensive. A natural question, first raised in~\cite{mahalanabis2007density}, is to what extent a polynomial-time preprocessing step can improve overall performance. In particular, we ask:

    \begin{quote}
        \emph{How fast can we obtain $C = 3$ using $s = O(\log n/\eps^2)$ samples if we are only allowed polynomial preprocessing time? }
        %Is it possible to achieve linear running time using only polynomial-time preprocessing?
    \end{quote}

We show that preprocessing can indeed be useful, allowing us to break the quadratic-time barrier in $n$ even in the high-probability regime where $\delta$ is an arbitrarily small inverse polynomial in $n$. We have the following theorem.

\begin{restatable}{theorem}{corPreProcessing}
    \label{cor:preprocessing}
    For every $\ \eps \in (0, 1)$, there exists an algorithm which uses $\textup{poly}(|\calX|, n)$ preprocessing time (without knowledge of $P$), $O(\log(n)/\eps^2)$ samples from $P$, and outputs a hypothesis $H_i$ satisfying 
    \[
    \dtv(P, H_i) \leq 3 \cdot \OPT  + O(\eps)
\]
with probability $1 - 1/\textup{poly}(n)$ in time $\tilde{O}\left(\frac{ n^{2-\Omega(\eps)}}{\eps^3} \right)$.
\end{restatable}

In contrast to our other results, our preprocessing result only applies when $\XX$ is finite though the sample size and query time do not depend on $|\calX|$.
We remark that progress was made in a different preprocessing setting~\cite{aamand2023data, aamand2024statistical}, where it is possible to obtain \emph{sublinear} runtime in $n$ if we further restrict the hypotheses to be discrete distributions, require that $P \in \HH$ (so that $\OPT$ is 0), and allow the sample size to grow (sublinearly) in $|\calX|$.

\section{Technical Overview}\label{sec:technical_overview}

\subsection{Preliminaries}\label{sec:prelim}

The key toolkit for designing algorithms with constant approximation factors in $s = O\p{\frac{\log(n/\delta)}{\eps^2}}$ samples are ``\Scheffe sets'' and ``semi-distances'' used extensively in prior works on this problem (e.g., ~\cite{devroye2001combinatorial,daskalakis2014faster,aliakbarpour2024optimal}).
We present the basic definitions and facts about these objects which will be used throughout this work.

\begin{definition}[\Scheffe set~\cite{devroye2001combinatorial}]\label{def:scheffe_set}
The \emph{\Scheffe set} for any pair $i \neq j \in [n]$ is 
\begin{equation}
    \Sij =
    \begin{cases}
        \crb{x \in \calX : H_i(x) < H_j(x)} & \qquad \text{ if } i < j \\
        \crb{x \in \calX : H_i(x) \leq H_j(x)} & \qquad \text{ if } i > j
    \end{cases}.
    \footnote{The difference between the two cases is arbitrary and exists for tiebreaking so that $\Sij \cup S_{j \rightarrow i} = \calX$.}
\end{equation}
It is straightforward to see that $\Sij$ is the set that witnesses the total variation distance of distributions $H_i$ and $H_j$:
\begin{equation}
    H_j(\Sij) - H_i(\Sij) = \sup_{S \subseteq \calX} H_j(S) - H_i(S) = \dtv(H_j, H_i).
\end{equation}
\end{definition}

% % \subsection{Computational Model}\label{sec:computational_model}

% \justin{update computational model}
% \paragraph{Computational Model}
% Our computational model is the standard one used for hypothesis selection \cite{devroye2001combinatorial}. We assume that each of the following operations take $O(1)$ time:
% \begin{enumerate}
%     \item We can request a sample from $P$,
%     \item Given any domain element $x \in \calX$ and any $i \ne j$, we can ask if $H_i(x) > H_j(x)$,
%     \item For any $k \in [n]$, in $O(1)$ time we can compute $H_k(\Sij)$, i.e. the probability mass that $H_k$ places on the \Scheffe set $\Sij$. 
% \end{enumerate}
% The second condition can be relaxed to only output $H_k(\Sij)$ up to an additive error of $\eps$ but we ignore this detail for simplicity. 
% This is the standard computational model used in prior work: if these operations take time $T$, then ours and prior bounds often are simply multiplicatively scaled by $T$.
% The third operation can be implemented given sampling and query access to the hypotheses \cite{canonne2014testing} or in a one-time preprocessing of the hypotheses.

% % We will assume that in $O(1)$ time,  we can query $H_i(S)$, the probability density assigned to $S$ under $H_i(S)$, for any (measurable) $S \subseteq \calX$, and we can find a set $\Sij = \argmax_{S \subseteq \calX} H_j(S) - H_i(S)$.

\begin{definition}[Semi-distances~\cite{devroye2001combinatorial}]\label{def:semi_distances}
For any pair $i \neq j \in [n]$, we define the \emph{semi-distance} $\semi{i}{j} \in [0,1]$ as $ \semi{i}{j} = \abs{H_j(\Sij) - P(\Sij)}.$
We will let $\semi{i}{i} = 0$.
We define the \emph{max semi-distance} $W(H_j)$ as $W(H_j) = \max_{i \in [n]} \semi{i}{j}.$
\end{definition}
Intuitively, $\semi{i}{j}$ is a guess of the total variation distance between $H_j$ and $P$ via the \Scheffe set differentiating $H_i$ and $H_j$.
Note that $\semi{i}{j}$ and $\semi{j}{i}$ use essentially the same Scheff\'e set as $\Sij = \overline{S_{j \rightarrow i}}$ but are measuring the distance between $P$ and two different hypotheses $H_j$ and $H_i$, respectively.
% Note that as $\Sij = \overline{S_{j \rightarrow i}}$, $\semi{i}{j} = \abs{H_j(\Sij) - P(\Sij)} = \abs{H_j(\Sji) - P(\Sji)}$. \sandeep{should we explicitly say that $\semi{i}{j}$ is not necessarily $\semi{j}{i}$ just to avoid confusion for reader?} 
We use the following key properties of semi-distances also utilized in prior works~\cite{devroye2001combinatorial,daskalakis2014faster,aliakbarpour2024optimal}. We defer these standard proofs to \Cref{sec:omitted_proofs}.

\begin{proposition}[Underestimation]\label{prop:semidist-underest}
For all $i \neq j \in [n]$, $\semi{i}{j} \leq \dtv(P, H_j)
.$ \end{proposition}

Let $i^* = \argmin_{i \in [n]} \dtv(P, H_i)$ be the index of the optimal hypothesis.
\begin{proposition}[Semi-Distance TV Approximation]\label{prop:semidist-approx}
For any $i \neq j \in [n]$, $ \dtv(P, H_j) \leq \dtv(P, H_i) + \semi{i}{j} + \semi{j}{i}.$
In particular, for any $j \in [n]$, as $\semi{j}{i^*} \leq \OPT$,  $ \dtv(P, H_j) \leq 2\OPT + \semi{i^*}{j}.$
\end{proposition}
This is the source of the $3$-approximation in prior work. If we find $H_j$ with $W(H_j) \leq \OPT$, then it must be a $3$-approximation as $W(H_j) \geq \semi{i}{j}$. We do not know $\OPT$ but simply picking the hypothesis with smallest $W(H_j)$ will suffice as $W(H_j) \leq W(H_{i^*}) \leq \OPT$ using \Cref{prop:semidist-underest}.

% \begin{corollary}%\label{corollary:max-semi-to-bound}
% For any $j$,
% \[
% \dtv(H_j,P)\leq 2\OPT + W(H_j).
% \]
% \end{corollary}
% \begin{proof}
% If $j=i^*$, this is clear. Otherwise, this follows by applying~\Cref{prop:semidist-approx} with $i=i^*$ and using that $\semi{j}{i^*}\leq \OPT$ and $\semi{i^*}{j}\leq W(H_j)$
% \end{proof}

% Another useful implication is that for any distinct $i, j, k \in [n]$,
% \[
%     \dtv(P, H_j) \geq \frac{\dtv(P, H_i) - \semi{j}{i}}{2} \geq \frac{\semi{k}{i} - \semi{j}{i}}{2}.
% \]
% This means, if we find a ``far'' $H_i$ (as witnessed by $H_k$), we can also rule out any $H_j$ which measures a small semi-distance $\semi{j}{i}$. In other words, hypotheses which fail to measure the distance of bad hypotheses are themselves bad.

Again recall our computation model: we can sample from $P$ and any distribution in $\HH$ in $O(1)$ time. Furthermore, we can query the pdf of any distribution in $\HH$ in $O(1)$ time as well. 

\begin{proposition}[Approximating Semi-Distances via Samples]\label{lem:sample-scheffe}
With $s = \Theta\p{\frac{\log(1/\delta)}{\eps^2}}$ samples from $P$ and $O(s)$ time, a semi-distance $\semi{i}{j}$ can be estimated by $\hsemi{i}{j}$ such that $|\semi{i}{j} - \hsemi{i}{j}| \leq \eps$
with probability $1 - \delta$. Similarly, a total of $O\left(\frac{ \log(n/\delta)}{\epsilon^2}\right)$
samples from $P$ and $O\left(\frac{n^2 \log(n/\delta)}{\epsilon^2}\right)$ total time are required to accurately estimate all $\binom{n}{2}$ pairs of hypothesis with an overall success probability of at least $1-\delta$.
\end{proposition}

% \textcolor{red}{
% Anders' example: minimizing max semi-distance is hard. Consider a graph with all $1$ edges except for a matching with $1.1$ edges and one node not part of the matching. This one node is hard to find.
% }

\subsection{Approximation Lower Bound for Mixtures}

In \Cref{sec:expected-lb}, we correct and generalize the prior lower bound construction~\cite{bousquet2019optimal} from $2$ to $n$ distributions.
The overall goal is to construct $n$ hypotheses $H_1, \ldots, H_n$ along with $n$ families of distributions $\calS_1, \ldots, \calS_n$ such that the following hold:
\begin{itemize}
    \item (Approximation Factor) For any $P_i \in \supp(\calS_i)$ $H_i$ is the closest distribution to $P_i$ and the rest of $H_j$'s ($j \neq i$), are far from $P_i$:  $\dtv(P_i, H_j) \geq 3 \dtv(P_i, H_i) = 3\OPT$.
    \item (Indistinguishability) Consider a guessing game where the algorithm is given samples from a $P_i \sim \calS_i$ for $i$ chosen uniformly at random from $[n]$ and tries to guess the index $i$.
    If the number of samples is too small (e.g., dimension-independent), then no algorithm can correctly guess with probability greater than $\frac{1}{n}\p{1 + o(1)}$.
\end{itemize}

A hard instance satisfying these properties implies that no algorithm can distinguish between a single hypothesis with distance $\OPT$ and $n-1$ hypotheses with distance $3\OPT$ for any $n$.
So, the expected approximation factor goes to $3$ as $n$ increases with the best possible approximation, ignoring to lower order terms, being $\frac{3(n-1) + 1}{n} = 3 - \frac{2}{n}$. Our specific lower bound instance even works against \emph{mixtures} of $\mathcal{H}$. We will discuss why at the end of this section.

\begin{figure}[th]
\centering
\begin{tikzpicture}[scale=0.68]
    % Define parameters for visualization
    \def\n{3}  % number of hypothesis pairs
    \def\k{3}  % number of subintervals (k)
    \def\ell{8}  % length of each subinterval (ℓ) - INCREASED FOR CLARITY
    \def\unitwidth{0.15}  % width per element in [d]
    \def\subintervalwidth{1.0}  % width for each T_i^j subinterval (reduced for page fit)
    \def\heightunit{1.6}  % height unit for probabilities (reduced for compactness)
    
    % Axis
    \draw[->] (0,0) -- (19,0) node[right] {$x \in \calX$};
    
    % Draw major interval divisions (T_i)
    \foreach \i in {1,...,6} {
        \pgfmathsetmacro{\xpos}{\i*\k*\subintervalwidth}
        \draw (\xpos,0.1) -- (\xpos,-0.1);
        \node[below] at (\xpos - \k*\subintervalwidth/2,-0.2) {$T_{\i}$};
    }
    
    % Draw subinterval divisions (T_i^j) - medium lines
    \foreach \i in {1,...,18} {
        \pgfmathsetmacro{\xpos}{\i*\subintervalwidth}
        \draw[gray!60, thick] (\xpos,0.08) -- (\xpos,-0.08);
    }
    
    % Additional labels
    % \node at (6*\k*\subintervalwidth + 0.8, 0) {$\cdots$};
    
    % Hypothesis H_i (i=2, so affects T_3 and T_4)
    \begin{scope}[yshift=7.5cm]
        \node[left] at (-0.3, \heightunit*0.5) {$H_i$};
        \draw[->] (0,0) -- (19,0);
        
        % Baseline (1/d everywhere)
        \draw[dashed, gray] (0,\heightunit*0.5) -- (18.5,\heightunit*0.5);
        
        % T_1, T_2: uniform 1/d
        \fill[blue!20] (0,0) rectangle (2*\k*\subintervalwidth,\heightunit*0.5);
        
        % T_3 (T_{2i-1}): (1+β)/d - higher
        \fill[green!40] (2*\k*\subintervalwidth,0) rectangle (3*\k*\subintervalwidth,\heightunit*0.65);
        
        % T_4 (T_{2i}): (1-β)/d - lower
        \fill[red!40] (3*\k*\subintervalwidth,0) rectangle (4*\k*\subintervalwidth,\heightunit*0.35);
        
        % T_5, T_6: uniform 1/d
        \fill[blue!20] (4*\k*\subintervalwidth,0) rectangle (6*\k*\subintervalwidth,\heightunit*0.5);
        
        % Labels
        \node at (2.5*\k*\subintervalwidth, \heightunit*0.9) {\small $\frac{1+\beta}{|\calX|}$};
        \node at (3.5*\k*\subintervalwidth, \heightunit*0.9) {\small $\frac{1-\beta}{|\calX|}$};
        \node at (5*\k*\subintervalwidth, \heightunit*0.9) {\small $\frac{1}{|\calX|}$};
        
        % Braces
        \draw[decorate,decoration={brace,amplitude=5pt,mirror}] 
            (2*\k*\subintervalwidth,-0.3) -- (3*\k*\subintervalwidth,-0.3) node[midway,below=5pt] {\small $T_{2i-1}$};
        \draw[decorate,decoration={brace,amplitude=5pt,mirror}] 
            (3*\k*\subintervalwidth,-0.3) -- (4*\k*\subintervalwidth,-0.3) node[midway,below=5pt] {\small $T_{2i}$};
    \end{scope}
    
    % Hypothesis H_j (j=1, so affects T_1 and T_2)
    \begin{scope}[yshift=5cm]
        \node[left] at (-0.3, \heightunit*0.5) {$H_{i'}$};
        \draw[->] (0,0) -- (19,0);
        
        % Baseline
        \draw[dashed, gray] (0,\heightunit*0.5) -- (18.5,\heightunit*0.5);
        
        % T_1 (T_{2j-1}): (1+β)/d - higher
        \fill[green!40] (0,0) rectangle (\k*\subintervalwidth,\heightunit*0.65);
        
        % T_2 (T_{2j}): (1-β)/d - lower
        \fill[red!40] (\k*\subintervalwidth,0) rectangle (2*\k*\subintervalwidth,\heightunit*0.35);
        
        % T_3, T_4, T_5, T_6: uniform 1/d
        \fill[blue!20] (2*\k*\subintervalwidth,0) rectangle (6*\k*\subintervalwidth,\heightunit*0.5);
        
        % Labels
        \node at (0.5*\k*\subintervalwidth, \heightunit*0.9) {\small $\frac{1+\beta}{|\calX|}$};
        \node at (1.5*\k*\subintervalwidth, \heightunit*0.9) {\small $\frac{1-\beta}{|\calX|}$};
        \node at (5*\k*\subintervalwidth, \heightunit*0.9) {\small $\frac{1}{|\calX|}$};
        
        % Braces
        \draw[decorate,decoration={brace,amplitude=5pt,mirror}] 
            (0,-0.3) -- (\k*\subintervalwidth,-0.3) node[midway,below=5pt] {\small $T_{2i'-1}$};
        \draw[decorate,decoration={brace,amplitude=5pt,mirror}] 
            (\k*\subintervalwidth,-0.3) -- (2*\k*\subintervalwidth,-0.3) node[midway,below=5pt] {\small $T_{2i'}$};
    \end{scope}
    
    % Distribution P_i (from S_i)
    \begin{scope}[yshift=1.5cm]
        \node[left] at (-0.3, \heightunit*0.5) {$P_i$};
        \draw[->] (0,0) -- (19,0);
        
        % Baseline
        \draw[dashed, gray] (0,\heightunit*0.5) -- (18.5,\heightunit*0.5);
        
        % T_1, T_2: uniform 1/d
        \fill[blue!20] (0,0) rectangle (2*\k*\subintervalwidth,\heightunit*0.5);
        
        % T_3 (T_{2i-1}): detailed pattern showing f^j(x mod ℓ)
        % For each subinterval T_{2i-1}^j
        \foreach \j in {1,2,3} {
            \pgfmathsetmacro{\substart}{(2*\k + \j - 1)*\subintervalwidth}
            
            % Draw subinterval boundary
            \draw[gray!60, thick] (\substart,0) -- (\substart, \heightunit*1.1);
            
            % Pick exactly one position to have mass 0, others have (1+β)/d
            \pgfmathsetmacro{\extremepos}{int(mod(\j*3 + 1, \ell))}  % one extreme per subinterval
            
            % Show individual positions within subinterval (ℓ = 8 positions)
            \foreach \pos in {0,1,2,3,4,5,6,7} {
                \pgfmathsetmacro{\xpos}{\substart + \pos*\subintervalwidth/\ell}
                \pgfmathsetmacro{\barwidth}{\subintervalwidth/\ell - 0.02}
                
                \pgfmathsetmacro{\isextreme}{int(\pos == \extremepos ? 1 : 0)}
                
                \ifnum\isextreme=1
                    % f^j(x mod ℓ) = 1 → mass is 0 (exactly ONE per subinterval)
                    \fill[white] (\xpos, 0) rectangle (\xpos + \barwidth, \heightunit*0.65);
                    \draw[gray!50] (\xpos, 0) rectangle (\xpos + \barwidth, \heightunit*0.01);
                \else
                    % f^j(x mod ℓ) = 0 → mass is (1+β)/d
                    \fill[green!40] (\xpos, 0) rectangle (\xpos + \barwidth, \heightunit*0.65);
                \fi
            }
        }
        
        % T_4 (T_{2i}): detailed pattern showing g^j(x mod ℓ)
        % For each subinterval T_{2i}^j
        \pgfmathsetmacro{\substart}{(3*\k + 4 - 1)*\subintervalwidth}
            
        % Draw subinterval boundary
        \draw[gray!60, thick] (\substart,0) -- (\substart, \heightunit*1.1);
        \foreach \j in {1,2,3} {
            \pgfmathsetmacro{\substart}{(3*\k + \j - 1)*\subintervalwidth}
            
            % Draw subinterval boundary
            \draw[gray!60, thick] (\substart,0) -- (\substart, \heightunit*1.1);
            
            % Pick exactly one position to have mass 2/d, others have (1-β)/d
            \pgfmathsetmacro{\extremepos}{int(mod(\j*2 + 3, \ell))}  % one extreme per subinterval
            
            % Show individual positions within subinterval (ℓ = 8 positions)
            \foreach \pos in {0,1,2,3,4,5,6,7} {
                \pgfmathsetmacro{\xpos}{\substart + \pos*\subintervalwidth/\ell}
                \pgfmathsetmacro{\barwidth}{\subintervalwidth/\ell - 0.02}
                
                \pgfmathsetmacro{\isextreme}{int(\pos == \extremepos ? 1 : 0)}
                
                \ifnum\isextreme=1
                    % g^j(x mod ℓ) = 1 → mass is 2/d (exactly ONE per subinterval)
                    \fill[orange!70] (\xpos, 0) rectangle (\xpos + \barwidth, \heightunit*1.0);
                \else
                    % g^j(x mod ℓ) = 0 → mass is (1-β)/d
                    \fill[red!40] (\xpos, 0) rectangle (\xpos + \barwidth, \heightunit*0.35);
                \fi
            }
        }
        
        % T_5, T_6: uniform 1/d
        \fill[blue!20] (4*\k*\subintervalwidth,0) rectangle (6*\k*\subintervalwidth,\heightunit*0.5);
        
        % Braces with labels
        \draw[decorate,decoration={brace,amplitude=5pt,mirror}] 
            (2*\k*\subintervalwidth,-0.3) -- (3*\k*\subintervalwidth,-0.3) node[midway,below=5pt] {\small $T_{2i-1}$};
        \draw[decorate,decoration={brace,amplitude=5pt,mirror}] 
            (3*\k*\subintervalwidth,-0.3) -- (4*\k*\subintervalwidth,-0.3) node[midway,below=5pt] {\small $T_{2i}$};
            
        % Label one subinterval
        \draw[decorate,decoration={brace,amplitude=3pt}] 
            (2*\k*\subintervalwidth, \heightunit*1.15) -- (2*\k*\subintervalwidth + \subintervalwidth, \heightunit*1.15) 
            node[midway,above=1pt] {\tiny $T_{2i-1}^1$};
        % \draw[decorate,decoration={brace,amplitude=3pt}] 
        %     (3*\k*\subintervalwidth, \heightunit*1.15) -- (3*\k*\subintervalwidth + \subintervalwidth, \heightunit*1.15) 
        %     node[midway,above=3pt] {\tiny $T_{2i}^1$};
            
        % % Show x mod ℓ values for first subinterval of T_{2i-1}
        % \foreach \pos in {0,1,2,3,4,5,6,7} {
        %     \pgfmathsetmacro{\xpos}{2*\k*\subintervalwidth + \pos*\subintervalwidth/\ell + \subintervalwidth/(2*\ell)}
        %     \node[font=\tiny, gray] at (\xpos, -0.7) {$\pos$};
        % }
        % \node[font=\tiny, gray] at (2*\k*\subintervalwidth + \subintervalwidth/2, -0.95) {positions: $x \bmod \ell$};
        
        % % Add annotation arrows
        % \draw[->, thick, blue!70] (2*\k*\subintervalwidth + 0.65, \heightunit*0.72) -- 
        %     (2*\k*\subintervalwidth + 0.45, \heightunit*0.68) 
        %     node[anchor=south, font=\tiny, align=center] at (2*\k*\subintervalwidth + 0.65, \heightunit*0.72) 
        %     {$f^1(x \bmod \ell) = 0$\\mass: $(1+\beta)/d$};
            
        % \draw[->, thick, purple!70] (3*\k*\subintervalwidth + 0.5, \heightunit*1.05) -- 
        %     (3*\k*\subintervalwidth + 0.3, \heightunit*1.0) 
        %     node[anchor=south, font=\tiny, align=center] at (3*\k*\subintervalwidth + 0.5, \heightunit*1.05) 
        %     {$g^1(x \bmod \ell) = 1$\\mass: $2/d$};
    \end{scope}
    
    % Legend and explanation
    \begin{scope}[xshift=19.5cm, yshift=3.7cm]
        \node[anchor=west, font=\small\bfseries] at (0, 2.3) {Legend:};
        
        \fill[green!40] (0, 1.6) rectangle (0.3, 1.9);
        \node[anchor=west, font=\small] at (0.4, 1.75) {$(1+\beta)/|\calX|$};
        
        \fill[red!40] (0, 1.1) rectangle (0.3, 1.4);
        \node[anchor=west, font=\small] at (0.4, 1.25) {$(1-\beta)/|\calX|$};
        
        \fill[blue!20] (0, 0.6) rectangle (0.3, 0.9);
        \node[anchor=west, font=\small] at (0.4, 0.75) {$1/|\calX|$ (uniform)};
        
        \fill[orange!70] (0, 0.1) rectangle (0.3, 0.4);
        \node[anchor=west, font=\small] at (0.4, 0.25) {$2/|\calX|$ (spike)};
        
        \draw[gray!50] (0, -0.4) rectangle (0.3, -0.1);
        \node[anchor=west, font=\small] at (0.4, -0.25) {$0$ (sink)};
        
        % \node[anchor=west, font=\small, align=left, text width=3cm] at (0, -1.1) {
        %     \textbf{$P_i$ structure:}\\
        %     Each $T_{2i-1}^j$ has $\ell$ positions.\\
        %     $f^j(x \bmod \ell)$ determines\\
        %     if mass is $0$ or $\frac{1+\beta}{d}$.
        % };
        
        % \node[anchor=west, font=\small, align=left, text width=3cm] at (0, -2.8) {
        %     Each $T_{2i}^j$ has $\ell$ positions.\\
        %     $g^j(x \bmod \ell)$ determines\\
        %     if mass is $\frac{1-\beta}{d}$ or $\frac{2}{d}$.
        % };
    \end{scope}
    
\end{tikzpicture}

\caption{Visualization of hard instance with hypotheses $H_i$, $H_{i'}$ and distribution $P_i \in \supp(\mathcal{S}_i)$ with $i=2$ and $i'=1$.
The domain $\calX$ is partitioned into intervals $T_1, \ldots, T_{2n}$, each subdivided into $k$ sub-intervals $T_i^j$ of length $\ell$. 
Hypotheses $H_i$ and $H_{i'}$ are mostly uniform but assign slightly higher mass $(1+\beta)/|\calX|$ to $T_{2i-1}$ or $T_{2i'-1}$ and slightly lower mass $(1-\beta)/|\calX|$ to $T_{2i}$ or $T_{2i'}$ respectively.
$P_i$ mostly agrees with $H_i$ but exhibits fine-grained structure within sub-intervals: on each $T_{2i-1}^j$, one element is chosen at random to be a \emph{sink} with 0 mass; on each $T_{2i}^j$, one element is chosen at random to be a \emph{spike} with mass $2/|\calX|$.}
\label{fig:expected-lb}
\end{figure}

The hard instance is displayed in \Cref{fig:expected-lb}.
To construct the indistinguishable families of  distributions, the underlying domain is split into $2n$ pieces $T_1, \ldots, T_{2n}$.
For any $i \in [n]$, the hypothesis $H_i$ and every $P_i \in \supp(\calS_i)$ are uniform over all pieces other then $T_{2i-1}$ and $T_{2i}$.
On elements in $T_{2i-1}$, $H_i$ has probability mass $\frac{1+\beta}{|\calX|}$ slightly more than uniform, and on elements in $T_{2i}$, $H_i$ has probability mass $\frac{1-\beta}{|\calX|}$ slightly less then uniform.

The true distribution $P_i$ is very close to $H_i$ but introduces some elements in $T_{2i-1}$ and $T_{2i}$ with more extreme heavy/light probability masses in order to ``balance out'' the distribution towards the uniform distribution.
We further split the pieces $T_{2i-1}$ and $T_{2i}$ into intervals of length $\ell \approx 1/\beta$.
On each piece in $T_{2i-1}$, all but one element of each interval in $P_i$ is set to $\frac{1+\beta}{|\calX|}$, the same as $H_i$.
The other element is a ``sink'' set to have mass zero, so that the total mass of the interval is $\ell/|\calX|$, the same as it would be under the uniform distribution.
Similarly, for each piece in $T_{2i}$, all but one element of each interval is set to $\frac{1-\beta}{|\calX|}$ with one element set to a ``spike'' at $2/|\calX|$.
The distribution $\calS_i$ is uniform over all such distributions, in other words, sampling $P_i \sim \calS_i$ involves picking a random element in each of these intervals to be the extreme element.

The distance between $P_i$ and $H_i$ essentially comes from the cost of moving the extreme elements to the uniform distribution ($H_i$ does not have exactly $1/|\calX|$ mass on those elements, but this is true up to low-order terms).
The distance between $P_i$ and $H_j$ for some $j \neq i$ is equal to the cost of moving the extreme elements to the uniform distribution plus the cost of moving the slightly more/less than uniform pieces $T_{2i-1}$, $T_{2i}$, $T_{2j-1}$, and $T_{2j}$ to the uniform distribution.
It is a straightforward calculation to verify that the latter distance is roughly three times the former.

It remains to show indistinguishability.
Le Cam's two point method is a classic result relating TV distance and indistinguishability in guessing games between two distributions: for any algorithm, the advantage in correctly guessing beyond $1/2$ probability is bounded by half the TV distance.
This can be shown by taking a coupling view of TV distance: if the TV distance is small, then there is a coupling between the two distributions where they are equal except with probability equal to the distance.
We show a simple lemma generalizing this result to $n$ distributions given a \emph{multi-dimensional coupling} where all $n$ distributions take on the same value simultaneously except with a small probability $\gamma$.
We call the subset of the domain where the distributions all take on the same value a \emph{coincident subset}.

In our setting, the distributions of interest are $\calD_1, \ldots, \calD_n$ where $\calD_i$ is obtained by sampling a random $P_i \sim \calS_i$ and then taking $s$ i.i.d.\ samples from $P_i$.
We design a multi-dimensional coupling with a coincident subset $\Bunique$ which is all possible samples of size $s$ which take at most one sample per length $\ell$ sub-interval.
The key property of our construction is that, for all $i$, $\calD_i$ assigns uniform probability mass $1/|\calX|^s$ to all samples in $\Bunique$ essentially because (a) $P_i$ has mass $\ell/|\calX|$ on every interval and (b) $\calS_i$ randomly permutes the probability masses within each interval.
Finally, we show that if the number of samples $s \ll \sqrt{|\calX|/n\ell}$, then it is highly likely that a sample from any $\calD_i$ will land in $\Bunique$.
Conditioned on this event, it is impossible to determine any information about the index $i$.
Therefore, if the sample size does not grow with $|\calX|$ (e.g., any finite sample size for real-valued distributions), then it is impossible to correctly guess the index with probability significantly more than $1/n$.

The above argument implies that no proper algorithm can obtain expected approximation guarantee better than $3-2/n$. It turns out that the lower bound instance even rules out better approximation guarantees when the algorithm is allowed to output a mixture of $\mathcal{H}$. Essentially, we show that for our lower bound instance, in the mixture setting, any algorithm obtaining an approximation factor significantly better than $3-2/n$, must put $\gg 1/n$ weight on $H_{i^*}$ which again contradicts the indistinguishability property above.
 An interesting observation is that the uniform distribution is a $2$-approximation in our hard instance, however, it cannot be written as a linear combination of the hypotheses.

\subsection{Expected Approximation Upper Bound}\label{sec:tech_ev}

%We develop a randomized algorithm which, given access to $s = O\p{\frac{\log n}{\eps^2}}$ samples, has an output distribution $\calD$ over hypotheses such that:
%\[
%\E[H \sim \calD]{\dtv(H, P)} \leq \p{3-\frac{2}{n}}\OPT + \eps,
%\]
%where the expectation is over the randomness of the algorithm.
%This is essentially the best possible due to the aforementioned lower bound against improper algorithms outputting mixture distributions.
%Given $s = o\p{\frac{\sqrt{|\calX|}}{n^{1.1}}}$ samples, any proper, randomized algorithm must have an output distribution $\calD$ such that:
%\[
%\E[H \sim \calD]{\dtv(H, P)} \geq \p{3 - \frac{2}{n} - o\p{\frac{1}{n}}} \OPT.
%\]

%Previously, these results were only known in the specific case of a $2$-approximation when $n=2$ due to \cite{mahalanabis2007density} and \cite{bousquet2019optimal}.\footnote{As previously noted, the lower bound in \cite{bousquet2019optimal} had a subtle error which we remedy in our result for general $n$, see \Cref{rem:lb-bug}.}
%Solving the problem for general $n$ was posed as an open problem in both of those works.

%\paragraph{Upper Bound}
The lower bound instance of the previous section is easy to handle from an upper bound perspective: Any hypotheses in $H_i$ with $i\neq i^*$ is a $3$ approximation, so outputting a uniformly random hypothesis gives expected approximation factor $\frac{1+3(n-1)}{n}=3-2/n$. We now give a technical overview of \Cref{thm:ev_ub} which provides an algorithm with this approximation factor for an arbitrary set of hypotheses. 
%\todo{Add some high level intuition about how the lb example is easy from an upper bound perspective: just randomly pick. The challenge in the general ub algo is ensuring this approx for an arbitrary input set of hypos. We cant simply pick uniformly and must carefully design a distribution over $\HH$ to sample from}
The output of the algorithm in \Cref{thm:ev_ub} will be an explicit distribution over $\HH$ (i.e., values $\{p_i\}_{i=1}^{n}$ which are non-negative and sum to $1$). The probabilities will be defined in terms of the semi-distances of \Cref{def:semi_distances}.  For simplicity, here we ignore the fact that our semi-distance estimates are noisy and assume we know the values exactly. %In the proof of \Cref{thm:ev_ub}, we incorporate noisy semi-distances satisfying \Cref{lem:sample-scheffe}.

The $n=2$ case is our starting point. In terms of semi-distances, we can view the  algorithm of \cite{mahalanabis2007density} in the following way: Define the distribution over $\HH = \{H_1, H_2\}$ as: pick $H_1$ with probability $\frac{\semi{1}{2}}{\semi{1}{2} + \semi{2}{1}}$ and pick $H_2$ with probability $\frac{\semi{2}{1}}{\semi{1}{2} + \semi{2}{1}}$. Sampling from this distribution gives:

\begin{claim}\label{claim:two_dist_ev}
Let $H$ be the output of the sampling described above. Then $ \E[]{\dtv(P, H)} \leq 2 \cdot \OPT.$
\end{claim}
The proof of the claim is in \Cref{sec:omitted_proofs}. The hope is to generalize this distribution over $\HH$ to larger $n$. Let $p_i$ be the probability that we pick $H_i$. As mentioned above, these probabilities will depend on the semi-distances, and \Cref{prop:semidist-underest} allows us to explicitly compute (an upper bound) on the the expected cost as follows:
\begin{align*}
    \E[]{\dtv(P, H)} \leq \OPT + \sum_{i \neq i^*} p_i \p{\semi{i}{i^*} + \semi{i^*}{i}} \leq \OPT + \sum_{i \neq i^*} p_i \p{W(H_{i^*}) + W(H_i)}.
\end{align*}
Since ultimately we want to bound the approximation ratio rather than the actual TV distance, we divide by $\OPT$ and use that $W(H_{i^*}) \leq \OPT$ (\Cref{prop:semidist-underest}) to upper bound,
\begin{align*}
    \frac{\E[]{\dtv(P, H)}}{\OPT} \leq 1 + \frac{\sum_{i \neq i^*} p_i \p{W(H_{i^*}) + W(H_i)}}{W(H_{i^*})} \leq 1 + \sum_{i \neq i^*} p_i \p{1 + \frac{W(H_i)}{W(H_{i^*})}}.
\end{align*}

Note that we do not know $i^*$ (this is the whole point of hypothesis selection!), so we need to design our probabilities $p_1, \cdots, p_n$ to be able to handle \emph{all} possible cases of $i^*$.  One natural attempt at this is to attempt to pick the $p_i$ such that the above 
upper bound on approximation ratio is the same regardless of $i^*$. With $\textbf{p}=(p_1,\dots p_n)^T$ and $A$ defined to be the $n\times n$ matrix with $A_{ii}=0$ for all $i$ and $A_{ij}=1+\frac{W(H_j)}{W(H_i)}$ for all $i\neq j$, this is equivalent (up to normalization of $\textbf{p}$) to  $A\textbf{p}=\textbf{1}$, where $\textbf{1}$ is the all ones vector.
%for all choices of $i^*$. This yields the following system of linear equations (we can then later normalize the $p_i$ to sum to $1$):
%\begin{equation}\label{eq:linear_system}
%\begin{bmatrix}
   % 0 & 1 + \frac{W(H_2)}{W(H_1)} & \ldots & 1 + \frac{W(H_n)}{W(H_1)} \\
  %  1 + \frac{W(H_1)}{W(H_2)} & 0 & \ldots & 1 + \frac{W(H_n)}{W(H_2)} \\
   % \vdots & \vdots & \vdots & \vdots \\
   % 1 + \frac{W(H_1)}{W(H_n)} & 1 + \frac{W(H_2)}{W(H_n)} & \ldots & 0
%\end{bmatrix}
%\begin{bmatrix}
   % p_1 \\
   % p_2 \\
   % \vdots \\
   %p_n
%\end{bmatrix}
%=
%\begin{bmatrix}
  %  1 \\
   % 1 \\
  %  \vdots \\
  %  1
%\end{bmatrix},
%\end{equation}
%In summary, each constraint above imposes that the approximation ratio upper bound should be the same no matter the choice of $i^*$.
In \Cref{sec:intuition_ev}, we give explicit solutions of this linear system. The solution for $n=3$ directly gives an algorithm with expected approximation $3-\frac{2}3$. However, the solution for the $n=4$ reveals an unexpected surprise: the resulting values $p_i$, while guaranteeing a $3-\frac{2}4$ expected error, can sometimes be \emph{negative}. %Thus, we have to deal with the following challenge.

%\begin{quote}
    %Challenge: How can we ensure that $p_1, \ldots, p_n$ defines a probability distribution?
%\end{quote}
%\begin{enumerate}
    %\item Challenge 1: Can we get a handle on the solution to the system given in \Cref{eq:linear_system} for general $n$?
    %\item Challenge 2: How can we ensure that the resulting values of $p_1, \ldots, p_n$ are non-negative and form a valid probability distribution?
%\end{enumerate}

%Towards the first challenge, we show how to explicitly solve the linear system of \Cref{eq:linear_system} in \Cref{lem:explicit_p} (up to scaling).

To overcome this issue, one natural idea is to impose a non-negativity constraint in addition to $A\textbf{p}=\textbf{1}$. However, this is a linear program which seems difficult to solve analytically. Instead, we take an alternate approach showing how to \emph{round} the $\{p_i\}_{i=1}^{n}$ to a valid probability distribution. The rounding procedure is given in \Cref{sec:formal_ev_ub_proof}; also see \Cref{lem:step2}. Three crucial steps must be checked: (1) That the explicit solution to the linear system $A\textbf{p}=\textbf{1}$ actually gives us a $3-\frac{2}n$ approximation factor, momentarily ignoring the fact that  $\{p_i\}_{i=1}^{n}$ may not form a valid distribution. (2) That the distribution resulting from the rounding maintains the $3-\frac{2}n$ approximation factor. (3) That only knowing noisy approximations to the semi-distances is not a problem. We deal with these steps in \Cref{thm:general_n_ev}, \Cref{lem:step2}, and \Cref{lem:step3}, respectively. Putting these lemmas together completes the proof of \Cref{thm:ev_ub}.

\subsection{Fast Algorithms with Moderate Success Probability}\label{sec:technical_overview_moderate}

The most technical challenge in this work is to obtain an algorithm for hypothesis selection with optimal approximation factor $C=3$ in time $\frac{n\cdot \text{polylog}(n)\log(1/\eps)}{\eps^2\delta}$ and with optimal sample complexity.
Here we sketch the main ideas behind the algorithm. In what follows, $\tilde O$ hides logarithmic factors in $n$. For the discussion, let us forget about $\varepsilon$ and assume that we have oracle access to exact semi-distances $\semi{i}{j}$ for any $i,j$ --- if we can understand this case, our result then follows rather easily from \Cref{lem:sample-scheffe}.
Our main contribution is an algorithm solving what we dub the \emph{Semi-Distance Threshold Problem}. This problem has input parameters $b,\beta \in [0,1]$ as well as the hypotheses set $\mathcal{H}$ and unknown distribution $P$. Consider the directed graph $G_b=([n],E)$ having a directed edge $(i,j)$ if and only if $\semi{i}{j}>b$. The goal in this problem is to output either
\begin{itemize}
\item[(1)] $\perp$ if all nodes in $G_b$ have an incoming edge, or
\item[(2)]  $j$ for any $j\in [n]$ such that there is no edge from $i^*$ to $j$ in $G_b$.
\end{itemize}
If we output $\perp$, we require that $G_b$ indeed satisfies the property in (1) with probability one, whereas if we output some $j$, then the property in (2) holds with probability $1-\beta$ for a parameter $\beta$ which is roughly $\delta$.
We here sketch why such an algorithm suffices to get our result and the main ideas behind designing the algorithm and why it can be made to run in $\tilde O(n/\beta )$ time. The details of this part can be found in \Cref{sec:binary_search}.

\underline{\emph{How to apply an algorithm for the Semi-Distance Threshold Problem}}: Suppose we have an algorithm for the Semi-Distance Threshold Problem. We first query the threshold $b=0$. If the algorithm outputs some $j$, then $\semi{i^*}{j}=0$ and returning $H_j$ already gives a valid approximation. Otherwise, the algorithm outputs $\perp$ at $b=0$. At the other endpoint, $G_1$ is empty and the algorithm successfully outputs some $j$ with probability $1-\beta$. We can therefore perform a binary search on $b$, employing the algorithm $\log(1/\eps)$ times with $\beta=\delta/\log(1/\eps)$, to find a $b_0$ such that with input $b_0$, the algorithm outputs $\perp$ whereas with input $b_0+\eps$ it outputs some $j_0$. These two endpoint cases ensure that the binary search can be performed: Indeed, think of an array of size $O(1/\eps)$ where we want to find a cell containing a $0$ (corresponding to outputting $\perp$) which is immediately to the left of a cell containing a $1$ (corresponding to outputting some $j$). Without any assumptions on the entries, this would require us to look at $\Omega(1/\eps)$ entries, but assuming the leftmost cell contains a zero and the rightmost cell contains a one, we can do it looking at just $O(\log 1/\eps)$ entries. This is achieved by checking the middle cell and recursively updating the search interval to ensure its left boundary is always a zero and its right boundary is a one.

\begin{figure}[!h]
    \centering
    \includegraphics[width=0.7\linewidth]{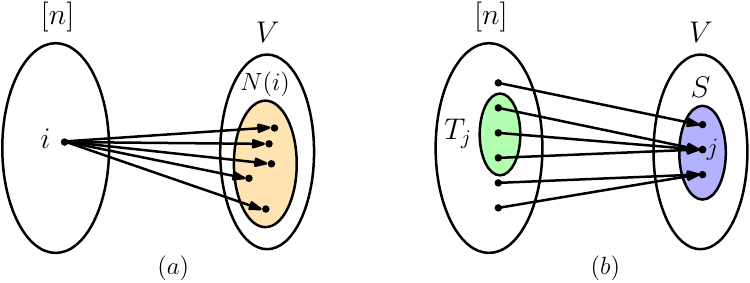}
    \caption{A node $i$ in $[n]$ is $\alpha$-prompting if it has edges to at least a $\alpha$-fraction of the nodes of $V$. (a) In the case where nodes in $[n]$ are $\gg\beta $ prompting on average, a $\beta$-prompting $i$ can be found through a sampling procedure which samples nodes of $V$. (b) When the nodes in $[n]$ are only $O(\beta)$ prompting on average, we can instead sample a set $S$ of size $O(\log n)$ from $V$ (blue), for $j\in S$, identify the set $T_j$ in $[n]$ with an edge to $j$, and test how prompting nodes in $T_j$ are. If we find a $\beta$-prompting hypothesis, we return it. On a high level (the concrete argument is more finicky), if $i^*$ is $\beta$-prompting and has an edge to a node in $S$, we will return some $\beta$-prompting hypothesis (not necessarily $i^*$). On the other hand, if $i^*$ is not $\beta$-prompting, the probability that $i^*$ has an edge to a node of the sampled set $S$ is $O(\beta \log n)$, so in this case, if we don't find a $\beta$-prompting hypothesis, we can instead return an arbitrary node $j$ from $S$ and with probability $1-O(\beta \log n)$, there is no edge from $i^*$ to $j$.}.
    \label{fig:prompting}
\end{figure}

Our algorithm outputs $H_{j_0}$ as the final result. Assuming all runs of the Semi-Distance Threshold Problem algorithm were correct, it holds that (1) $\OPT\geq b_0$ and (2) $\semi{i^*}{j_0}\leq b_0+\eps\leq \OPT+\eps$.
Indeed, since the algorithm with input $b_0$ outputs $\perp$, $i^*$ must have some incoming edge $(j,i^*)$ in $G_{b_0}$ which implies that $b_0\leq \semi{j}{i^*}\leq \OPT$ via \Cref{prop:semidist-underest}. Moreover, since the algorithm on input $b_0+\eps$ outputs $j_0$, $\semi{i^*}{j_0}<b_0+\eps\leq \OPT+\eps$.
By~\Cref{prop:semidist-approx}, returning $j_0$ thus suffices.   

\underline{\emph{An algorithm for the Semi-Distance Threshold Problem}}: 
For $V\subseteq [n]$ and $i\in [n]$, we say that $i$ is \emph{$\alpha$-prompting} (with respect to $V$) if $\{j\in V \mid (i,j)\in E(G_b)\}$ constitutes at least a $\alpha$-fraction of $V$. To solve the Semi-Distance Threshold Problem, we use a subroutine, which for any set $V\subseteq [n]$ can either identify an $i\in [n]$ which is at least $\beta$-prompting or a $j \in V$ satisfying that there is no edge from $i^*$ to $j$ in $G_b$ with probability $1-\tilde O(\beta)$. The ideas behind this subroutine are illustrated in~\Cref{fig:prompting}.

To do so, consider first the case where the average ``promptingness'' $\hatd$ of $i\in [n]$ is $\gg \beta$, that is to say for a random $i\in [n]$, $\dbar:=\mathbb{E}[|\{j\in V \mid (i,j)\in E(G_b)\}|]/|V|=\frac{|E(G_b)|}{n\cdot |V|}\gg \beta$. Sampling $\tilde O(1/\beta)$ random pairs $(i,j)\in U\times V$, and checking if there is an edge, we can detect if this is the case, and if so, estimate $\dbar$ within a constant factor. Now in this case, for some positive integer $k=O(\log n)$, there must be at least $n/(2^k/\log n)$ values $i\in [n]$ that are $\beta 2^k$-prompting. To check if an $i\in [n]$ is $\beta 2^k$-prompting, it suffices to sample $\tilde O(1/(\beta 2^k))$ values $j\in V$ and check if there is an edge $(i,j)$. Since for the "correct" value of $k$, there are $n/(2^k/\log n)$ many $\beta 2^k$-prompting hypothesis, finding one by random sampling from $[n]$ requires $\tilde O(2^k/\log n)$ samples and checking if it is  $\beta2^k$-prompting, takes $\tilde O(1/(\beta 2^k))$ samples from $V$ resulting in a runtime of $\tilde O(1/\beta)$. Guessing the correct value of $k$ has overhead $O(\log n)$. Further details can be found in \Cref{sec:ave-deg-prompting}. 

On the other hand, in the case $\dbar =O(\beta)$, we proceed differently. In this case, by Markov's inequality at least a $9/10$-fraction of $i\in [n]$ are $O(\beta)$-prompting and by a double counting argument, at least a $9/10$-fraction of $j\in V$ have $|\{i\in [n]\mid (i,j)\in E(G_b)\}|/n=O(\beta)$. In this case, we first sample $O(\log n)$ elements $j\in V$, let $T_j=\{i\in [n]\mid (i,j)\in E(G_b)\}$, and add $j$ into a set $S$ if $|T_j| \leq O(\beta n)$. By a Chernoff bound, this happens with high probability for at least a $(4/5)$-fraction of the sampled $j$'s. The process takes $\tilde O(n)$ time. Now for each $j\in S$, we check if \emph{any} element in $T_j$ is $\beta$-prompting. Checking if a single element $i\in T_j$ is prompting can be done in $\tilde O(1/\beta)$ time, so this check takes $\tilde O(n)$ times in total due to the bounded size of $T_j$. If we find a $\beta$-prompting $i$, our subroutine has succeeded. 

Finally, in the case where no element in any $T_j$ is $\beta$-prompting, we repeat the above sampling process \emph{again}, constructing $S'\subset V$, but this time \emph{resisting the temptation} to ever return a $\beta$-prompting hypothesis even if we find one. Instead, we commit to returning a $v_j$ which has no edge from $i^*$ with probability $\geq 1-\beta$. To do so, we search for a hypothesis that has low degree $O(\beta n)$, and \emph{no} $\beta$-prompting neighbors. Since the previous sampling step failed to produce a $\beta$-prompting hypothesis even though we sampled $\gg \log n$ values $j$, we know with very high probability that at least roughly half the element of $V$ have low degree $O(\beta n)$ and no $\beta$-prompting neighbors. Since we also sample $\gg\log n$ values $j\in V$ in the second sample, with very high probability we find such a $j\in V$, and when we do, we return it as the final output.
The point is that if $i^*$ is $\gg \beta$-prompting then $(i^*,j)$ cannot be an edge since, $j$ has no $\beta$ prompting neighbors. Thus in this case, the success probability is $1$.  On the other hand, if $i^*$ is only $O(\beta)$ prompting, then the probability that the randomly sampled set of size $O(\log n)$ will have any edge to $i^*$ is $O(\log n\cdot \beta)=\tilde O(\beta)$. In particular, with probability $1-\tilde O(\beta)$, the $j$ returned by the algorithm has no edge from $i^*$. We provide a more detailed discussion in \Cref{sec:small-neighborhood}.

 To apply the subroutine, start with $V=[n]$ and search for a $\beta$-prompting $i$. Then remove all $j$ such that $(i,j)\in E(G_b)$ from $V$ and repeat. Since we always remove a $\beta$ fraction, we can at most perform this step $\tilde O(1/\beta)$ times, so the total runtime becomes $\tilde O(n/\beta)$. From the analysis above, we know that even over the many runs of the subroutine, if we \emph{ever} return a $j\in [n]$, then with probability $1-O(\beta)$, there is no edge $(i^*,j)$ and our Semi-Distance Threshold algorithm succeeded. The point is that unless some very low probability error event happened, then the first time the algorithm commits to outputting a $j$, it will in fact be the case with probability $1-\tilde O(\beta)$, it holds that $j$ has no edge from $i^*$. In particular, and crucially, we do not have to union bound over all $\tilde O (1/\beta)$ searches for a $\beta$-prompting hypothesis.
 
 The other case is that eventually $V=\emptyset$, and then the algorithm returns $\perp$. If this happens, then since every element from the initial $V=[n]$ were evicted, with probability $1$, for any $j\in [n]$ there exists $i\in [n]$ such that $(i,j)$ is an edge. In other words, each node of $G_b$ has an incoming edge and the output is correct.
A more comprehensive version of this argument is presented in \Cref{sec:moderate_delta_putting_together}.

In \Cref{sec:quantile}, we describe a different, much simpler algorithm which improves upon prior work~\cite{aliakbarpour2024optimal} but achieves an overall worse dependence of $\tilde{O}\p{\frac{n}{\eps^2\delta^2}}$.

\subsection{Near-Linear Time with Known $\OPT$}\label{sec:tech_overview_known_OPT}
In the case where some upper bound $R\geq\OPT$ is revealed to us as part of the input, we show that it is possible to further improve the runtime of the algorithm. In particular, when $R=\OPT$, we get the optimal approximation factor $C=3$ while avoiding the polynomial dependence on $\delta$ in the runtime. The algorithm is a variant of an idea from~\cite{aliakbarpour2024optimal} which, without the assumption that $\OPT$ is known, obtains $C=4$ in a similar runtime. Let us for the discussion again assume that we have oracle to the exact semi-distances $\semi{i}{j}$ and for simplicity that $R=\OPT$.

Consider the directed graph $G=(V,E)$ on $V=[n]$ that has a directed edge $(i,j)$ if and only if $\semi{i}{j}>\OPT$. From the  discussion in the preliminaries section, it suffices to find a $j$ such that there is no edge going from $i^*$ to $j$ and return $H_j$. Also from that section, we know the additional piece of information that $i^*$ has no incoming edges. We first search for a $j_1$ such that $j_1$ has less than $n/2$ incoming edges. This can be done by iterating over all $j\in [n]$ and for each such $j$ sampling $O(\log n)$ nodes $i\in [n]$ and checking whether $i$ has an edge to $j$. Note that this search always succeeds since $i^*$ has the desired property. Inductively, assuming we have found a $j_k$ with at most $n/2^k$ incoming edges, we define $S_k=\{i\mid (i,j_k)\in G\}$ which has size $\leq n/2^k$ and search for an element $j_{k+1}$ in $S_k$ that has less than $n/2^{k+1}$ incoming edges. To test whether a single $j\in S_k$ has this property, we must sample $O(2^{k+1}\log n)$ random $i\in [n]$ and check if they have an edge to $j$. However, the fact that $S_k$ only has size $n/2^k$ ensures that the total runtime of this step still becomes $O(n\log n)$. Now two things can happen. Either this process continues for $O(\log n)$ steps culminating in $S_k=\emptyset$ eventually. But in this case $j_k$ has no incoming edges and we are happy to return $H_{j_k}$. Alternatively, we eventually do not succeed in finding $j_{k+1}$, but in this case $i^*$ cannot lie in $S_k$; if it did, the search would find it if nothing else. By the definition of $S_k$, this means that there is no edge from $i^*$ to $j_k$ and so we can return $H_{j_k}$. The total runtime for the $O(\log n)$ steps is $O(n\log^2 n)$ assuming oracle access to $\semi{i}{j}$ in constant time, so the final runtime becomes $O(n \log^3 n/ \eps^2)$.

\subsection{Subquadratic Time with Preprocessing} \label{sec:tech_overview_preprocessing}
The fastest known algorithm for hypothesis selection achieving $1/\text{poly}(n)$ failure probability and the optimal approximation factor of $3$ is due to \cite{mahalanabis2007density} and runs in quadratic time $O(n^2 \log(n)/\eps^2)$. This runtime is a natural barrier for this algorithm, and in general, for algorithms which compute high quality estimates of all quadratically many semi-distances $\semi{i}{j}$, which can be done up to accuracy $\pm \eps$ with probability $1-1/\text{poly}(n)$ using $O(\log(n)/\eps^2)$ samples, and then postprocess these semi-distances to output a hypothesis close to $P$. 

While our (linear in $n$) upper bound of Section \ref{sec:moderateprob} obtains an improved dependence on the failure probability $\delta$, our algorithm does not yield improvements over this prior quadratic-time algorithm in the setting where $\delta$ is an arbitrarily small inverse polynomial in $n$. This motives us to ask if a polynomial time \emph{preprocessing} step can help bridge this gap. Indeed, \cite{mahalanabis2007density} explicitly raised this question (See Question $6$ in their work).

We make partial progress by giving an algorithm which uses polynomial preprocessing, uses the optimal number of samples, and with high-probability outputs a $3$ approximation (plus $\eps$ additive error) in \emph{subquadratic} time $\tilde{O}(n^{2-\Omega(\eps)})$. This shows that (if preprocessing is allowed), we can go beyond computing all semi-distances to obtain an algorithm with a high-probability of success. We briefly note that the problem is not interesting with exponential time preprocessing, since we can enumerate all possible $O(\log(n)/\eps^2)$ sized samples that can be observed and store the optimal answer for those samples.\footnote{\cite{mahalanabis2007density} give another exponential time preprocessing algorithm with linear-time queries where the preprocessing time does not depend on the underlying domain size.} This would even allow for \emph{sublinear} query time. 

It turns out that the exponential preprocessing algorithm of \cite{mahalanabis2007density} can be abstracted to the following problem which is sufficient to solve hypothesis selection.

\begin{restatable}{definition}{tournament}\label{def:tournament_problem}
    We are given a weighted, undirected complete graph on n vertices. Assume that the edge-weights are distinct. We preprocess the weighted graph and then play the following game with an adversary until only one vertex remains: we report the edge with the largest weight and the adversary chooses one of the endpoints of the edge and removes it from the graph (together with all the adjacent edges).
\end{restatable}

The hope is to use polynomial time pre-processing to simulate the tournament revelation problem. However, as stated, the problem is far too general. But in the case of hypothesis selection, we can use the additional available \emph{geometric} structure: the known hypothesis $H_i$ can be viewed as vectors in $\R^{|\calX|}$, where recall $\calX$ is the domain. This additional geometric structure allows us to simulate the tournament revelation problem in overall subquadratic time.  

In more detail, the geometric viewpoint means it suffices to quickly find the diameter of a dynamic point set under arbitrary deletions. To obtain an overall subquadratic time bound, each operation must operate in sublinear in $n$ time. This is not possible to do exactly with known techniques in computational geometry, but it is possible to obtain $1+\eps$ approximate diameter queries for any small $\eps > 0$ in $\approx n^{1-\eps}$ time, ignoring dependency on the dimension (see \Cref{cor:preprocessing} for the formal bound). Such a datastructure is shown to exist by composing two existing results in computational geometry (\Cref{thm:indyk} of \cite{indyk2003better} and \Cref{thm:eppstein} of \cite{eppstein1995dynamic}), which allows us to simulate the tournament revelation problem, and hence solve hypothesis selection in subquadratic time.
% Note the "preprocessing power" is used to create the dynamic diameter data structure of \Cref{cor:preprocessing}.
Finally, we note that we can remove the linear dependence on $|\calX|$ during the actual query phase (and only pay $|\calX|$ dependence during preprocessing) by performing dimension reduction mapping $\ell_1$ into a lower dimensional $\ell_2^2$ space; see \Cref{cor:dim_reduction}.

\subsection{Other Related Work}
% \paragraph{Other frameworks of hypothesis selection: }
% Hypothesis selection has also been studied in the \textit{improper setting}, where the algorithm can output a density that does not necessarily belong to $\mathcal{H}$. In this setting, an approximation factor of $C=2$ is known to be achievable and tight~\cite{bousquet2019optimal, bousquet2022statistically}.
% This problem has been also studied under the constraints of differential privacy~\cite{bun2019private, gopi2020locally, PourAA24, Aliakbarpour2025}. We note that all of our algorithms are proper and return a hypothesis in $\HH$.

% \paragraph{Broader Applications of Hypothesis Selection}
While our paper is focused on the fundamental hypothesis selection problem itself, we briefly remark that progress on our questions can potentially have many downstream applications, since hypothesis selection is used as a standard ``meta tool'' in many learning problems. For example, a classic approach to learning a structured distribution is to choose a net or cover for the class (a finite set of representatives that are `close by' to all possible distributions in the class). The cover is viewed as the set of hypothesis $H_i$. Examples of where this strategy has been successful 
include Poisson binomial distributions \cite{daskalakis2012learning}, mixtures of Gaussians \cite{kalai2012disentangling, daskalakis2014faster, suresh2014near, kothari2018robust, ashtiani2018sample, ashtiani2020near}, distributions with piecewise polynomial density functions \cite{AcharyaD0S17}, and
histograms \cite{ChanDSS14, canonne2022nearly}; we refer to the survey \cite{diakonikolas2016learning} for more details.

The significance of the optimal constant factor, $C=3$, also becomes more apparent when using the standard cover methods. If we use a $\gamma$-cover of a parametric distribution class, running a hypothesis selection algorithm on this cover implies that $\OPT \approx \gamma$. To achieve the same final accuracy, an algorithm with a suboptimal approximation factor requires a much denser cover (a smaller $\gamma$) than one with the optimal factor. In some cases, such as for mixtures of $k$-Gaussians, this decrease in $\gamma$ can cause the size of the cover to grow exponentially. Therefore, using an algorithm with the best possible approximation factor is crucial for maintaining computational feasibility. A detailed discussion is provided in~\cite{aliakbarpour2024optimal}.

%\maryam{also significance of constant improvement in the approximation factors}

\section{Approximation Lower Bounds}\label{sec:expected-lb}

In this section, we show that the following information-theoretic lower bounds on the approximation factor achieved by any (randomized) algorithm which takes a number of samples which does not grow polynomially with the underlying domain size:\footnote{Domain-independent sample complexity is a \emph{key statistical phenomenon} of hypothesis selection.}
\begin{enumerate}
    \item A proper algorithm cannot return a $C<3$ approximation with probability greater than $\frac{1}{n}\p{1 + o(1)}$ (\Cref{cor:proper-lb}).
    \item An (improper) algorithm which returns a hypothesis in the convex hull of $\calH$ must achieve an expected approximation $C > 3-\frac{2}{n} - o\p{\frac{1}{n}}$ (\Cref{cor:exp-lower-bound}).
\end{enumerate}
We note that the second result is in a strictly more general setting than the proper, expected approximation regime considered in \Cref{sec:expected}.
This implies that our upper bound achieves the tight approximation factor and that returning a mixture of the hypotheses provides no benefit over proper algorithms.
We additionally show by a standard result in distribution learning that $\Omega(\log(n)/\eps^2)$ samples are necessary for an algorithm achieving \emph{any constant} multiplicative factor expected value guarantee in Theorem \ref{thm:C-lower-bound}. 

It was previously claimed in \cite{bousquet2019optimal} that no proper algorithm can return a $C < 3$ approximation with probability greater than $1/2$, similarly to our first result stated above.
In addition to strengthening this result, we identify and fix a subtle error in the prior work.
Our construction still takes its starting point from that of \cite{bousquet2019optimal}.
We discuss this further at the end of this section in \Cref{rem:lb-bug}.

To achieve the claimed lower bounds, we prove a slightly stronger statement: there exist a set of hypotheses and a family of true distributions such that (a) all but one hypothesis $H_{i^*}$ incurs an approximation factor which is nearly $3$ and (b) for a large enough domain size and a uniformly random true distribution from the family, it is impossible guess the optimal index $i^*$ from samples with probability better than $1/n$.

\paragraph{Hard Instance}
The hard instance is given by a set of hypotheses $H_1, \ldots, H_n$ and distribution families $\calS_1, \ldots \calS_n$ defined as follows and depicted in \Cref{fig:expected-lb}.

Let $\calX = [2n k \ell]$ for some $k, \ell \in \N$.
Let $\beta = \frac{1}{\ell-1}$.
Consider the $2nk$ length $\ell$ intervals $T_1^1, \ldots T_1^k$, $\ldots$, $T_{2n}^1$, $\ldots$, $T_{2n}^k$ where $T_i^j= [(i-1)k\ell+(j-1)\ell+1, (i-1)k\ell+j\ell]$.
Let $T_i = \bigcup_{j=1}^k T_i^j = [(i-1)k\ell + 1, ik\ell]$.

The $i$th hypothesis will be uniform everywhere but on the interval spanned by $T_{2i-1} \cup T_{2i}$, where it will have mass slightly more than $1/|\calX|$ and slightly less than $1/|\calX|$ on the first and second half, respectively.
We will define the $i$th hypothesis as follows:
\begin{equation}\label{eq:exp-lb-hypotheses}
    \Pr[Z \sim H_i]{Z = x} =
    \begin{cases}
        \frac{1 + \beta}{|\calX|} &\qquad \text{if } x \in T_{2i-1} \\
        \frac{1 - \beta}{|\calX|} &\qquad \text{if } x \in T_{2i} \\
        \frac{1}{|\calX|} &\qquad \text{otherwise}
    \end{cases}.
\end{equation}
It is straightforward to verify that $H_i$ is a valid distribution over $[|\calX|]$ as long as $\beta \leq 1$. Let $\calF$ be a the uniform distribution over one-hot indicator functions $f: \{0,\ldots, \ell-1\} \rightarrow \{0,1\}$ (functions with support size $1$).
We will now describe the distribution-over-distributions $\calS_i$ for any $i \in [n]$.
A distribution sampled from $\calS_i$ is created by sampling $f^1, \ldots, f^k, g^1, \ldots, g^k \iid \calF$. Let $F = \p{f^1, \ldots, f^k}$ and $G = \p{g^1, \ldots, g^k}$.
The probability mass function $P_{(i, F, G)}$ is defined as follows:
\begin{equation}\label{eq:exp-lb-true-distbn}
    \Pr[Z \sim P_{(i,F,G)}]{Z = x} = 
    \begin{cases}
        0 &\qquad \text{ if } x \in T_{2i-1}^j \text{ and } f^j\p{x \bmod \ell} = 1 \\
        \frac{1+\beta}{|\calX|} &\qquad \text{ if } x \in T_{2i-1}^j \text{ and } f^j\p{x \bmod \ell} = 0 \\
        \frac{2}{|\calX|} &\qquad \text{ if } x \in T_{2i}^j \text{ and } g^j\p{x \bmod \ell} = 1 \\
        \frac{1-\beta}{|\calX|} &\qquad \text{ if } x \in T_{2i}^j \text{ and } g^j\p{x \bmod \ell} = 0 \\
        \frac{1}{|\calX|} &\qquad \text{otherwise}
    \end{cases}.
\end{equation}
We will sometimes drop the full subscript and simply use $P_i$ to refer to a distribution drawn from $\calS_i$.

For any $i, F, G$, it is straightforward to verify that $P_{(i,F,G)}$ is a valid probability distribution over $[|\calX|]$.
In fact, the distribution $P_{(i,F,G)}$ induces over the intervals $T_{i'}^j$ is the uniform distribution as shown below.
For any $i', j$,
\begin{equation}\label{eq:exp-lb-uniform-blocks}
    \Pr[Z \sim P_{(i,F,G)}]{Z \in T_{i'}^j}
    = \sum_{x \in T_{i'}^j} \Pr[Z]{Z = x}
    = \frac{\ell}{|\calX|}.
\end{equation}
This is trivial for the intervals where $P_{(i,F,G)}$ is uniform.
If $i' = 2i-1$, then the total mass in the interval is
\begin{equation*}
    \p{\ell-1}\p{\frac{1+\beta}{|\calX|}}
    = \p{\ell-1}\p{\frac{1+\frac{1}{\ell-1}}{|\calX|}}
    = \frac{\ell}{|\calX|}.
\end{equation*}
If $i' = 2i$, then the total mass is in the interval is:
\begin{equation*}
    \frac{2}{|\calX|} + \p{\ell-1}\p{\frac{1-\beta}{|\calX|}}
    = \frac{2}{|\calX|} + \p{\ell-1}\p{\frac{1-\frac{1}{\ell-1}}{|\calX|}}
    = \frac{2}{|\calX|} + \frac{\ell-2}{|\calX|}
    = \frac{\ell}{|\calX|}.
\end{equation*}

\begin{theorem}\label{thm:exp-lower-bound}
Fix any $\alpha$ and consider a domain of size $|\calX| \gg n/\alpha$ and hard instance defined above.
There exists constants $d_0, d_1 \in [0,1]$ such that for every $i \neq j \in [n]$ and distribution $P_i \in \supp(\calS_i)$:
\[d_0 = \dtv\p{P_i, H_i}, d_1 = \dtv\p{P_i, H_j}, \text{ and } \frac{d_1}{d_0} \geq 3-\alpha.\]
Furthermore, given fewer than $o\p{\sqrt{\frac{|\calX|}{n/\alpha}}}$ samples from $P_{i^*} \sim \calS_{i^*}$ for $i^* \sim \Unif([n])$, no randomized algorithm can correctly guess the index $i^*$ with probability more than $\frac{1}{n} \p{1 + o(1)}$.
\end{theorem}

Given this theorem, the approximation lower bounds follow as corollaries.

\begin{corollary}\label{cor:proper-lb}
Let $\alpha > 0$ and consider the hard instance with $|\calX| \gg n/\alpha$.
Given $o\p{\sqrt{\frac{|\calX|}{n/\alpha}}}$ samples, no proper (randomized) algorithm can return a hypothesis $H_i$ with approximation factor less than $3 - \alpha$ with probability more than $\frac{1}{n}\p{1 + o(1)}$.
\end{corollary}

\begin{proof}
By \Cref{thm:exp-lower-bound}, no algorithm with the given number of samples can guess the optimal index $i^*$ with probability more than $\frac{1}{n}\p{1 + o(1)}$.
Furthermore, the only hypothesis with approximation less than $3-\alpha$ is $H_{i^*}$.
\end{proof}

\begin{corollary}\label{cor:exp-lower-bound}
Consider the hard instance with $|\calX| \gg n^{2.2}$.
Given $o\p{\frac{\sqrt{|\calX|}}{n^{1.1}}}$ samples, no (randomized) algorithm can output a distribution $Q$ in the convex hull of $\calH$ with expected approximation ratio better than $3 - \frac{2}{n} - o\p{\frac{1}{n}}$ where the expectation is taken over the randomness of the algorithm, the instance, and the samples.
\end{corollary}

\begin{proof}
For the sake of contradiction, assume that there exists such an algorithm $\calA$ which uses $o\p{\frac{\sqrt{|\calX|}}{n^{1.1}}}$ samples and has expected approximation factor at most $3 - \frac{2}{n} - \frac{C}{n}$ for some constant $C_0 > 0$.
Consider the hard setting of \Cref{thm:exp-lower-bound} with $\alpha = 1/n^{1.2}$.
The theorem implies that no algorithm with $o\p{\frac{\sqrt{|\calX|}}{n^{1.1}}}$ samples can guess the index of the true distribution $i^*$ with probability greater than $\frac{1}{n}\p{1+o(1)}$.

As $Q$ is a convex combination of the hypotheses $\calH$, it can be written as $Q = \sum_{i=1}^n w_i H_i$ with (random variable) weights $(w_1, \ldots, w_n)$ taking on values in the $n$-dimensional simplex.
Note that $\dtv(P_{i^*}, Q) = \frac{1}{2}\norm{P_{i^*}- Q}_1 = \frac{1}{2} \sum_{x=1}^{|\calX|} \abs{P_{i^*}(x) -  Q(x)}$.
For $i \in [n]$, let
$$r_i(Q) = \sum_{x \in T_{2i-1} \cup T_{2i}} \abs{P_{i^*}(x) -  Q(x)}$$
be the component of the $\ell_1$ distance on the intervals corresponding to $i$.

For all $i \neq i^*$, $r_i(Q) = w_i r_i(H_i)$. This is because for all $j \neq i$, $H_j(x) = P_{i^*}(x) = 1/|\calX|$ for all $x \in T_{2i-1} \cup T_{2i}$: $H_j$ and $P_{i^*}$ are equal on these intervals.
For this reason, letting $d_0, d_1$ from \Cref{thm:exp-lower-bound}: $r_{i^*}(H_{i^*}) = 2d_0$ and for any $i \neq i^*$, $r_{i^*}(H_i) + r_i(H_i) = 2 d_1$.

Consider $r_{i^*}(Q)$ and any $j \neq i^*$.
Note that $H_j(x) = 1/|\calX|$ for all $x \in T_{2i^*-1} \cup T_{2i^*}$.
Then,
\begin{align*}
    r_{i^*}(Q)
    &= \sum_{x \in T_{2i^*-1} \cup T_{2i^*}} \abs{\sum_{i=1}^n w_i H_i(x) - P_{i^*}(x)} \\ 
    &= \sum_{x \in T_{2i^*-1} \cup T_{2i^*}}
    \abs{(1-w_{i^*}) H_j(x) + w_{i^*} H_{i^*}(x) - P_{i^*}(x)}. \\
    &= \sum_{x \in T_{2i^*-1}}
    \abs{\frac{1-w_{i^*} + w_{i^*}(1+\beta)}{|\calX|}  - P_{i^*}(x)}
    + \sum_{x \in T_{2i^*}}
    \abs{\frac{1-w_{i^*} + w_{i^*}(1-\beta)}{|\calX|}  - P_{i^*}(x)} \\
    &= \frac{2k}{|\mathcal{X}|} \p{(\ell-1) (1-w_{i^*}) \beta + 1 + w_{i^*}\beta} \\
    &= C_1 + w_{i^*}\beta \p{-\ell + 2}.
\end{align*}
Here, let $C_1$ be a constant independent of $w_{i^*}$.
Therefore, $r_{i^*}(Q)$ is a linear function of $w_{i^*}$.
When $w_{i^*} = 0$, $r_{i^*}(Q) = r_{i^*}(H_j)$ and when $w_{i^*}=1$, $r_{i^*}(Q) = r_{i^*}(H_{i^*}) = 2d_0$.
By linearity, $r_{i^*}(Q) = 2 w_{i^*} d_0 + (1-w_{i^*})r_{i^*}(H_j)$.

Overall,
\begin{align*}
    \|P_{i^*} - Q\|_1 
    &= \sum_{i=1}^n r_i(Q) \\
    &= \p{\sum_{i \neq i^*} w_i r_i(H_i)} + (1-w_{i^*})r_{i^*}(H_j) + 2 w_{i^*} d_0 \\
    &= \p{\sum_{i \neq i^*} w_i r_i(H_i) + w_i r_{i^*}(H_i)}  + 2w_{i^*}d_0 \\
    &= \p{\sum_{i \neq i^*} 2 w_i d_1} +  2w_{i^*}d_0 \\
    &= 2(1 - w_{i^*}) d_1 + 2w_{i^*} d_0.
\end{align*}

Then, $\dtv(P_{i^*}, Q) = (1-w_{i^*}) d_1 + w_{i^*} d_0$.
Using the fact that $d_1/d_0 \geq 3-\alpha$,
\[
    \frac{\dtv(P_{i^*}, Q)}{\dtv(P_{i^*}, H_{i^*})}
    \geq \frac{(1 - w_{i^*}) d_1 + w_{i^*} d_0}{d_0}
    \geq \p{1-w_{i^*}}\p{3 - \alpha} + w_{i^*}
    = 3 - \alpha - 2 w_{i^*} + \alpha w_{i^*}.
\]

By the initial assumption, the algorithm achieves an expected approximation factor of $3-\frac{2}{n}-\frac{C_0}{n}$, then $\E{w_{i^*}} > \frac{1+C_0/2}{n}$. Therefore, the output distribution of this algorithm can be used to guess $i^*$ with probability more than $\frac{1}{n}\p{1 + o(1)}$, reaching a contradiction with the indistinguishability guarantee of \Cref{thm:exp-lower-bound}.
\end{proof}

To prove the main theorem, we will make use of the following lemma about indistinguishability and multi-dimensional couplings.
This is an extension of the classic relation between TV distance, couplings, and indistinguishability to $n$ rather than $2$ distributions.

\begin{definition}[Multi-dimensional Couplings and Coincident Subsets]
Let $\calD_1, \ldots, \calD_n$ be distributions over the same  domain $\calX$.
A distribution $\calQ$ over $\calX^n$ is a \emph{multi-dimensional coupling} of $\calD_1, \ldots, \calD_n$ if, for all $i \in [n]$ and for all $x \in \calX$,
\begin{equation*}
    \Pr[A \sim \calQ]{A_i = x} = \Pr[X \sim \calD_i]{X = x}.
\end{equation*}
Given a multi-dimensional coupling $\calQ$, a subset $S \subseteq \calX$ is a \emph{coincident subset} if, for all $x \in S$ and for all $i \in [n]$,
\begin{equation}
    \Pr[A \sim Q]{A_i = x} = \Pr[A \sim Q]{\bigcup_{j=1}^n A_j = x}.
\end{equation}
\end{definition}

\begin{lemma}[Multi-dimensional Coupling Indistinguishability]\label{lem:multi-coupling}
Let $\calD_1, \ldots, \calD_n$ be distributions over $\calX$, let $\calQ$ be a multi-dimensional coupling of these distributions, and let $S \subseteq \calX$ be a coincident subset.
Let $\gamma = \Pr[A \sim Q]{\bigcup_{i=1}^n A_i \notin S}$ be the total mass of elements not in $S$.
Consider any (randomized) algorithm $\calA(x;r)$ which receives a sample from a uniformly random distribution among $\calD_1, \ldots, \calD_n$ and guesses the index of the distribution.
Then, the probability of $\calA$ guessing correctly, denoted as $\phi(\calA)$, is at most:
\begin{equation*}
    \phi(\calA)
    = \Pr[i^* \sim \Unif\p{[n]}, X \sim \calD_{i^*}, r]{\calA(X;r) = i^*} \leq \gamma + \frac{1-\gamma}{n}.
\end{equation*}
\end{lemma}

\begin{proof}
By the definition of a multi-dimensional coupling and coincident subset, for any $i \in [n]$, $\Pr[X \sim \calD_i]{X \in S} = 1-\gamma$.
$\phi(\calA)$ can be expressed as the sum of conditional probabilities depending on whether or not $X \sim \calD_{i^*}$ is in $S$. 
Then, to prove the desired bound of $\phi(\calA) \leq \gamma + \frac{1-\gamma}{n}$, it suffices to show that the correctness probability is at most $1/n$ conditioned on $X \in S$.

Let $\calE$ be the conditional distribution where, for all $i \in [n]$ and $x \in S$, $\Pr[Y \sim \calE]{Y = x} = \Pr[X \sim \calD_i]{X = x | X \in S}$.
The existence of such a distribution which is independent of the choice of $i$ is directly implied by the definition of a coincident subset.
Then, 
\begin{align*}
    \Pr[i^* \sim \Unif\p{[n]}, X \sim \calD_i, r]{\calA(X;r) = i^* | X \in S}
    &= \Pr[i^* \sim \Unif\p{[n]}, Y \sim \calE, r]{\calA(Y;r) = i^*} \\
    &= \sum_{i=1}^n \Pr[i^* \sim \Unif\p{[n]}, Y \sim \calE, r]{\p{\calA(Y;r) = i} \cap \p{i^* = i}} \\
    &= \sum_{i=1}^n \Pr[Y \sim \calE, r]{\calA(Y;r) = i} \Pr[i^* \sim \Unif\p{[n]}]{i^* = i} \\
    &= \sum_{i=1}^n \Pr[Y \sim \calE, r]{\calA(Y;r) = i}\p{\frac{1}{n}} \\
    &= \frac{1}{n}.
\end{align*}
\end{proof}

We are now ready to prove the main theorem.
\begin{proof}[Proof of \Cref{thm:exp-lower-bound}]
Consider any $i \in [n]$ and any valid choices of indicator functions $F, G$.
We will calculate the distances between a true distribution $P_{(i,F,G)}$ and the hypotheses.

\paragraph{Distance to $H_i$}
By \Cref{eq:exp-lb-true-distbn}, it suffices to only consider the elements in $T_{2i-1} \cup T_{2i}$ as $P_{(i,F,G)}(x) = H_i(x)$ on all elements $x$ outside of this interval.
The distance is:
\begin{align*}
    \dtv\p{P_{(i,F,G)}, H_i}
    &= \frac{1}{2} \norm{P_{(i,F,G)} - H_i}_1 \\
    &= \frac{1}{2|\calX|}\p{
        2k\p{1+\beta} + 2k(\ell-1) \p{0}
    } \\
    &= \frac{k\p{1+\beta}}{|\calX|}.
\end{align*}

\paragraph{Distance to $H_j$}
Consider any $j \neq i$.
It suffices to only consider the elements in $T_{2i-1} \cup T_{2i} \cup T_{2j-1} \cup T_{2j}$ as $P_{(i,F,G)}(x) = H_i(x)$ on all elements $x$ outside of this interval.
The distance is:
\begin{align*}
    \dtv\p{P_{(i,F,G)}, H_j}
    &= \frac{1}{2} \norm{P_{(i,F,G)} - H_j}_1 \\
    &= \frac{1}{2|\calX|}\p{
        2k + 2k(\ell-1)\p{\beta} + 2k\ell\p{\beta}
    } \\
    &= \frac{k(1 + 2\ell\beta - \beta)}{|\calX|} \\
    &= \frac{k\p{1 + 2\p{\frac{1}{\beta} + 1}\beta - \beta}}{|\calX|} \\
    &= \frac{k(3 + \beta)}{|\calX|}.
\end{align*}

This ratio of distances $\frac{\dtv\p{P_{(i,F,G)}, H_j}}{\dtv\p{P_{(i,F,G)}, H_i}} = \frac{3+\beta}{1+\beta}$ is lower bounded by $3-\alpha$ if $\beta = O(\alpha)$, so it suffices to choose $\ell = \Theta(1/\alpha)$.

\paragraph{Indistinguishability}
It remains to show that, given insufficient samples from a randomly chosen distribution $P_{(i,F,G)}$, it is impossible to guess $i$ with probability much greater than $1/n$.
Let $\calD_i^s$ be the probability distribution over sets of $s \ll |\calX|/\ell$ samples constructed by first sampling $f^1, \ldots, f^k, g^1, \ldots, g^k \sim F$ and then sampling $s$ elements from the corresponding distribution $P_{(i,F,G)}$.
Note that $\calD_i^s$ has the following probability mass function for any sample $\bx \in [|\calX|]^s$:
\begin{equation}\label{eq:exp-lb-marginal}
    \Pr[B \sim \calD_i^s]{B = \bx}
    = \E[F, G]{\Pr[Z_1, \ldots, Z_s \iid P_{(i,F,G)}]{\bigcup_{j=1}^s Z_j=\bx_j}}.
\end{equation}

We will construct a multi-dimensional coupling $\calQ$ over the domain $[|\calX|]^{n \times s}$ where the $i$th row of a sample from $\calQ$ is distributed as $\calD_i^s$.
We will design $\calQ$ to have the following coincident subset.
Let $\Bunique \subset [|\calX|]^s$ be the set of all vectors which contain at most one element from any interval $T_u^v$ for all $u, v$.
Observe that all distributions $\calD_i^s, \ldots, \calD_n^s$ are equivalent up to a reordering of the intervals $T_u^v$, so the probability of sampling an element of $\Bunique$ is the same for all $\calD_1^s, \ldots, \calD_n^s$.
Let $\gamma$ be the probability that a sample of size $s$ contains at least two samples within the same interval.
For a matrix $\bM$, let $\bM_{i*}$ be its $i$th row vector.
We are now ready to define $\calQ$:
% \anders{Also, is it simpler to say the following: We sample $s$ buckets. If they are distinct, we let all $\calD_i^s$ be the same, and be generated by sampling a uniformly random element within each sampled bucket. If they are not distinct, we further generate the $f^j,g^j$ and sample the appropriate number within each bucket?}
\begin{equation*}
    \Pr[A \sim \calQ]{A = \bM} = 
    \begin{cases}
        \frac{1}{|\calX|^s} &\qquad \text{if } \bM_{1*} = \ldots = \bM_{n*} \in \Bunique \\
        \frac{1}{\gamma^{(n-1)}} \prod_{i=1}^n \Pr[B_i \sim \calD_i^s]{B_i = \bM_{i*}} &\qquad \text{if } \bM_{i*} \notin \Bunique \text{ for all } i \in [n] \\
        0 &\qquad \text{otherwise}
    \end{cases}.
\end{equation*}

It remains to show that for $A \sim \calQ$, the row-marginal distributions are distributed as $A_{i*} \sim \calD_i^s$.
Consider any $i \in [n]$ and any $\bx \in [|\calX|]^s$.
We will proceed by comparing $\Pr[A \sim \calQ]{A_{i*} = \bx}$ and $\Pr[B \sim \calD_i^s]{B = \bx}$ by cases.

\textbf{Case 1:} $\bx \notin \Bunique$.
The row-marginal distribution is:
\begin{align*}
    \Pr[A \sim \calQ]{A_{i*} = \bx}
    &= \sum_{\bM \in [|\calX|]^{n \times s}: \bM_{i*} = \bx} \Pr[A \sim \calQ]{A = \bM} \\
    &= \sum_{\bM \in [|\calX|]^{n \times s}: \bM_{i*}  = \bx \text{ and } \bM_{j*} \notin \Bunique \text{ for all } j \in [n]} \Pr[A \sim \calQ]{A = \bM} \tag{All other terms are zero}\\
    &= \sum_{\bw_1, \ldots, \bw_n \notin \Bunique: \bw_i = \bx} \gamma^{-(n-1)} \prod_{j=1}^n \Pr[B_j \sim \calD_j]{B_j = \bw_j}  \\
    &= \gamma^{-(n-1)} \Pr[B_i \sim \calD_i^s]{B_i = \bx} \sum_{\bw_1, \ldots, \bw_{i-1}, \bw_{i+1}, \ldots, \bw_n \notin \Bunique} \prod_{j \neq i} \Pr[B_j \sim \calD_j]{B_j = \bw_j}  \\
    &= \gamma^{-(n-1)} \Pr[B_i \sim \calD_i^s]{B_i = \bx} \cdot \\
    & \qquad \underbrace{\p{\sum_{\bw_1 \notin \Bunique}  \Pr[B_1 \sim \calD_1^s]{B_1 = \bw_1}} \ldots \p{\sum_{\bw_n \notin \Bunique}  \Pr[B_n \sim \calD_n^s]{B_n = \bw_n}}}_{n-1} \\
    &= \gamma^{-(n-1)} \Pr[B_i \sim \calD_i^s]{B_i = \bx} \cdot \\
    & \qquad \underbrace{\Pr[B_1 \sim \calD_1^s]{B_1 \notin \Bunique} \ldots \Pr[B_n \sim \calD_n^s]{B_n \notin \Bunique}}_{n-1} \\
    &= \gamma^{-(n-1)} \Pr[B_i \sim \calD_i^s]{B_i = \bx} \gamma^{n-1} \\
    &= \Pr[B_i \sim \calD_i^s]{B_i = \bx}.
\end{align*}

\textbf{Case 2:} $\bx \in \Bunique$.
The row-marginal distribution of $\calQ$ is:
\begin{align*}
    \Pr[A \sim \calQ]{A_{i*} = \bx}
    &= \sum_{\bM \in [|\calX|]^{n \times s}: \bM_{i*} = \bx} \Pr[A \sim \calQ]{A = \bM} \\
    &= \Pr[A \sim \calQ]{A = 
    \begin{bmatrix}
        \bx \\
        \ldots \\
        \bx
    \end{bmatrix}
    } \tag{All other terms are zero}\\
    &= \frac{1}{|\calX|^s}.
\end{align*}
% $A_{i*}$ is uniform over all elements in $\Bunique$.
% Note that $\Pr[B \sim \calD_i^s]{B \in \Bunique} = 1 - \gamma$ and that $\abs{\Bunique} = \frac{d!\ell^s}{(d-m)!}$.
% It remains to show that, for any $\bx \in \Bunique$, $\Pr[B \sim \calD_i^s]{B = \bx}$ is a constant function of $\bx$.
Recall from \Cref{eq:exp-lb-marginal} that the probability mass function of $B$ is an expectation over the following terms:
\begin{equation}\label{eq:exp-lb-unique-sample-target}
    \Pr[B \sim \calD_i^s]{B = \bx}
    = \E[F, G]{\prod_{j=1}^s \Pr[Z \sim P_{(i,F,G)}]{Z = \bx_j}}.
\end{equation}

Let $\bw$ be an $s$ vector of index tuples such that $\bw_j = \p{u_j, v_j}$ is the index of the interval containing $\bx_j$: $\bx_j \in T_{u_j}^{v_j}$.
Note that as $\bx \in \Bunique$, all entries of $\bw$ are unique.
For a given element $v \in T_u^v$, we will consider the following equivalent hierarchical sampling process where we first sample a interval and an element within the interval:
\begin{equation*}
    \Pr[Z \sim P_{(i,F,G)}]{Z = v}
    = \Pr[Z]{Z \in T_u^v} \cdot \Pr[Z]{Z = v | Z \in T_u^v}.
\end{equation*}

By \Cref{eq:exp-lb-uniform-blocks}, $\Pr[Z \sim P_{(i,F,G)}]{Z \in T_u^v} = \ell/|\calX|$ does not depend on the distribution parameters $i, F, G$.
The conditional probability $\Pr[Z \sim P_{(i,F,G)}]{Z = v | Z \in T_u^v}$ depends on $i$, $f^t$, and $g^t$ (if $t=2i-1$ or $t=2i$) but not on any of the other coordinates of $F, G$.
% Note that for any probability mass function $p$ over the domain $\crb{0, \ldots, \ell-1}$, the expected value of the inner product of $p$ and a random indicator function $f \sim \calF$ is constant:
% \begin{equation*}
%     \E[f \sim \calF]{\sum_{x \in \{0, \ldots, \ell-1\}} p(x)f(x)}
%     = \sum_{x \in \{0, \ldots, \ell-1\}} \frac{p(x)}{\ell}
%     = \frac{1}{\ell}.
% \end{equation*}
As $F, G$ are product distributions with each coordinate chosen independently as a uniformly random indicator function, we can rewrite \Cref{eq:exp-lb-unique-sample-target} as:
\begin{align*}
    \Pr[B \sim \calD_i^s]{B = \bx}
    &= \E[F, G]{\prod_{j=1}^s  \Pr[Z \sim P_{(i,F,G)}]{Z \in T_{u_j}^{v_j}} \Pr[Z]{Z = \bx_j | Z \in T_{u_j}^{v_j}}} \\
    &= \p{\frac{\ell}{|\calX|}}^s \E[F, G]{\prod_{j=1}^s \Pr[Z]{Z = \bx_j | Z \in T_{u_j}^{v_j}}} \\
    &= \p{\frac{\ell}{|\calX|}}^s \prod_{j=1}^s \E[f^{v_j}, g^{v_j} \iid \calF]{\Pr[Z]{Z = \bx_j | Z \in T_{u_j}^{v_j}}} \\
    &= \p{\frac{\ell}{|\calX|}}^s \prod_{j=1}^s 
    \begin{cases}
        \E[f \sim \calF]{f(\bx_j) \cdot 0 + \p{1-f(\bx_j)}\p{\frac{1 + \beta}{\ell}}} & \quad \text{if } v_j = 2i-1 \\
        \E[g \sim \calF]{g(\bx_j) \cdot 2 + \p{1-g(\bx_j)}\p{\frac{1 - \beta}{\ell}}} & \quad \text{if } v_j = 2i \\
        \frac{1}{\ell} & \quad \text{o.w.}
    \end{cases} \\
    &= \p{\frac{\ell}{|\calX|}}^s \frac{1}{\ell^s} \\
    &= \frac{1}{|\calX|^s}.
\end{align*}
Therefore, $\calQ$ is a valid multi-dimensional coupling of $\calD_1^s, \ldots, \calD_n^s$ with coincident subset $\Bunique$.

By \Cref{lem:multi-coupling}, no algorithm can correctly guess the index of the true distribution with probability greater than $\gamma + \frac{1-\gamma}{n}$.
Recall that $\gamma$ is the probability that a sample of $s$ elements contains at least two samples within the same interval $T_u^v$.
There are $|\calX|/\ell$ total intervals, each of size $\ell$.
We can upper bound this probability as:
\begin{align*}
    \gamma
    &= 1 - \frac{\abs{\Bunique}}{|\calX|^s} \\
    &= 1 - \frac{\p{\frac{|\calX|}{\ell}} \ell \p{\frac{|\calX|}{\ell}-1} \ell \ldots \p{\frac{|\calX|}{\ell}-(s-1)} \ell}{|\calX|^s} \\
    &= 1 - \p{1} \p{1-\frac{\ell}{|\calX|}} \p{1-\frac{2\ell}{|\calX|}} \ldots \p{1-\frac{(s-1)\ell}{|\calX|}} \\
    &\leq 1 - \p{1 - \frac{s\ell}{|\calX|}}^s \\
    &\leq 1 - \p{1 - \frac{s^2\ell}{|\calX|}} \tag{Bernoulli's inequality} \\
    &= \frac{s^2 \ell}{|\calX|}.
\end{align*}

Recall that $|\calX| = 2nk\ell$. If $k = \omega(s^2)$, then $\gamma + \frac{1-\gamma}{n} \leq \frac{1}{n}\p{1 + o(1)}$, completing the proof.
\end{proof}

\begin{remark}\label{rem:lb-bug}
This theorem is a generalization of the lower bound in \cite{bousquet2019optimal} for $n=2$.
We now present the issue with that analysis and how our result avoids it.

Consider our construction also with $n=2$.
The key difference is that in the prior construction, the extreme elements for each distribution $P_i$ are chosen uniformly at random within the meta-interval $T_i$, and there are no sub-intervals $T_i^j$.
In our construction, there is a single extreme element chosen within each sub-interval.
The analysis in \cite{bousquet2019optimal} argues that, the distribution $\calD_i^s$ of $s$ samples from a randomly chosen distribution $P_i$ conditioned on sampling no duplicate elements is the same as the distribution of $s$ samples from the uniform distribution $U$ on $[|\calX|]$, also conditioned on sampling no duplicate elements.
We argue that, for our construction, this holds among samples which contain no duplicate sub-intervals (no two sampled elements belong to the same sub-interval).

While both claims are true when $s=1$, for larger samples, the former claim is untrue.
Note that in their setting, every distribution $P_i$ in the hard family of distributions has a fixed number $z$ of domain element where its probability mass is zero. In particular, within all but some single meta-interval $T_k$, the distribution will have nonzero mass on every domain element. On $T_k$, there will be a fixed nonzero number of domain elements which have probability mass zero. Consider the extreme case where the sample size is $s=|\calX|-z$. 
Consider the distribution of counts of elements in each meta-interval from a sample from $\calD_i^s$ conditioned on no collisions.
This count distribution will be supported only on a single element: namely, the count vector which has full count on every meta-interval other than $T_k$.
On the other hand, the count distribution of the uniform distribution conditioned on no collisions is supported on many different count vectors as $s < |\calX|$.
Therefore, the distributions $\calD_i^s$ and a uniform sample of $s$ elements, both conditioned on no collisions, are not equivalent.

As an alternative counter-argument, one can consider the case where $s=2$ and directly compute the probabilities conditional on no collisions of both samples lying in $T_{2i-1}$ (where some extreme elements have probability $0$) as opposed to lying in $T_{2i}$ (where some extreme elements have probability $2/|\calX|$), the latter of these probabilities being slightly larger due to sampling without replacement. By symmetry, this is not the case for the uniform distribution. 

%consider the extreme case where the sample size is $s = d$.
%There is a sample $\bx$ which contains no duplicates, namely consider $\bx = (1, 2, \ldots, d)$ which happens with non-zero probability when the samples are drawn from $U$.

%However, under $P_i$, this sample has \emph{zero} probability of occurring as $P_i$ assigns zero probability mass to some elements in $[d]$, and so conditioning on the event of no collisions is not even well-defined. 
%Therefore $P_i$ clearly differs from the uniform distribution on $\bx$. \anders{Is this last sentence really what we want to say? }

The error occurs in the first two sentences of Claim 26 in \cite{bousquet2019optimal} which extends the claim from $s=1$ to general $s$.

The introduction of sub-intervals containing a single extreme element in our construction alleviates this problem by introducing independence.
Each sub-interval has probability exactly $\ell/|\calX|$, so conditioning on no collisions, the collection of sub-intervals that we sample from is a uniformly random size $s$ subset of all sub-intervals. Moreover, the random choices of the functions $f^j,g^j$ for $j=1,\dots,k$ ensure that elements sampled within the intervals are independent and uniformly random within these intervals. Sampling from the uniform distribution, conditioned on no duplicate sub-intervals, can be simulated by first sampling $s$ sub-intervals without replacement, and then independently sampling a uniformly random element within each of these: this exact same probabilistic process generates a sample from $D_i^s$ conditioned on no sub-interval collisions.

We have contacted the authors, and they have corrected the issue in an updated manuscript.
%The fact that no collisions occur somewhere else in the domain does not affect the probability that a collision occurs within a given sub-interval as each sub-interval has exactly $\ell/d$ probability of having an element within it sampled under any $P_i$.
\end{remark}

\section{Expected Approximation Upper Bound}\label{sec:expected}

%We prove our main expected value approximation result of \Cref{thm:general_n_ev} (see \Cref{sec:expected-ub}), showing that we can output a mixture over $\HH$ with approximation factor $C = 3-\frac{2}n$, or alternatively, we can output a distribution over $\HH$ such that the expected approximation factor is $3-\frac{2}n$. We also prove a matching lower bound in \Cref{sec:expected-lb}. 

%Previously, these results were only known in the specific case of a $2$-approximation when $n=2$ (\cite{mahalanabis2007density} and \cite{bousquet2019optimal}), and the general $n$ was posed as an open problem in both of those works.

%\subsection{Upper Bound}\label{sec:expected-ub}

The main theorem of this section is the following. 

\thmEVUB*

Recall from~\Cref{sec:tech_ev} that in order to prove the theorem, we are interested in solving the following system of linear equations (we can then later normalize the $p_i$ to sum to $1$).

\begin{equation}\label{eq:linear_system}
\begin{bmatrix}
    0 & 1 + \frac{W(H_2)}{W(H_1)} & \ldots & 1 + \frac{W(H_n)}{W(H_1)} \\
    1 + \frac{W(H_1)}{W(H_2)} & 0 & \ldots & 1 + \frac{W(H_n)}{W(H_2)} \\
    \vdots & \vdots & \vdots & \vdots \\
    1 + \frac{W(H_1)}{W(H_n)} & 1 + \frac{W(H_2)}{W(H_n)} & \ldots & 0
\end{bmatrix}
\begin{bmatrix}
    p_1 \\
    p_2 \\
    \vdots \\
   p_n
\end{bmatrix}
=
\begin{bmatrix}
    1 \\
    1 \\
    \vdots \\
    1
\end{bmatrix},
\end{equation}
Denote by $A$ the $n\times n$ matrix in the above equation. Each of the constraints impose that the approximation ratio upper bound should be the same no matter the choice of $i^*$. For now, we assume that we know the (max) semi-distances and hence $A$ exactly, but we will get rid of that assumption with~\Cref{lem:step3} later. The below lemma gives the solution to the above linear system.

\begin{restatable}{lemma}{lemevcalc}\label{lem:explicit_p}
Define the constants 
\[ C = \sum_{i = 1}^n \frac{1}{W(H_i)}, \quad D = \sum_{i=1}^n W(H_i),\]
and let 
 \begin{equation}\label{eq:probs}
    p_i = \frac{D/W(H_i) - (n-2)}{DC - n(n-2)}
\end{equation} for all $1\le i \le n$. Let $p \in \R^n$ be the vector with coordinates $p_i$.  We have $\sum_i p_i = 1$ and 
\[Ap = \left( \frac{DC - (n-2)^2}{DC - n(n-2)} \right) \cdot \mathbf{1}, \]
where $\mathbf{1}$ is the vector of all ones and 
\[\frac{DC - (n-2)^2}{DC - n(n-2)}  \le 2 - \frac{2}n. \]
\end{restatable}
We prove the above lemma in \Cref{sec:explicitp}. We briefly note that this lemma mostly serves to give intuition of \emph{where} our final probability values are coming from. We could have alternatively proved \Cref{thm:general_n_ev} (which bounds the expected approximation guarantee with this choice of $p_i$) by directly introducing the probabilities $p_i$ from \Cref{lem:explicit_p} and explicitly calculating the expected value under this distribution. Note that for showing our upper bound, we do not require the solution for the vector $p$ to be unique.

Before we prove \Cref{thm:ev_ub}, we build some intuition, starting from the $n=3$ and $n=4$ case of the linear system \Cref{eq:linear_system}. However, the following intuition can be skipped. %and the formal proof is given in \Cref{sec:formal_ev_ub_proof}.

\subsection{Gaining Intuition via Small $n$ Examples}\label{sec:intuition_ev}

\paragraph{Solution for $n=3$}
 For $n=3$, we can solve \Cref{eq:linear_system} explicitly, which gives the following solution for $p_i$. The solution is symmetric, so we will display the results for $H_1$:
\begin{equation}
    p_1 = \frac{W(H_2) W(H_3) \p{W(H_2) + W(H_3)}}{\sum_{i \neq j} W(H_i)^2 W(H_j)}.
\end{equation}
Then, assuming $i^* = 3$ (which is an arbitrary choice since all possible $i^*$ values yield the same approximation factor by design),
\begin{align*}
    p_1 \p{1 + \frac{W(H_1)}{W(H_3)}}
    &= \frac{W(H_2) W(H_3) \p{W(H_2) + W(H_3)}}{\sum_{i \neq j} W(H_i)^2 W(H_j)} \p{1 + \frac{W(H_1)}{W(H_3)}} \\
    &= \frac{W(H_2) \p{W(H_1) + W(H_3)} \p{W(H_2) + W(H_3)}}{\sum_{i \neq j} W(H_i)^2 W(H_j)},
\end{align*}
and thus the upper bound on the approximation ratio will be
\begin{align*}
    \frac{\E[]{\dtv(P, H)}}{\OPT} 
    &\leq 1 + \frac{\p{W(H_1) + W(H_2)} \p{W(H_1) + W(H_3)} \p{W(H_2) + W(H_3)}}{\sum_{i \neq j} W(H_i)^2 W(H_j)} \\
    &= 1 + \frac{2W(H_1)W(H_2)W(H_3) + \sum_{i \neq j} W(H_i)^2 W(H_j)}{\sum_{i \neq j} W(H_i)^2 W(H_j)}.
\end{align*}

By Muirhead's inequality,
\begin{equation*}
    6W(H_1)W(H_2)W(H_3) \leq \sum_{i \neq j} W(H_i)^2 W(H_j).
\end{equation*}
Therefore,
\begin{equation*}
    \frac{\E[\br]{\dtv(P, H)}}{\OPT} \leq 1 + \frac{4}{3}.
\end{equation*}

\paragraph{The problem with $n=4$.} Now solving the $n=4$ case of the linear system gives  us the following probabilities:
\begin{equation*}
    p_1 = \frac{W(H_2)W(H_3)W(H_4) \p{-W(H_1) + W(H_2) + W(H_3) + W(H_4)}}{ - 4W(H_1)W(H_2)W(H_3)W(H_4) + \sum_{\text{distinct }i, j, k} W(H_i)^2 W(H_j) W(H_k)},
\end{equation*}
and again in a similar manner we can calculate that the upper bound on the approximation ratio is $ \le 1 + \frac{3}2$, which concurs with our the desired $3-\frac{2}n$ bound.  However, a serious problem is that if $W(H_1) > W(H_2) + W(H_3) + W(H_4)$, the $p_i$'s do not form a probability distribution as they are negative.

Thus, to summarize, two challenges remain to generalize the approach to general $n$:
\begin{enumerate}
    \item Can we get a handle on the solution to the system given in \Cref{eq:linear_system} for general $n$?
    \item How can we ensure that the resulting values of $p_1, \ldots, p_n$ are non-negative and form a valid probability distribution?
\end{enumerate}

\paragraph{Generalizing to Arbitrary $n$}
We handle challenge (1) first. We present an explicit solution to the system given in \Cref{eq:linear_system}, which will be an explicit formula for the probabilities $p_1, \cdots, p_n$ in terms of the max-semi distances of \Cref{def:semi_distances}. Initially, some $p_i$ can be negative and later we will deal with challenge (2) to round these values to be a proper probability distribution. Again, we remark that this section can be safely skipped as it is just for intuition, and the full proof details are given in \Cref{sec:formal_ev_ub_proof}.

\subsection{Proof of~\Cref{lem:explicit_p}}\label{sec:explicitp}
Let $A$ be the matrix of \Cref{eq:linear_system} defined as $A_{ij} = 1 + \frac{W(H_j)}{W(H_i)}$ for $i \ne j$ and $0$ for $i = j$. Our goal is to find probabilities $p = (p_1, \ldots, p_n)$ such that $Ap$ is a constant vector with $\sum_i p_i = 1$. Note that without loss of generality, we can solve for $Ap = \mathbf{1}$ and then normalize the $p_i$'s. 

We now prove~\Cref{lem:explicit_p} providing a solution to the system \Cref{eq:linear_system}, up to scaling. As stated in \Cref{sec:tech_ev}, this lemma mainly serves to give intuition about how that probabilities, values used sample a hypothesis from, are derived. If desired, the proof can be skipped entirely and the reader can instead directly proceed to \Cref{sec:formal_ev_ub_proof}, where we prove \Cref{thm:ev_ub}. The only part of \Cref{lem:explicit_p} that is used is the form of $p_i$ from \eqref{eq:probs}.

%\lemevcalc*

\begin{proof}[Proof of~\Cref{lem:explicit_p}]
Let $\mathbf{W}$ to be the vector
\[ \mathbf{W} = (1/W(H_1), \ldots, 1/W(H_n)).\]
Note that $\sum_{i=1}^n p_i = 1$ by construction so we just have to check that $p_i$ solves the system of \Cref{eq:linear_system}. We have
\[ (A \mathbf{1})_i = \sum_{j \ne i} \left( 1 + \frac{W(H_j)}{W(H_i)} \right)= (n-1) + \frac{1}{W(H_i)} \left( D - W(H_i)\right) = n-2 + \frac{D}{W(H_i)}, \]
so 
\begin{equation}\label{eq:relation1}
    A \mathbf{1} = (n-2)\mathbf{1} + D \mathbf{W}.
\end{equation}
We also have
\[ (A\mathbf{W})_i = \sum_{j \ne i}\left( 1 + \frac{W(H_j)}{W(H_i)} \right) \cdot \frac{1}{W(H_j)} = \sum_{j \ne i} \left( \frac{1}{W(H_i)} + \frac{1}{W(H_j)} \right) = \frac{n-2}{W(H_i) } + C,\]
meaning
\begin{equation}\label{eq:relation2}
    A \mathbf{W} = (n-2)\mathbf{W} + C \mathbf{1}.
\end{equation}
Combining equations \eqref{eq:relation1} and \eqref{eq:relation2} gives us
\[ A(D \mathbf{W} - (n-2) \mathbf{1}) = ( DC - (n-2)^2) \mathbf{1}. \]
Since $DC \ge n^2 > (n-2)^2$ by Cauchy-Schwarz, we know that $DC - (n-2)^2 > 0$ and so the vector $p = \frac{D \mathbf{W} - (n-2)}{DC - (n-2)^2}$ indeed satisfies
\[Ap = \left( \frac{DC - (n-2)^2}{DC - n(n-2)} \right) \cdot \mathbf{1}.\]
Finally, we can check that \[ \frac{DC - (n-2)^2}{DC - n(n-2)} = 1 + \frac{2(n-2)}{DC - n(n-2)} \le 1 + \frac{n-2}n,\]
where the last inequality follows from 
\[ DC \ge n^2 \implies DC - n(n-2) \ge 2n \implies \frac{2(n-2)}{DC - n(n-2)} \le \frac{n-2}n,\]
as desired.
\end{proof}

% Since 
% \[  \sum_{i=1}^n \left(\frac{D}{W(H_i)} - (n-2) \right)= DC - n(n-2),\]
% as desired.

Note that it is entirely possible for some of the $p_i$ to be negative, as in the $n=4$ case. 

\subsection{Proof of the Main \Cref{thm:ev_ub}}\label{sec:formal_ev_ub_proof}
We are now ready to prove our main theorem, \Cref{thm:ev_ub}. There are three steps to the proof. 
\begin{itemize}
    \item Step 1: We show that $\{p_i\}_{i=1}^n$, ignoring the fact that it is not a valid distribution, is able to give a $3-\frac{2}n$ approximation factor.
    \item Step 2: We show how to round $\{p_i\}_{i=1}^n$ to a a valid distribution $\{q_i\}_{i=1}^n$ which maintains the same approximation factor.
    \item Step 3: Lastly, we show how to we can form $\{q_i\}_{i=1}^n$ using the approximated max-semi distances instead, introducing an additive $+\eps$ factor in the bound of \Cref{thm:ev_ub}.
\end{itemize}

We begin with the first step.

\begin{lemma}[Step 1]\label{thm:general_n_ev}
    Define $p_i$ as in Equation \ref{eq:probs}. We have
    \[ 1 + \sum_{i \neq i^*} p_i \p{1 + \frac{W(H_i)}{W(H_{i^*})}} \le 3 - \frac{2}n.\]
\end{lemma}
\begin{proof}
Without loss of generality, suppose $i^* = n$. Then,
\[ p_i \left( 1 + \frac{W(H_i)}{W(H_n)} \right) = \frac{\left(W(H_i) + W(H_n)\right) \left( \frac{D}{W(H_i)} - (n-2) \right)}{W(H_n) \left( DC - n(n-2) \right)}, \]
which gives
\[
\sum_{i \ne n} p_i \left( 1 + \frac{W(H_i)}{W(H_n)} \right)  = \frac{1}{W(H_n) \left( DC - n(n-2) \right)} \cdot \sum_{i \ne n} \left(W(H_i) + W(H_n)\right) \left( \frac{D}{W(H_i)} - (n-2) \right).\]
Let's calculate the numerator. We have
\begin{align*}
    &\sum_{i \ne n}  \left(W(H_i) + W(H_n)\right) \left( \frac{D}{W(H_i)} - (n-2) \right) \\
    &= D(n-1) + D W(H_n) \left (C - \frac{1}{W(H_n)} \right) - (n-2)(D-W(H_n)) - (n-2)(n-1)W(H_n) \\
    &=  W(H_n)(DC - (n-2)^2).
\end{align*}
Plugging back in, the sum we want has value 
\[ \frac{DC - (n-2)^2}{DC - n(n-2)} = 1 + \frac{2(n-2)}{DC - n(n-2)} \le 1 + \frac{n-2}n,\]
where the last inequality follows from 
\[ DC \ge n^2 \implies DC - n(n-2) \ge 2n \implies \frac{2(n-2)}{DC - n(n-2)} \le \frac{n-2}n,\]
as in the proof of \Cref{lem:explicit_p}.
Thus,
\[ 1 + \sum_{i \ne n}  p_i \left( 1 + \frac{W(H_i)}{W(H_n)} \right) \le 3 - \frac{2}n,\]
as desired.
\end{proof}

Again, the main issue with the setting of the probabilities given in \Cref{eq:probs} is that some $p_i$ values can be negative. We show how to find a set of non-negative probabilities that also give an approximation ratio at most $3-2/n$. First, let's recall key formulas:

\[ p_i(n) = \frac{D/W(H_i) - (n-2)}{DC - n(n-2)}, \quad C = \sum_{i = 1}^n \frac{1}{W(H_i)}, \quad D = \sum_{i=1}^n W(H_i). \]

Suppose without loss of generality that 
\[ W(H_n) \ge W(H_{n-1}) \ge \ldots \ge W(H_1).\]
We call an index $k \in [n]$ good if it satisfies 
\[ \forall j \in [k], \quad (k-3)W(H_j) \le \sum_{i \in [k] \setminus \{j\}} W(H_i), \]
and otherwise, bad. Note that
\begin{itemize}
    \item $k \le 3$ is always good.
    \item If $k$ is bad then $k'$ is also bad for all $k' > k$. This is true since $k$ being bad means there exists a $j \in [k]$ such that 
    \begin{align*}
        (k-3)W(H_j) > \sum_{i \in [k] \setminus \{j\}} W(H_i) 
        &\implies (k-2)W(H_j) > \sum_{i =1}^k W(H_i) \\
        &\implies (k-2)W(H_{k+1}) > \sum_{i = 1}^k W(H_i).
    \end{align*}
\end{itemize}
Now given the values $p_1, \ldots, p_n$ of Equation \ref{eq:probs}, we consider the rounding procedure procedure given in \Cref{fig:rounding}.
\begin{figure}
\begin{mdframed}
\begin{enumerate}
    \item If $n$ is good then all the probabilities are non-negative so simply set $q_i = p_i(n)$ for all $1 \le i \le n$.
    \item Otherwise, consider the largest good $k$. Let $q_1, \ldots, q_k$ be the values $p_1(k), \ldots, p_k(k)$ (probabilities for the $k$ case) and set $q_{k+1} = \ldots = q_n = 0$. Note that we only use the formula for the $k$ case but plug in the fixed $W(H_i)$ values from the $n$ case. 
\end{enumerate}
\end{mdframed}
\caption{The rounding procedure takes values $\{p_i\}_{i=1}^n$ and outputs a valid distribution $\{q_i\}_{i=1}^n$. \label{fig:rounding}}
\end{figure}
Since $k$ is good, we know all $q_i \ge 0$. The goal is to show that our desired approximation ratio under $q$, which is a valid probability distribution, is also at most $3-2/n$.

\begin{lemma}[Step 2]\label{lem:step2}
     Define $q_i$ as in \Cref{fig:rounding}. We have
    \[ 1 + \sum_{i \neq i^*} q_i \p{1 + \frac{W(H_i)}{W(H_{i^*})}} \le 3 - \frac{2}n.\]
\end{lemma}

\begin{proof}

    We consider two cases. If the optimum hypothesis is among the first $k$, then our distribution already achieves approximation ratio $3-2/k \le 3-2/n$ from Theorem \ref{thm:general_n_ev}. So suppose the optimum index is some $i^* > k$. We wish to bound
    \[1+ \sum_{i = 1}^k q_i \left( 1 + \frac{W(H_i)}{W(H_{i^*})} \right). \]

Since $k+1$ is bad, we know that there exists a $j \in [k+1]$ satisfying
\[ (k-2)W(H_j) > \sum_{i \in [k+1] \setminus \{j\}} W(H_i) \implies (k-1)W(H_{k+1}) \ge (k-1) W(H_j) > \sum_{i=1}^{k+1} W(H_i)  \]
and thus
\[(k - 2) W(H_{i^*}) \ge (k-2) W(H_{k+1}) > \sum_{i=1}^k W(H_i) \implies W(H_{i^*}) > \frac{1}{k-2} \sum_{i=1}^k W(H_i). \]
Define $D' = \sum_{i=1}^k W(H_i)$ and similarly $C' = \sum_{i=1}^k 1/W(H_i)$. We have

\begin{align*}
 \sum_{i = 1}^k q_i \left( 1 + \frac{W(H_i)}{W(H_{i^*})} \right) &\le \sum_{i=1}^k q_i \left( 1 + \frac{(k-2) W(H_i)}{\sum_{i=1}^k W(H_i)} \right)  \\
 &= \sum_{i=1}^k \left( \frac{\frac{D'}{W(H_i)}  - (k-2)}{D'C' - k(k-2)} \right) \left( 1 + \frac{(k-2) W(H_i)}{D'} \right)  \\
 &= \frac{1}{D'(D'C' - k(k-2))}\sum_{i=1}^k \left(\frac{D'}{W(H_i)} -(k-2) \right)(D' + (k-2)W(H_i)) \\
 &= \frac{1}{D'(D'C' - k(k-2))} \cdot  \left( (D')^2C'  - (k-2)^2 D'\right) \\
 &= \frac{D'C' - (k-2)^2}{D'C' - k(k-2)} \\
 &=  1 + \frac{2(k-2)}{D'C' - k(k-2)}\\
 &\le 1 + \frac{k-2}k.
\end{align*}
where the last inequality follows as in Theorem \ref{thm:general_n_ev} since $D'C' \ge k^2$. The lemma follows.
\end{proof}

Finally, the third step shows that it suffices to compute the distribution $\{q_i\}_{i=1}^n$ using noisy semi-distances. Towards that end,  consider the values $\{p_i\}_{i=1}^n$ defined in \Cref{eq:probs} and let $\{\tilde{p}_i\}_{i=1}^n$ be the corresponding values where we use the same formula, but instead use approximated max-semi distances satisfying $\forall i, W(H_i) \le \hat{W}(H_i) \le W(H_i) + \eps$. Now run the same rounding procedure as in \Cref{fig:rounding} on $\{\tilde{p}_i\}_{i=1}^n$ to produce $\{\tilde{q}_i\}_{i=1}^n$. Note that this is a valid operation since the rounding procedure described in \Cref{fig:rounding} only depends on the inputted approximate semi-distances. So we can still define a good index $k$ as done before, but with respect to the approximate max semi-distances, and round $\{\tilde{p}_i\}_{i=1}^n$ to a valid distribution $\{\tilde{q}_i\}_{i=1}^n$. The same proof as in \Cref{lem:step2} readily implies that 
\begin{equation}\label{eq:approx_bound}
     1 + \sum_{i \neq i^*} \tilde{q}_i \p{1 + \frac{\hat{W}(H_i)}{\hat{W}(H_{i^*})}} \le 3 - \frac{2}n,
\end{equation}
where note that we are using the approximate $\hat{W}(H_i)$ values.

The challenge of the next lemma is showing that $\tilde{q}_i$ are also a `valid' set of hypothesis to use for the \emph{true} semi-distances (which determines the actual expected value approximation factor). Thus, the equation in \Cref{lem:step3} involves the actual max-semi distance values, and the "noisy" distribution $\{\tilde{q}_i\}_{i=1}^n$.

\begin{lemma}[Step 3]\label{lem:step3}
Consider the distribution $\{\tilde{q}_i\}_{i=1}^n$ defined above. We have 
\[       \sum_{i \neq i^*} \tilde{q}_i \p{W(H_{i^*}) + W(H_i)} \le \left(2 - \frac{2}n \right) \cdot W(H_{i^*}) + 2\eps. \]
\end{lemma}

\begin{proof}
Note that all values  $\{\tilde{q}_i\}_{i=1}^n$ are non-negative by construction. By subtracting $1$, \Cref{eq:approx_bound} implies 

    \[       \sum_{i \neq i^*} \tilde{q}_i \p{\hat{W}(H_{i^*}) + \hat{W}(H_i)} \le \left(2 - \frac{2}n \right) \cdot \hat{W}(H_{i^*}). \]
    Since we can assume (by shifting and adjusting $\eps$ in \Cref{lem:sample-scheffe}) that
    \[ \forall i, W(H_i) \le \hat{W}(H_i) \le W(H_i) + \eps, \]
    we have 
    \[   \sum_{i \neq i^*} \tilde{q}_i \p{W(H_{i^*}) + W(H_i)} \le  \sum_{i \neq i^*} \tilde{q}_i \p{\hat{W}(H_{i^*}) + \hat{W}(H_i)} \]
    and 
    \[ \left(2 - \frac{2}n \right) \cdot \hat{W}(H_{i^*}) \le  \left(2 - \frac{2}n \right) \cdot (W(H_{i^*}) + \eps) \le \left(2 - \frac{2}n \right) \cdot W(H_{i^*}) + 2 \eps.\]
    Combining the above two bounds implies 
    \[   \sum_{i \neq i^*} \tilde{q}_i \p{W(H_{i^*}) + W(H_i)} \le   \left(2 - \frac{2}n \right) \cdot W(H_{i^*}) + 2 \eps, \]
    as desired.
    
\end{proof}

Putting together the three steps completes the proof of \Cref{thm:ev_ub}.

\begin{proof}[Proof of Theorem \ref{thm:ev_ub}]
We use the probabilities $\{\tilde{q}_i\}_{i=1}^n$ from \Cref{lem:step3}. We have (where the expectation is under the draw of $\{\tilde{q}_i\}_{i=1}^n$),

\begin{align*}
    \E[]{\dtv(P, H)} &\leq \OPT + \sum_{i \neq i^*} \tilde{q}_i \p{\semi{i}{i^*} + \semi{i^*}{i}} \\
    &\leq \OPT + \sum_{i \neq i^*} \tilde{q}_i \p{W(H_{i^*}) + W(H_i)}.
\end{align*}
Then using \Cref{lem:step3}, 

\begin{align*}
    \E[]{\dtv(P, H)} &\le \OPT  + \left(2 - \frac{2}n \right) \cdot W(H_{i^*}) + 2\eps\\
    &\le \OPT  + \left(2 - \frac{2}n \right) \cdot \OPT + 2\eps \\
    &=  \left(3 - \frac{2}n \right) \cdot \OPT + 2\eps,
\end{align*}
 as desired, since $\OPT \ge W(H_{i^*})$. The theorem follows by adjusting $\eps$.
\end{proof}

%!TEX root = main.tex
\section{Fast Hypothesis Selection with Moderate Probability}\label{sec:moderateprob}

This section is devoted to presenting an algorithm that achieves a time complexity of $\tilde{O}\left(\frac{n}{\epsilon^2\delta}\right)$. We refer to this as the \textit{moderate failure probability} regime, as the running time is competitive with quadratic-time algorithms only when $\delta$ is not too small (i.e., $\delta \geq 1/n$). A high-level description of our algorithm and our proofs can be found in ~\Cref{sec:technical_overview_moderate}. We recall the main theorem which we will prove in this section.

\thmFAST*

\subsection{Terminology and definitions}

\paragraph{A graph theory terminology:}
For the ease of description, we borrow the language of graph theory to describe our algorithm. 
Consider a fixed threshold parameter $\tb \in [0,1]$. Here, we describe a graph $G_b = \left(U = [n],E\right)$. For every hypothesis $\Hi$ in $\HH$, we have a vertex $u_i$ in the graph. We have a subset of vertices $\bS$ that is initially $U$, but it gradually may be emptied out. For every vertex $u_i \in U$ and $v_j \in V$, if $\tb < \hatwij$, then we set a {\em directed} edge from $u_i$ to $v_j$. Connecting these definition to our high-level description in \Cref{sec:technical_overview_moderate}, $\bS$ is the set of vertices that we have not found an incoming edge to it. When $\bS$ becomes empty, we can return $\perp$ for the semi-distance threshold problem. 

We define the {\em (fractional) out-degree} and {\em in-degree} of the vertices as follows:

\begin{align*}
    \dui & \coloneqq \frac{\abs{\left\{v_j \mid (u_i, v_j) \in E\right\}}}{\abs{\bS}}\,,
    \\
    \dvj & \coloneqq  \frac{\abs{\left\{u_i \mid (u_i, v_j) \in E\right\}}}{\abs{U}}\,. 
\end{align*}
Note that the out-degree of a vertex $u_i$ indicates how prompting it is. In particular, every $u_i \in U$ is $\tau$-prompting with respect to $\bS$ for any $\tau \leq \dui$. 
Clearly, the average degree in both $\bS$ and $U$ is $$\dbar \coloneqq \frac{\abs{E}}{\abs{U}\cdot\abs{\bS}}\,.$$

%\textcolor{red}{high level description is needed.}

%In the remainder of this section, we develop a fast algorithm that satisfies the requirements of \Cref{def:threshold-problem}.
%Let $A$ be a small set of hypotheses that allow us to approximate the distances of a hypothesis $\Hj$ to the unknown distribution. In particular, we define the approximate distance: $$\tildeWj \coloneqq \max_{i\in A} \hatwij\,.$$
%Later, we describe a procedure that by adding hypothesis to set $A$ we improve this approximation, which leads us to come up with the final output. 

%Let $\bS\subseteq \HH$ be the set of all hypotheses $\Hj$ for which the approximate distance $\tildeWj \coloneqq \max_{i\in A} \hatwij$ is at most $\tb$. 

For a given vertex $u_i$, its degree $\dui$ determines the fraction of vertices in $\bS$ to which it has an edge. An outgoing edge from $u_i$ to a vertex $v_j \in \bS$ constitutes evidence for removing $v_j$ from $\bS$. Therefore, the degree of $u_i$ indicates the fraction of vertices that can be removed from $\bS$ if we spend $O(\abs{V})$ time to scan for all of its neighbors. For efficiency, we ideally want to find a vertex with a high degree to prune $\bS$ as quickly as possible. We refer to such a vertex as a \emph{prompting hypothesis}, which we formally define as follows:

\begin{definition}
Suppose we are given a set of hypotheses $\bS \subset \HH$ and two parameters $\tau, \tb \in [0,1]$. We say a hypothesis $H_i \in \HH$ is \emph{$(\tau, \tb)$-prompting with respect to $\bS$} if, for at least a $\tau$-fraction of the hypotheses $H_j \in \bS$, we have:
$$ \tb < \hatwij\,. $$
When $\tb$ is clear from the context, we write $\tau$-prompting. 
\end{definition}

\subsection{Estimating degrees and average degrees}
It is straightforward to estimate both the average degree and individual vertex out-degrees up to a constant factor. We provide \Cref{algo:avg_degree_estimation} to estimate the average degree and \Cref{algo:degree_estimation} to estimate individual degrees. The performance guarantees for these algorithms are proven in \Cref{lem:deg_estimation}.

\begin{algorithm}[ht]
\caption{Algorithm for estimating average degree}
\label{algo:avg_degree_estimation}
\begin{algorithmic}[1]
\Procedure{\textsc{Estimate-Average-Degree}}{$\gamma$, $\beta$}
    \State $T \gets C \log(1/\beta)/\gamma$
    \State $X \gets 0$
    \For{$i = 1, \ldots, T$}
        \State $(u,v) \gets$ draw a random pair from $U \times \bS$ uniformly.
        \If{there is an edge from $u$ to $v$}
            \State $X_i \gets 1$
        \Else
            \State $X_i \gets 0$
        \EndIf
        \State $X \gets X + X_i$
    \EndFor
    \State \textbf{return} $\frac{X}{T}$
\EndProcedure
\end{algorithmic}
\end{algorithm}

\begin{algorithm}[ht]
\caption{Algorithm for estimating out-degree}
\label{algo:degree_estimation}
\begin{algorithmic}[1]
\Procedure{\textsc{Estimate-Out-Degree}}{$u$, $\gamma$, $\beta$}
    \State $T \gets C \log(1/\beta)/\gamma$
    \State $X \gets 0$
    \For{$i = 1, \ldots, T$}
        \State $v \gets$ draw a random pair from $\bS$ uniformly.
        \If{there is an edge from $u$ to $v$}
            \State $X_i \gets 1$
        \Else
            \State $X_i \gets 0$
        \EndIf
        \State $X \gets X + X_i$
    \EndFor
    \State \textbf{return} $\frac{X}{T}$
\EndProcedure
\end{algorithmic}
\end{algorithm}

\begin{lemma}\label{lem:deg_estimation}
For two given parameters $\beta$ and $\gamma$, there is an algorithm that uses $O(\log(1/\beta)/\gamma)$ many queries to semi-distances and provide an estimate $\hatd$ for the average degree such that the following holds with probability $1-\beta$
\begin{itemize}
    \item If $\dbar \geq \gamma/2$, then $\hatd$ is a constant approximation of $d$. That is, $\hatd \in [\dbar/2, 2\,\dbar]$. 
    \item If we have $\dbar < \gamma/2$ then $\hatd < \gamma$.
\end{itemize}
Similarly, one can estimate the out-degree of a vertex $u_i \in U$ with the same number of queries to the semi-distances:
\begin{itemize}
    \item If $\dui \ge \gamma/2$, then $\hatd$ is a constant approximation to $\dui$: $\hatd \in [\dui/2, 2\, \dui]$.
    \item If $\dui < \gamma/2$, then $\hatd < \gamma$.
\end{itemize}
The identical statement also holds for in-degree estimation.
\end{lemma}
\begin{proof}
Let $T \ge C \log(1/\beta)/\gamma$ for a sufficiently large constant $C \ge 1$ chosen later. Let $\hatd = \frac{1}T \sum_{i = 1}^T X_i$ be the random variable where each $X_i$ is the indicator for a uniformly chosen pair $(u, v)$ for $u \in U, v \in V$ being a directed edge in $G$. Note that $\E{X_i} = \dbar$ by definition and computing a single $X_i$ takes one semi-distance query.

In the first case where $\dbar \ge \gamma/2$, the Chernoff bound gives us
\[\Pr{|\hatd - \dbar| \ge \dbar/2 } \le 2e^{-\frac{T \dbar}{12}},\]
which implies that if as long as $T = \Omega(\log(1/\beta)/\dbar)$, then $\hatd$ is within a constant factor of $\dbar$ with probability $1-\beta$. Since $\dbar \ge \gamma/2$, we have $1/\gamma \ge 1/\dbar$, so our choice of $T$ is valid. 

In the second case, the upper tail of the Chernoff bound gives
\[
\Pr{\hatd \ge \gamma } \le \Pr{\hatd \ge \left( 1+ \left(\frac{\gamma}{\dbar} -1\right) \right) \cdot \dbar} \le e^{- \Theta( \dbar \cdot T \cdot \frac{\gamma}{\dbar}) \le e^{-\Theta(T \cdot \gamma)}},
\]
so again our choice of $T$ implies that with probability $1-\beta$, $\dbar < \gamma$ in the case where $\dbar = \E{\hatd} \le \gamma/2$. Putting together the two cases completes the proof.

Note that the same analysis also applies to estimating the fractional out-degree or fractional in-degree of vertices in $U$ and $V$ respectively, with the only change being that the expected value is different. 
\end{proof}

\subsection{The Semi-Distance Threshold Problem}\label{sec:binary_search}
We will define the following key subroutine, parameterized by a threshold $\tb$.
\begin{definition}[Semi-Distance Threshold Problem]\label{def:threshold-problem}
The input to \emph{semi-distance threshold} problem are parameters $\tb, \delta \in [0,1]$, hypotheses $\HH$, and samples from $P$.
% An algorithm for this problem may make queries to semi-distance estimates $\hsemi{i}{j}$ for any $i, j \in [n]$ in $O(m)$ time.
The valid outputs are:
\begin{enumerate}[label=(\alph*)]
    \item $\perp$ if, for all $j \in [n]$, there exists $i \in [n]$ s.t. $\hsemi{i}{j} > b$.
    \item $H_j$ for any $j \in [n]$ where $\hsemi{i^*}{j} \leq b$.
\end{enumerate}
An algorithm solving this problem must produce a valid output with probability $1-\delta$.
\end{definition}

We will now describe how to efficiently solve the hypothesis selection problem given access to a subroutine for the semi-distance threshold problem via binary search.
Then, the rest of this section will be devoted to designing an efficient algorithm for this problem.

\begin{theorem}\label{thm:bin-search}
Let $\calA(b)$ be an algorithm for the semi-distance threshold problem with variable parameter $b$ and fixed inputs $\eps, \delta' = \Theta\p{\frac{\delta}{\log(1/\eps)}}, \HH$ and $s=O\p{\frac{\log n + \log(1/\delta) + \log\log(1/\eps)}{\eps^2}}$ samples from $P$.
Let $T \geq 1$ be an upper bound on the runtime of $\calA(b)$ for any $b \in [0,1]$.
Then, there exists an algorithm which solves the hypothesis selection problem with approximation factor $3$, additive error $\eps$, success probability $1 - \delta$, and runtime $O\p{T \log(1/\eps)}$.
\end{theorem}

\begin{proof}
Let $\eps' = \eps/3$. By the choice of $s$ and \Cref{lem:sample-scheffe}, $\abs{\hsemi{i}{j} - \semi{i}{j}} \leq \eps'$ for all $i, j \in [n]$ with probability $1 - \delta/2$.
Consider the following hypothesis selection algorithm.
We will assume that all calls to $\calA$ produce a valid output, though it may fail with probability $\delta'$, we will consider the overall failure probability at the end of the proof.

Query $\calA(0)$ and $\calA(1)$.
If $\calA(0) = H_j$ for some $j$, then the algorithm may return $H_j$.
By \Cref{prop:semidist-approx}, $\dtv(H_j, P) \leq 2\OPT + \semi{i^*}{j} \leq 2\OPT + \eps'$, so $H_j$ is a valid solution.
Note that all semi-distances are bounded in $[0,1]$, so $\perp$ is not a valid output for $\calA(1)$.
We will proceed assuming $\calA(0) = \perp$ and $\calA(1) \neq \perp$.

Let $k_0 = 0$ and $\ell_0 =1$ and consider the following binary search.
For any $t > 0$, $\ell_t - k_t > \eps'$, query the midpoint threshold $\calA\p{\frac{k_t+\ell_t}{2}}$.
If the answer is $\perp$, set $k_{t+1} = \frac{k_t+\ell_t}{2}$ and $\ell_{t+1} = \ell_t$.
Otherwise, set $k_{t+1}=k_t$ and $\ell_{t+1} = \frac{k_t+\ell_t}{2}$, maintaining the invariant that $\calA(k_t) = \perp$ and $\calA(\ell_t) \neq \perp$.
After $t' = \Theta(\log(1/\eps))$ iterations, $\ell_{t'} - k_{t'} \leq \eps'$.
At this point, return $\calA(\ell_{t'})$.

By the correctness of $\calA(k_{t'})$, for all $j \in [n]$, there exists some $i$ such that $\hsemi{i}{j} > k_{t'} > \ell_{t'} - \eps'$.
By the sampling error, $\semi{i}{j} > \ell_{t'} - 2\eps'$.
In particular, this holds for $j = i^*$.
By the underestimation property of semi-distances from \Cref{prop:semidist-underest}, $\OPT > \ell_{t'} - 2\eps'$.

Consider the output hypothesis $H_j = \calA(\ell_{t'})$.
It must be the case that $\hsemi{i^*}{j} \leq \ell_{t'}$, implying that $\semi{i^*}{j} \leq \ell_{t'} + \eps'$.
By \Cref{prop:semidist-approx}, $H_j$ is a valid solution:
\begin{equation*}
    \dtv(H_j, P)
    \leq 2\OPT + \ell_{t'} + \eps'
    \leq 3\OPT + 3\eps'
    = 3\OPT + \eps.
\end{equation*}

It is clear that the total runtime of the algorithm is $O(T\log(1/\eps))$.
Furthermore, by the choice of $\delta'$, the probability that any call to $\calA(b)$ fails is at most $\delta/2$.
Union bounding with the failure probability of the semi-distance estimates, the overall failure probability of the algorithm is at most $\delta$.
\end{proof}

\subsection{Finding a hypothesis with average degree promptingness}\label{sec:ave-deg-prompting}
% We continue using the graph theory terminology introduced above.
The goal of this section is to \emph{find} a vertex $u_i \in U$ with $\dui \ge \Omega(\dbar)$ in roughly $\tilde{O}(1/\dbar)$ semi-distance queries (see \Cref{lem:large_degree} for the precise bound). To do so, we first partition the degrees in $U$ in geometrically decreasing values based on their fractional degrees. This is similar in spirit to "level sets" arguments commonly used in sublinear algorithms, e.g. see the description of the "work investment strategy" of \cite{berman2014lp}.

More formally, recall that $\dbar$ is the average fractional out-degree of the vertices in $U$. Note the prior subsection helps us test the value of $\dbar$, which will be crucially used in the final algorithm of the next subsection. 

Now consider a partition of the interval $[0,1]$ into $k \coloneqq \ceil{\log_2 \left(1/\dbar \right)} + 2$ 
intervals $[0,1] = \bigcup_{r=1}^k\{I_r\}$ as follows:

\begin{align*}
I_r \coloneqq
\begin{cases}
  \left(2^{-r},\,2^{-r+1}\right]  & \quad\quad r \in \{1,\dots,k-1\}\,,
    \vspace{3mm}\\
  \left[0,\,2^{-r+1}\right]         & \quad\quad r = k.
\end{cases}
\end{align*}

The intervals naturally partition the vertices as $U = \bigcup_{r=1}^k  \{U_r\}$, via
$$U_r \coloneqq \left\{ u_i \in U \mid \dui \in I_r \right\}\,.$$
Note we may not explicitly know this partition, but all we require is that such a partition exists. The first lemma states that some partition has relatively large fraction of all the vertices.

\begin{lemma}\label{lem:r*}
There exists an $r^* \in \{1, \cdots, k-1\}$ satisfying
\[ \frac{|U_{r^*}|}{|U|} \ge \frac{ 2^{r^*} \cdot \dbar }{4 \cdot \log_2(1/\dbar)}. \]
\end{lemma}
\begin{proof}
    We can write the average degree $\dbar$, based on the contribution of the vertices in each of these subsets $U_r$ as follows:

\begin{align*}
    \abs{U} \cdot \dbar & = \sum_{u_i \in U} \dui\\  
    &= \sum_{r=1}^k \sum_{u_i \in U_r} \dui
    %\\ & |U| \cdot 2^{-(k-1)} + \sum_{r=1}^{k-1} \sum_{u_i \in U_r} \dui
    \\ & \leq \sum_{r=1}^k \abs{U_r} \cdot 2^{-r+1}\\
    &\leq |U| \cdot 2^{-(k-1)} + \sum_{r=1}^{k-1} \abs{U_r} \cdot 2^{-r+1}
    \,.
\end{align*}
We can upper bound the first term by $|U| \cdot \dbar/2$ using the definition of $k$:

\begin{align*}
	\frac{1}{\dbar} \leq  2^{\ceil{\log_2 \left(1/\dbar \right)}} \le 2^{k-2} \quad \Longleftrightarrow \quad \abs{U} \cdot 2^{-(k-1)} \leq \frac{\abs{U} \cdot \dbar}{2}.
\end{align*}

Combining the previous two equations, we have:

$$\frac{\abs{U} \cdot  \dbar}{2} \leq \sum_{r=1}^{k-1} \abs{U_r} \cdot 2^{-r+1} \le (k-1) \cdot \max_{ 1\le r \le k-1} |U_r| \cdot 2^{-r+1}.$$
Thus, there exists an $r^*$ in $\{1, 2, \ldots, k-1\}$ satisfying
$$\frac{\abs{U} \cdot \dbar}{2\,(k-1)}\leq \abs{U_{r^*}} \cdot 2^{-r^* + 1},$$
and the lemma statement follows by rearranging and using our setting of $k$.
\end{proof}

Now we can prove the main lemma of the subsection, which gives a simple sampling algorithm which, combined with \Cref{lem:deg_estimation}, can output a vertex with out-degree comparable to $\dbar$ in approximately $1/\dbar$ semi-distance queries.

\begin{lemma}\label{lem:large_degree}
    Suppose $\dbar \ge \gamma$ for some parameter $\gamma \in (0, 1)$. There exists an algorithm that uses $  O\left( \frac{\log\left(\frac{1}{\beta\gamma  } \right) \log^2\left(\frac{1}\gamma\right)}{\gamma} \right)$ many queries to semi-distances and outputs a vertex $u_i \in U$ such that $\dui \ge \frac{\gamma}{1000}$ with probability at least $1-\beta$.
\end{lemma}
\begin{proof}
Consider the following simple algorithm which loops over $r = 1, \ldots, k-1$. 
% \justin{$k$ is defined with respect to $\dbar$, should we say somewhere that we are using an estimate for $\dbar$ or something?}
\begin{enumerate}
    \item Set $T_r = \frac{10^4 \log_2(1/\gamma) \log(100/\beta)}{2^{r} \gamma}$ and sample $T_r$ vertices $u_1, \cdots, u_{T_r} \in U$ uniformly at random. 
    \item For $1 \le i \le T_r$, let $X_{u_i}$ be the estimate of the outdegree $\dui$ given by \Cref{lem:deg_estimation} with failure probability $\beta/(100 k T_r)$ and parameter $\gamma_r = 2^{-r}/100$ (in the input to \Cref{lem:deg_estimation}).
    \item If $X_{u_i} > \gamma_r,$ then the algorithm returns $u_i$ and terminates.
\end{enumerate}
(If the algorithm has not terminated at this stage, then we simply return an arbitrary vertex in $U$).

First we bound the total number of semi-distance queries. In the $r$th loop, we make
\[ O\left( T_r \cdot \log\left(\frac{kT_r}\beta\right) \cdot 2^r\right) = O\left( \frac{\log\left(\frac{1}{\beta \gamma} \right) \log\left(\frac{1}\gamma \right)}{\gamma} \right)\]
semi-distance queries, using the bound from \Cref{lem:deg_estimation} and $\dbar \ge \gamma$. Summing across $k = O(\log(1/\dbar)) = O(\log(1/\gamma))$ iterations, we can can upper bound total number of semi-distance queries by
\[ O\left( \frac{\log\left(\frac{1}{\beta\gamma  } \right) \log^2\left(\frac{1}\gamma\right)}{\gamma} \right), \]
as claimed.

Now we handle the approximation guarantee. For every fixed $r$, the estimate $X_{u_i}$ satisfies the guarantees of \Cref{lem:deg_estimation} simultaneously for all $T_r$ vertices with probability at least $1-\beta/(100k)$ via the union bound. Then taking the union bound across the $\le k$ iterations, we know that all the estimates $X_{u_i}$ across all sampled vertices $u_i$ satisfy the guarantees of \Cref{lem:deg_estimation} (with their respective $\gamma_r$ parameters), with probability at least $1-\beta/100$. Concretely, this means that across all iterations $r$, we only terminate if it is the case that the sampled vertex $u_i$ satisfies $\dui \ge \gamma_r/2$. This is because otherwise, $\dui < \gamma_r/2$ and \Cref{lem:deg_estimation} guarantees that our estimate $X_{u_i} < \gamma_r$. We condition on this event. 

Thus, \emph{if} we terminate inside a loop for some $r$, then we know $\dui \ge \gamma_r/2 \ge 2^{-r}/200 \ge \gamma/1000$ since $1 \le r \le k-1$. However, we still need to guarantee that there exists \emph{some} $r$ for which we terminate with high probability. 

Towards this end, let $1 \le r^* \le k-1$ be the value guaranteed by \Cref{lem:r*}. We don't know what this $r^*$, but we must loop over it at some point in our algorithm since we loop over all $1 \le r \le k-1$. When this $r^*$ is considered in the outer loop, we have
    \[\Pr{u_i \in U_{r*}} \ge \frac{ 2^{r*} \cdot \dbar }{100 \cdot \log_2(1/\dbar)}.\]

Thus by our choice of $T_{r^*}$, we know that the probability that at least one sampled vertex $u_i \in U_{r*}$ is at least 
\[ 1 - \left( 1 - \frac{ 2^{r*} \cdot \dbar }{100 \cdot \log_2(1/\dbar)} \right)^{T_{r^*}} \ge 1 - e^{-  \frac{ 2^{r*} \cdot \dbar }{100 \cdot \log_2(1/\dbar)}  \cdot T_{r^*}}  \ge 1-\beta/100,\]
since 
\begin{align*}
     \frac{ 2^{r*} \cdot \dbar }{100 \cdot \log_2(1/\dbar)}  \cdot T_{r^*} &= \frac{ 2^{r*} \cdot \dbar }{100 \cdot \log_2(1/\dbar)} \cdot  \frac{10^4 \log_2(1/\gamma)}{2^{r^*} \gamma} \cdot \log(100/\beta)\\
     &\ge 10^2 \cdot \frac{\dbar/\log_2(1/\dbar)}{\gamma/\log_2(1/\gamma)} \cdot \log(100/\beta) \\
     &\ge 10^2 \log(100/\beta),
\end{align*}
since the function $x/\log_2(1/x)$ is increasing in the interval $[0, 1]$ and $\dbar \ge \gamma$ by assumption. We condition on this event. 

Thus, some vertex $u_i \in U_{r^*}$ is sampled in the $r^*$th iteration of the outer loop. For this sampled vertex, we have $\dui \ge 2^{-r^*} \ge 100\gamma_{r^*}$ by definition of being in $U_{r^*}$. By our earlier stated guarantees on degree estimation across all sampled vertices, we know that the estimate satisfies $X_{u_i} \ge 10\gamma_r$ for this vertex, so if we have not terminated before, we are guaranteed to at least terminate on this vertex. This completes the proof.
\end{proof}

\subsection{Finding prompting hypothesis via neighbor set}
\label{sec:small-neighborhood}
In this section, we propose another approach for finding a prompting hypothesis. More precisely, we find a vertex with  out-degree $\dui$ at least $\tilde{O}(\beta)$ with probability $1-\beta$. The main difference between this approach  and the one we describe in the previous section is that the quality of the vertex we find remains the same regardless of the average degree, however its query complexity varies depending on the average degree of the underlying graph. The query complexity of this approach is $\tilde{O}(n+ n\dbar/\beta)$, which makes this approach suitable for scenarios where the average degree is low, i.e., $\bar{d} \ll \beta \approx \delta$. When the average degree is this small, attempting to approximate the degree of all vertices in $U$ by querying semi-distances would be too costly. 

At a high-level, the algorithm samples a vertex $v_j$ from the bottom set, $\bS$, uniformly at random and then exhaustively scans the upper set, $U$, to find all of its neighbors, which we denote as $\text{Nei}(v_j)$. Since the average degree is low, we anticipate that this neighbor set will be small by Markov's inequality. To achieve a high probability of success, we repeat this sampling process $t = O(\log(1/\beta))$ times to make sure that the one with the minimum-sized neighbor set, $v_{j_{\min}}$, certainly has a small set of neighbors.

Once we have the neighbor set, two possibilities arise. First, there might be a $ \tilde{O}(\beta)$-prompting hypothesis among the neighbors. Given the small size of the neighbor set, we can efficiently test each neighbor to see if it is a prompting one (i.e., if its estimated degree $\hat{d}_{u_i}$ exceeds a threshold $ \tilde{O}(\beta)$). If such a vertex is found, the algorithm returns it and halts.

Second, it's possible that no prompting hypothesis exists in the neighbor set. In this case, while the immediate search for a prompting hypothesis has failed, the sampled vertex $v_j$ has a useful property: it is unlikely to have an edge from $u_{i^*}$, meaning it cannot be prompted past our threshold, $\tb$ via $\Hstar$. This property can then be utilized to determine the final output for the hypothesis selection (which we discuss later).

The pseudocode for this procedure is described in \Cref{algo:delta_prompting}. And, the correctness of our algorithm is formally proven in the following theorem.

\begin{algorithm}
\caption{Find prompting}
\label{algo:delta_prompting}
\begin{algorithmic}[1]
\State \textbf{input}: $\beta$, $\hatd$, and query access to edges of $G$.
\State Set $\beta' \gets \beta/4$.
%\State Set $t \gets \ceil{\log_2(1/\beta')}$.
\State Set $t=c\log n$ for a sufficiently large constant $c$.%new
\For{$r = 1, \ldots, t$}
    \State $v_{j_r} \gets$ Sample a vertex from $\bS$ uniformly at random.
    
    \State Search all $u_i \in U$, and find the set of vertices $\text{Nei}(v_{j_r})$ which have an incoming edge to $v_{j_r}$.
\EndFor
\State $S\gets $  all sampled vertices $v_{j_r}$. %new
\State $T\gets$ all  vertices $v_{j_r}\in S$ where $|\text{Nei}(v_{j_r})|\leq 20\hatd\cdot n$ %new

%\State Let $v_{j_{\min}}$ be the vertex in $\{v_{j_1}, v_{j_2}, \ldots, v_{j_t}\}$ with the smallest neighbor set.

%\If{$\abs{\text{Nei}(v_{j_{\min}})} \geq 4\,\hatd \cdot n$}
%\State \Return $\bot$ and halt. \label{line:bot} 
%\EndIf

%\For{every $u_i \in \text{Nei}(v_{j_{\min}})$}
    %\State $\hatdui \gets$ \textsc{Estimate-Out-Degree}$\left(u_i,\ \beta'/\abs{\text{Nei}(v_{j_{\min}})},\ \beta'/(2\,t)\right)$. \label{line:out-deg-estimation}
    
    %\If{$\hatdui \geq \beta'/(2\,t)$}
        %\State \textbf{return} $H_i$, corresponding to $u_i$, as a prompting hypothesis and halt. \label{line:output_prompting}
    %\EndIf
%\EndFor
\For{every $v_j\in T$ and every $u_i\in \text{Nei}(v_j)$} %new
    %\State $\hatdui \gets$ \textsc{Estimate-Out-Degree}$\left(u_i,\ \beta'/\abs{\text{Nei}(v_{j})},\ \beta'/(2\,t)\right)$. \label{line:out-deg-estimation}
    \State $\hatdui \gets$ \textsc{Estimate-Out-Degree}$\left(u_i,\ \beta'/(2t),\ n^{-11}\right)$. \label{line:out-deg-estimation} %new

    %\If{$\hatdui \geq \beta'/(2\,t)$}
    \If{$\hatdui \geq \beta'/(2t)$} %new
        \State \textbf{return} $u_i$ as a prompting hypothesis and halt. \label{line:output_prompting}
    \EndIf
\EndFor
\If {no prompting hypothesis is found}\label{line:no-prompting} 
\State $S'\gets$ another $t$ randomly sampled vertices from $V$.
\State $T'\gets$ all  vertices $v_{j_r}\in S'$ where $|\text{Nei}(v_{j_r})|\leq 20\hatd\cdot n$ %new
\For{every $v_j\in S'$}
\For{every $u_i\in \text{Nei}(v_j)$}

\State $\hatdui \gets$ \textsc{Estimate-Out-Degree}$\left(u_i,\ 2\beta'/t,\ n^{-11}\right)$. \label{line:out-deg-estimation2} %new
\EndFor
\If{$\hatdui\leq 2\beta'/t$ for all $u_i\in \text{Nei}(v_j)$}
\State \textbf{return} $v_j$
\EndIf
\EndFor
%\State\textbf{return} $H_{j_{\min}}$ as a vertex with no edge from $\ustar$.
%\State \textbf{return} a random hypothesis $v_j\in S$. %new
\EndIf
\State \textbf{return} \textsf{FAIL}
\end{algorithmic}
\end{algorithm}

\begin{lemma}\label{lem:delta_prompting}
Consider \Cref{algo:delta_prompting} with arbitrary input parameters $\beta, \hatd \in (0,1)$ and query access to the graph $G$ (or equivalently to the semi-distances). Suppose that $\hatd$ is guaranteed to satisfy $\hatd \leq 2\dbar$. There exists an error event $\mathcal{E}$ such that $\Pr{\mathcal{E}}\leq n^{-10}$. Moreover, conditioning on $\neg \mathcal{E}$, if~\Cref{algo:delta_prompting} returns
\begin{itemize}
    \item a hypothesis $u_i$, then $u_i$ is $\Omega\left(\beta/(\log n)\right)$-prompting with probability 1, and if it returns
    \item a vertex $v_{j}$, then with probability $1-\beta$, there  is no edge from $\ustar$ to $v_{j}$.
\end{itemize}
The query complexity of the algorithm is:
$$O\left(n\log n+\frac{\hatd n(\log n)^2}{\beta}\right)\,.$$
In the special case where $\hatd =O(\beta)$, this complexity becomes nearly-linear in $n$: $O\left(n \cdot (\log n)^2\right)$.
\end{lemma}

\begin{proof}
Let us first focus on the probabilistic correctness guarantee. To do so, we will show that certain {\em bad} events each occur with low probability and define $\mathcal{E}$ as the union of these events. We next show that conditioned on $\neg\mathcal{E}$, the probabilistic guarantees in the theorem hold.
%Let us first focus on the probabilistic correctness guarantee. To do so, we will show that certain {\em bad} events each occur with a probability of at most $\beta' \coloneqq \beta/3$. By a union bound over these events, the desired output guarantees will hold with a total probability of at least $1-\beta$.

We first argue that the probability of any out-degree estimation in \Cref{line:out-deg-estimation} or~\Cref{line:out-deg-estimation2} being incorrect is at most $n^{-10}$. The proof follows directly from the guarantee of \Cref{lem:deg_estimation}. Each degree estimation is accurate with a probability of at least $1-n^{-11}$. By a union bound, all estimations satisfy the properties of \Cref{lem:deg_estimation} with a collective probability of at least $1-n^{-10}$. Denote the small probability error event by $E_1$.

We first argue that if $E_1$ doesn't hold, then the first bullet point in the theorem statement is true, namely if the algorithm outputs an $u_i$, then $u_i$ is indeed sufficiently prompting. To later prove the statement of the second bullet point, we have to extend $E_1$ to the larger error event $\mathcal{E}$, but if an output is correct with probability $1$ conditioned on $\neg E_1$, this will also hold when conditioning on $\neg \mathcal{E}\subset \neg E_1$.
\begin{enumerate}
    \item[1.] \textbf{Correctness:} Conditioning on $\neg E_1$, the algorithm does not output an invalid prompting hypothesis. Specifically, if the algorithm outputs a hypothesis $u_i$ in \Cref{line:output_prompting}, its degree is guaranteed to be $\dui \geq \beta'/(4t) = \Omega(\beta/(\log n))$.
\end{enumerate}

We next turn our attention to the statement of the second bullet point. For this we require an additional \emph{completeness} guarantee of the algorithm, namely that a $(\beta'/t)$-prompting hypothesis in the neighborhood of $v_j$ cannot escape the degree estimators attention.

\begin{enumerate}
    \item[2.] \textbf{Completeness:} Conditioning on $\neg E_1$, the algorithm finds any substantially prompting hypothesis within the chosen neighbor set. More formally, if a vertex $u_i$ exists in the neighbor set of some $v_{j}$ with $\dui \geq \beta'/t$, the algorithm is guaranteed to output a hypothesis in \Cref{line:output_prompting} and halt. This is because its estimated degree will be at least $\beta'/(2t)$.
\end{enumerate}

Towards the result, we first show a simple but useful lemma stating that most vertices in $v$ have small neighborhoods in $U$.% with high probability in $n$, a large fraction of the vertices in $S$ are also in $T$ (and an analogues claim holds for $S'$ and $T'$).
\begin{lemma}\label{lemma:markov-degree}
at least a $(9/10)$-fraction of elements $v_{j_r}\in V$ have $\abs{\text{Nei}(v_{j_r})}\leq 20\,\hatd \cdot \abs{U}$.
%Suppose we are given an estimate $\hatd$ of the average degree such that $\hatd \geq \dbar/2$. Let $E_2$ denote the event that $|T|< \frac{4}{5}|S|$ Then, $\Pr{E_2}\leq n^{-11}$.
\end{lemma}
\begin{proof}
For a vertex $v_{j_r}$ chosen randomly from $\bS$, its expected neighborhood size is $\E{\abs{\text{Nei}(v_{j_r})}} = \dbar \cdot \abs{U}$ by definition. Therefore, by Markov's inequality and the guarantee that $\hatd \geq \dbar/2$,
\[
\Pr{\abs{\text{Nei}(v_{j_r})} > 20\,\hatd \cdot \abs{U}} \leq \Pr{\abs{\text{Nei}(v_{j_r})} > 10\,\dbar \cdot \abs{U}} \leq \frac{1}{10},
\]
as desired.
%In other words, at least a $(9/10)$-fraction of elements $v_{j_r}\in V$ have $\abs{\text{Nei}(v_{j_r})}\leq 20\,\hatd \cdot \abs{U}$. It then follows from a standard Chernoff bound that if we sample $c\log n$ such $v_{j_r}$ for a sufficiently large constant $c$, then  at least a $(4/5)$-fraction of the samples will have the desired property with the desired high probability. This completes the proof of the claim.
\end{proof}

Denote by $F$ denote the set of elements in $V$ having degree at most $20\hatd |U|$ and satisfying that at least one of their neighbors $u_i\in U$ has $\dui \geq \beta'/t$, namely is  $(\beta'/t)$-prompting. Denote by $f=|F|/|V|$ the fraction of such elements. Let $E_2$ denote the error event that $f\geq 1/2$ and that~\Cref{line:out-deg-estimation} does not return a prompting hypothesis.
\begin{lemma}
It holds that $\Pr{E_2}=O( n^{-10})$.
\end{lemma}
\begin{proof}
We can bound $\Pr{E_2}\leq \Pr{E_2\wedge \neg E_1}+O(n^{-10})$. However, note that if $f\geq 1/2$, then sampling $r=c\log n$ random hypotheses from $V$, one of them $v_j$ will lie in $F$ with the desired high probability. If this happens, since $E_1$ implies that all degree estimates are accurate, and containment in $F$ implies that $v_j$ has  neighbor in $U$ with $\dui \geq \beta'/t$, it follows from the Completeness guarantee above, that the algorithm will return a $(\beta'/(2t))$-prompting neighbor of $v_j$.
\end{proof}

%Our first claim is that the algorithm is unlikely to output $\bot$ in \Cref{line:bot}. The formal statement is as follows.

%\begin{lemma}
%Suppose we are given an estimate $\hatd$ of the average degree such that $\hatd \leq 2\,\dbar$. Then, with probability at least $1-\beta'$, the size of the neighbor set of $v_{j_{\min}}$ satisfies $|\text{Nei}(v_{j_{\min}})| \leq 4\,\hatd \cdot |U| \leq 4\,\hatd \cdot n$.
%\end{lemma}
%\begin{proof}
   %For a vertex $v_{j_r}$ chosen randomly from $\bS$, its expected neighborhood size is $\E{\abs{\text{Nei}(v_{j_r})}} = \dbar \cdot \abs{U}$ by definition. Therefore, by Markov's inequality, the probability of its neighborhood size being more than twice the expected value is at most $1/2$:
%\[
%\Pr{\abs{\text{Nei}(v_{j_r})} > 2\,\dbar \cdot \abs{U}} \leq \frac{1}{2}
%\]
%Now, consider the $t$ independently chosen vertices $v_{j_1}, \ldots, v_{j_t}$. The probability that all of them have a neighbor set of size greater than $2\,\dbar \cdot \abs{U}$ is at most $(1/2)^t = 2^{-t} \leq \beta'$.

%In the algorithm, we set $v_{j_{\min}}$ to be the vertex with the smallest neighborhood. The event $\abs{\text{Nei}(v_{j_{\min}})} > 2\,\dbar \cdot \abs{U}$ is the same as the event where every sampled vertex's neighbor set is larger than that quantity. Thus, the probability that the minimum neighbor set size  exceeds $2\,\dbar \cdot \abs{U}$ is also at most $\beta'$. The statement of the lemma is then derived from the fact that $\hatd \leq 2\,\dbar$.
%\end{proof}

We finally define $\mathcal{E}=E_1\cup E_2$ and show that if $\mathcal{E}$ does not occur, then if the algorithm outputs some $v_j$, the probability that $v_{j}$ has an edge from $\ustar$ is at most $1-\beta$.  For the analysis, consider the first time that the algorithm proceeds to line~\Cref{line:no-prompting} and assume that $\mathcal{E}$ does not occur. Since $E_2$ did not occur, we know that $f<1/2$. By~\Cref{lemma:markov-degree}, at least a $(9/10)$-fraction of elements $v_{j_r}\in V$ have $\abs{\text{Nei}(v_{j_r})}\leq 20\,\hatd \cdot \abs{U}$, and it follows that at least at least an $1/2-1/10=2/5$ fraction of elements in $V$ have no $(\beta'/t)$-prompting neighbors and have degree at most $ 20\,\hatd \cdot \abs{U}$. Since $S'$ has size $c\log n$ for a large constant $c$, it follows that with high probability in $n$, at least one such $v_j$ will be sampled into $S'$, and the algorithm will indeed return some $v_j$ and not $\perp$. We now argue that $v_{i^*}$ has no edge to $v_j$ with probability at least $1-\beta$. We consider the following two cases.

%These two properties imply that with probability at least $1-\beta'$, the algorithm does not output a vertex $v_{j_{\min}}$ that has an edge from $\ustar$. This claim is proven by considering two cases for the degree of $\ustar$.

\textbf{Case 1: $\dustar \geq 4\beta'/t$}. 
Since $f<1/2$, it holds that $|V\setminus F|>|V|/2$ and in particular, more than half of the elements of $V$ have no $(\beta'/t)$-prompting neighbors.  For any such $v_j$, since $E_1$ implies that all degree estimates are accurate, for all neighbors $u_j$ of $v_j$, we must have that $\hatdui\leq 2\beta'/t$, so the algorithm can return $v_j$. Moreover, by the Completeness guarantee, if  $u_{i^*}$ is in the neighborhood of some $v_j$, then since $\dustar \geq 4\beta'/t$ and by the correctness guarantee, the algorithm will never return such a $v_j$. We conclude that in this case, the algorithm with very high probability returns a $v_j$ such that $u_{i^*}$ has no edge to $v_j$, as desired.

%We consider two sub-cases based on whether an edge exists from $\ustar$ to $v_{j_{\min}}$.
%\begin{itemize}
    
    %\item If there is no edge from $\ustar$ to $v_{j}$ for any $j\in T$, then algorithm's output is in fact correct with probability 1. It will either output a valid prompting hypothesis (by the Correctness property) or it will output a random $v_{j}\in T$ but none of these have an edge from $\ustar$.

   % \item If there is an edge from $\ustar$ to $v_{j}$ for some $v_j\in T$, then $\ustar$ is in the neighbor set $\text{Nei}(v_j)$. Because $\dustar \geq \beta'/t$, the Completeness property guarantees that the algorithm will identify $\ustar$ as a prompting hypothesis and halt (either outputting $\ustar$ or some other hypothesis which is at least $(\beta'/(4t))$-prompting). Therefore, $v_{j}$ will not be the final output.
%\end{itemize}

\textbf{Case 2: $\dustar < 4\beta'/t$}.
Since $u_{i^*}$ has small degree, the probability that any single sample in $S'$ has $u_{i^*}$ as a neighbor is at most $4\beta'/t$. By a union bound, the probability that there exists an element in $v_j\in S'$ such $u_{i^*}$ has an edge to $v_j$ is at most $4\beta'=\beta$. Thus, with probability $1-\beta$, the $v_j$ returned does not have an edge from $u_{i^*}$, as desired. 
%In this case, we argue it is unlikely for $\ustar$ to have an edge to any vertex in the randomly sampled set $\{v_{j_1}, \ldots, v_{j_t}\}$. The probability that $\ustar$ has an edge to any single randomly chosen vertex $v_{j_r}$ is, by definition, its degree $\dustar$. Using a union bound over the $t$ samples, the probability that $\ustar$ has an edge to any vertex in the set is:
%$t \cdot \dustar < \beta'' $.
%Thus, with probability at least $1-\beta''$, $\ustar$ has no edge to any vertex in $\{v_{j_1}, \ldots, v_{j_t}\}$. If this occurs, the algorithm's output is correct because it will either output a valid prompting hypothesis or one of the $v_{j_r}$, none of which have an edge from $\ustar$. This completes the proof of the algorithm's correctness.

\paragraph{Query complexity:}
For each of the $2t$ sampled vertices we have to iterate over all possible $u_i$ and query the semi-distances to check whether an edge exists or not. Hence, finding the neighbor sets takes $O(\abs{U} \cdot t) = O(n \cdot \log n)$ queries. For every $u_i$ in some $\text{Nei}(v_{j})$, where $\abs{\text{Nei}(v_{j})}\leq 20\hatd n$, we have to estimate their out-degree, which takes the following number of queries via \Cref{lem:deg_estimation}:

\begin{align*}
    \text{Query complexity of each degree estimation:} & = O\left(
    \frac{t \cdot \log n}{\beta'}\right) =  
    O\left(\frac{(\log n)^2}{\beta}\right) 
    % \\ & 
    % = O\left(\frac{\log(\frac{1}{\beta}) \cdot \max\left(\log(\frac{1}{\beta})
    % ,\ 
    % \log(\hatd\cdot n)\right)}{\beta'}
    % \right)
\,.
\end{align*}  

Hence, the total query complexity is: 
\begin{align*}
    \text{Query complexity } =
    O\left(n\log n+\frac{\hatd n(\log n)^2}{\beta}\right)
\end{align*}
Hence, the proof is complete. 
\end{proof}

\subsection{Putting the pieces together} \label{sec:moderate_delta_putting_together}

\begin{algorithm}[h!]
\caption{Semi-Distance Threshold Algorithm}
\label{alg:empty_set}
\begin{algorithmic}[1]
\Require Threshold $\tb$, failure probability $\delta \geq 1/n$, and access to the graph $G=(U \cup V, E)$ 
%, set $A$, and $\tildeWj$'s
\State $\zeta \gets \Theta(1/n^4)$
\While{$\abs{\bS} > 0$}
    \State $\hat{d} \gets$ \textsc{Estimate-Average-Degree}$\left(\delta, \zeta\right)$

    \If{$\hat{d} < \delta$}
        \State Run \Cref{algo:delta_prompting} with parameter $\beta = \delta$.
        \If{it returns a prompting hypothesis $u_i$}
            \State Update $\bS$ by removing any $v_j$ with $\hsemi{i}{j} \geq \tb$.
            \State \textbf{continue}
        \ElsIf{it returns $v_j$ (a vertex without an edge from $u_{i*}$)}
            \State \Return $\Hj$ and halt.
        \EndIf
    \Else
        \State Run the algorithm of \Cref{lem:large_degree} with $\gamma=\delta/2$ and failure probability $\zeta$ to get a prompting hypothesis $\Hi$.
        % \State Add $\Hi$ to $A$.
        % \State Update all $\tilde{W}_j$'s and update $\bS$ by removing any $v_j$ with $\tildeWj \geq \tb$.
        \State Update $\bS$ by removing any $v_j$ with $\hsemi{i}{j} \geq \tb$.
    \EndIf
\EndWhile
\State \Return $\bot$.
\end{algorithmic}
\end{algorithm}

In this section, we combine the tools we have developed to solve the Semi-Distance Threshold Problem via an iterative process.
The theorem will then follow from the binary search over thresholds from \Cref{thm:bin-search}.
We repeatedly find prompting hypotheses to empty out $\bS$.
We will end this process with high probability in one the following cases:
\begin{enumerate}
    \item The set $\bS$ is fully emptied out and we return $\perp$.
    \item A vertex $v_j$ is returned that does not have an incoming edge from $\ustar$: $\hatwistarj$ is at most $b$.
\end{enumerate}

\Cref{alg:empty_set} describes our procedure. Our goal in this section is to show that this process stops in few rounds.

\begin{theorem}\label{thm:semidist-thresh}
\Cref{alg:empty_set} solves the Semi-Distance Threshold Problem (see \Cref{def:threshold-problem}) with $s$ samples, $n$ hypotheses, and failure probability $\delta$  and runs in time 
\begin{equation*}
    O\p{\frac{n s \log^4(n)}{\delta}}.
\end{equation*}
\end{theorem}

\begin{proof}
Consider the first $\Theta\p{n^3}$ iterations of \Cref{alg:empty_set}.
By the choice of failure probabilities with \Cref{lem:deg_estimation}, \Cref{lem:large_degree}, and \Cref{lem:delta_prompting}, all calls to subroutines succeed with probability greater than $1-1/n > 1 - 1/\delta$.
We will condition on this event and later show that the algorithm will terminate within this number of iterations.

Consider any single such iteration of the while loop in \Cref{alg:empty_set}.
By the degree estimation guarantees of \Cref{lem:deg_estimation}, if $\hat{d} < \delta$, then $\dbar < \delta/2$.
In this case, by the guarantee of \Cref{lem:delta_prompting}, either we find a $\delta/(\log n)$-prompting hypothesis or we find a vertex $v_j$ such that $\semi{i^*}{j} \leq b$ with probability $1-\delta$.
If we are in the latter case, then the algorithm will return $v_j$ and satisfy the guarantee of the Semi-Distance Threshold Problem with overall failure probability $O(\delta)$.
Assuming that this event does not occur, then $V$ will shrink by a factor of $1 - \Omega(\delta/\log(1/\delta))$.
After $O\p{\frac{\log n}{\delta}}$ of these iterations, $V$ will shrink by half.
There can be a total of $O\p{\frac{\log^2 n}{\delta}}$ such iterations.
By \Cref{lem:delta_prompting}, each iteration requires queries $O(n\log^2 n)$.
The total time of these iterations is:
\begin{equation*}
    O\p{\frac{n s \log^4 n}{\delta}}.
\end{equation*}

If we are in the case that $\hat{d} \geq \delta$, then $\bar{d} \geq \delta/2$.
By \Cref{lem:large_degree}, $V$ will shrink by a factor of $1 - \Omega(\bar{d}) \geq 1 - \Omega(\delta)$.
So, there can be a total of $O(\log n / \delta)$ of these iterations.
By \Cref{lem:large_degree}, the queries for each iteration is $O(\log n \log^2(1/\delta) / \delta)$, yielding total time:
\begin{equation*}
    O\p{\frac{s \log^2 n \log^2(1/\delta)}{\delta^2}}.
\end{equation*}
As $\delta > 1/n$, this quantity is dominated by the preceding one.

Finally, each call to estimate the average degree requires $O(s\log(n)/\delta)$ time, which is dominated by the time required for calls to either of the other two sub-routines.
\end{proof}

\begin{proof}[Proof of \Cref{thm:fast}]
The proof follows directly from \Cref{thm:bin-search} and \Cref{thm:semidist-thresh}.
Letting $\delta' = \delta/\log(1/\eps)$ and $s = O\p{\frac{\log n + \log(1/\delta) + \log\log(1/\eps)}{\eps^2}}$, the overall runtime becomes
\begin{align*}
    &O\p{\frac{n\log^4(n)\p{\log n + \log(1/\delta) + \log\log(1/\eps)}\p{\log(1/\eps)}}{\eps^2\delta}} = \tilde{O}\p{\frac{n}{\eps^2 \delta}}.
    % = \; &
    % O\p{\frac{n \log^5(n) \log(1/\eps) + n \log^3(n) \log(1/\eps) (\log\log(1/\eps))^2}{\eps^2 \delta}}.
\end{align*}
\end{proof}

\section{Fast Hypothesis Selection with High Probability via Additional Resources}\label{sec:highprob}

\subsection{Near-Linear Time with Known Optimal Distance}\label{sec:knownopt}
In this section, we consider the setting where $\OPT=\min_i d_{TV}(P,H_i)$ is provided to us as part of the input. We prove the following theorem.

\thmKnownOPT*

% \begin{theorem}
%     Suppose the algorithm is given some real $R$ as part of the input with the guarantee that $R\geq\OPT$. Then, there exists an algorithm with sample complexity $O\p{\frac{\log (n/\delta)}{\eps^2}}$ and runtime $O\p{\frac{n\log^3 (n/\delta)}{\eps^2}}$, which with probability at least $1-\delta$ returns a hypothesis $H_i$ such that $\dtv(P, H_i) \leq 2 \cdot \OPT + R   + \eps$.
% \end{theorem}
\begin{proof}
Our algorithm is again based on the semi-distances used in the previous sections. The algorithm is quite simple and we now describe how it operates. Define $\eps_0=\eps/2$ and $\delta_0=\delta/3$. Based on the samples, we define the approximate semi-distances $\hsemi{i}{j}$. By~\Cref{lem:sample-scheffe} and a union bound over all distinct $i,j\in[n]$, we have that $|\hsemi{i}{j}-\semi{i}{j}|\leq \eps_0$ for all distinct $i,j\in[n]$ with probability $1-\delta_0$. We know that the optimal hypothesis $H_{i^*}$ has $W(H_{i^*})\leq \OPT$, and so, we have that $\hat W(H_{i^*})\leq \OPT+\eps_0$. We let $S_0=\{H_1,\dots,H_n\}$ and construct sets $S_1, S_2,\cdots,$ as follows. Having constructed $S_{k-1}$, we iterate over all hypotheses $H_j\in S_{k-1}$, and sample with replacement a multiset of hypotheses $A_j\subset \{H_1,\dots, H_n\}\setminus \{H_j\}$ of size $s_k=\alpha 2^k \log (n/\delta_0)$ for a sufficiently large constant $\alpha$ that does not depend on $k$. Define,
\[
\lambda_j=\frac{|\{H_i\in A_j\mid \hsemi{i}{j}>R+\eps_0\}|}{s_k}.
\]
If we find a $j$ such that $\lambda_j \leq \frac{1}{2^{k+1}}$, we define $S_k=\{i\in [n]\setminus \{j\}\mid \hsemi{i}{j}>R+\eps_0\}$ and further define $j_{k-1}=j$. In this case, if $S_k=\emptyset$, we return $H_{j_{k-1}}$. Otherwise, if no such $j$ with $H_j\in S_{k-1}$ exists, we halt the process of constructing the sets $S_k$, and our algorithm returns $H_{j_{k-2}}$ as the final output. We will show shortly that with high probability, if the algorithm does not halt, then $S_{k_0}= \emptyset$ for some  for $k_0\leq \lceil\lg n\rceil $.

Define
\[
\lambda_j'=\frac{|\{i\in [n]\setminus\{j\}\mid \hsemi{i}{j}>R+\eps_0\}|}{n-1}.
\]
We state and prove two claims bounding the probability of two error events over the randomness of the algorithm.
\begin{claim}\label{claim:all-small}
With probability at least $1-\delta_0$, for all $1\leq k\leq \lceil\lg n\rceil$, such that $S_k$ exists, we have that $|S_k|<n/2^k$.
\end{claim}
\begin{proof}[Proof of Claim]
Suppose $j$ with $H_j\in S_{k-1}$ is such that $\lambda_j'> 1/2^{k}$. Then, a Chernoff bound gives that with probability at least $1-\delta_0/n^2$, it holds that $\lambda_j> 1/2^{k+1}$ as long as $\alpha$ is sufficiently large, and thus such a $j$ is not used to define $S_k$. In particular, union bounding over all such $j$ with $H_j\in S_{k-1}$, we obtain that if $S_k$ exists, then $|S_k|\leq 1/2^{k}(n-1)<n/2^k$ with probability $1-\delta_0/n$. Crudely union bounding over all  $k\leq \lceil\lg n\rceil$ gives the result.
\end{proof}
\begin{claim}\label{claim:no-halt}
With probability at least $1-\delta_0$, there exists no $k\leq \lceil \lg n \rceil$ such that for some $H_j\in S_{k-1}$, it holds that $\lambda_j'\leq 1/2^{k+2}$ but where the algorithm halts before constructing $S_k$. 
\end{claim}
\begin{proof}[Proof of Claim]
If there exists a $j$ with $H_j\in S_{k}$ such that $\lambda_j'\leq 1/2^{k+2}$, then it follows from a Chernoff bound that with probability $1-\delta_0/n$, $\lambda_j \leq 1/2^{k+1}$ (if  $\alpha$ is sufficiently large) in which case the algorithm does not halt. Crudely union bounding over all $k\leq \lceil \lg n\rceil$ gives the result.
\end{proof}

Let $E_0$ denote the event that $|\hsemi{i}{j}-\semi{i}{j}|\leq \eps_0$ for all distinct $i,j\in [n]$. Let $E_1$ and $E_2$ denote the probability at most $\delta_0$ events of respectively~\Cref{claim:all-small}  and~\Cref{claim:no-halt}. We show that if neither of these events occur, then the algorithm does in fact return a hypothesis $H_i$ with $\dtv(P, H_i) \leq 3 \cdot \OPT  + \eps$. The result follows as $\Pr{\bigcup_{i=1}^3 E_i}\leq 3\delta_0=\delta$.
\paragraph{Case 1: The algorithm does not halt.}
Since $E_1$ did not occur, $S_{k_0+1}=\emptyset$ for some $0\leq k_0<\lceil \lg n\rceil$. At the point where this happens, the algorithm terminates by returning $H_{j_{k_0}}$.
Now by definition, $S_{k_0+1}=\emptyset$ is equivalent to $\hsemi{i}{j_{k_0}}\leq R+\eps_0$ holding for all $i\in [n]\setminus \{j_{k_0}\}$, which in turn implies that $\semi{i}{j_{k_0}}\leq R+2\eps_0$ for all $i\in [n]\setminus \{j_{k_0}\}$, and in particular, $\semi{i^*}{j_{k_0}}\leq R+2\eps_0$. The result then follows from~\Cref{prop:semidist-approx}.

\paragraph{Case 2: The algorithm halts.}
Note first that the algorithm does not halt when constructing $S_1$. Indeed, $H_{i^*}\in S_0$ and $\hat W(H_{i^*})\leq \OPT+\eps_0\leq R+\eps_0$, so for $j=i^*$, $\lambda_j=0$ regardless of the sampling from $S_0$. Let thus $k\geq 2$ be such that the algorithm halts in constructing $S_{k}$. The algorithm then returns $H_{j_{k-2}}$ which is well-defined since $k\geq 2$.
Since $E_2$ did not occur, for all $H_j\in S_{k-1}$, $\lambda_j'>1/2^{k+2}$. In particular $H_{i^*}\notin S_{k-1}$ since $\lambda_{i^*}=0$. But $S_{k-1}=\{i\in [n]\setminus \{j_{k-2}\}\mid \hsemi{i}{j_{k-2}} >R+\eps_0\}$, and thus $\hsemi{i^*}{j_{k-2}}\leq R+\eps_0$. This implies that $\semi{i^*}{j_{k-2}}\leq R+2\eps_0$ and the result again follows from~\Cref{prop:semidist-approx}.

\paragraph{Running time.}
We next argue about the runtime of the algorithm. Estimating a single semi-distance based on the samples takes time $O(\frac{\log (n/\delta)}{\eps^2})$, so estimating a single $\lambda_j$ takes time $O(\frac{2^k\log^2 (n/\delta)}{\eps^2})$, and the time to estimate $\lambda_j$ for all $H_j\in S_{k-1}$ is thus $O(\frac{n\log^2 (n/\delta)}{\eps^2})$ assuming that $E_1$ did not occur. Defining $S_k$ requires us to calculate all semi-distances from $j_{k-1}$ which thus takes time $O(\frac{n\log (n/\delta)}{\eps^2})$, but this is of lower order than the previous bound. Summing over all $k\leq \lceil \lg n\rceil $, the total runtime is thus $O\p{\frac{n\log^3 (n/\delta)}{\eps^2}}$.
\end{proof}

\subsection{Subquadratic Time with Preprocessing}\label{sec:preprocess}

Our main theorem in the section is the following; see \Cref{sec:tech_overview_preprocessing} for a technical overview and discussion.

\corPreProcessing*

The starting point of our algorithm is the algorithm of \cite{mahalanabis2007density} discussed in \Cref{sec:tech_overview_preprocessing}. It is described in \Cref{algo:mahalanabis}.

\begin{algorithm}
\caption{Algorithm of \cite{mahalanabis2007density} \label{algo:mahalanabis}}
\begin{algorithmic}[1]
\State \textbf{input}: Hypothesis $H_1, \ldots, H_n$, list $L$ of pairs $\{H_i, H_j\}$ sorted in decreasing order by $\dtv(H_i, H_j) = \frac{1}2 \|H_i - H_j\|_1$.
\State $S \leftarrow \mathcal{H}$
\Repeat
    \State pick the first edge $\{H_i, H_j\}$ in $L$
    \If{$\hsemi{j}{i} < \hsemi{i}{j}$}
        \State $h' \leftarrow H_i$
    \Else
        \State $h' \leftarrow H_j$
    \EndIf
    \State remove $h'$ from $S$
    \State remove pairs containing $h'$ from $L$
\Until{$|S| = 1$}
\State output the distribution in $S$
\end{algorithmic}
\end{algorithm}

As stated, the algorithm takes quadratic time in the worst case since we need to traverse the list $L$ which is of size $\Theta(n^2)$. Our goal is to speed up this algorithm with polynomial preprocessing. Note that this is non-trivial since removing the losing hypothesis $h'$ from $S$ (as done in Line $9$) may cause $L$ to be significantly updated, since $h'$ could participate in many remaining pairs further down in the list. 

Before proceeding, we make one simplification. To compare two hypothesis in Line 5 of \Cref{algo:mahalanabis}, we compute (approximations) to the appropriate semi-distances by using samples from $P$. This computation takes $O(\log(n)/\eps^2)$ time since we need to count the number of samples of $P$ that fall in the set $\Sij$. Thus, it suffices to abstract away this computation and solve the \emph{Tournament Revelation Problem} defined in \cite{mahalanabis2007density} (and then multiply the overall runtime by $O(\log(n)/\eps^2)$). We recall its definition below.

\tournament*

As stated, the Tournament Revelation problem is very general and it is not clear if a sub-quadratic algorithm exists, even if preprocessing is allowed. However, we can exploit \emph{geometric} structure by viewing the hypothesis as vectors in $\ell_1$ space in dimension $|\calX|$. Then the largest weight edge is simply the pair which witnesses the $\ell_1$ diameter of the point set that remains. Thus, we seek a dynamic diameter data structure for the $\ell_1$ norm which takes $o(n)$ query time. However, we should not expect to obtain the \emph{exact} diameter in sublinear time.\footnote{e.g. for very similar geometric problems, one can show near quadratic time lower bound via SETH \cite{rubinstein2018hardness}.}
Instead, we aim for $1-\eps$ approximation of the diameter at every step, which \emph{is} possible in sublinear query time. We remark that this only affects the approximation guarantee of the algorithm mildly, which we touch upon shortly. 

Our first goal is to process the hypothesis so that we can instead work in a lower dimensional space. This will follow from applying two successive dimensionality reductions, one not so well known from \cite{nguyen2014algorithms}, and the other the classical Johnson-Lindenstrauss (JL) lemma \cite{johnson1984extensions}. The first result allows us to embed $\ell_1$ into $\ell_2^2$ in a slightly larger dimension. Then finally, we embed everything into logarithmic in $n$ dimensions using JL.

Note that the theorem stated below assumes an aspect ratio (ratio of the largest to smallest distances) of $\poly(1/\eps)$. This can be easily obtained in our setting by rounding the coordinates of the hypothesis, again viewed as vectors in $\ell_1$ to multiples of $\poly(\eps)$ and removing any duplicate vectors. This only affects the additive $\eps$ guarantee in \Cref{cor:preprocessing}. 

\begin{theorem}[Theorem 116 in \cite{nguyen2014algorithms}]
    Let  $H \subset \R^d$ be a point set. There is a mapping $f: H \rightarrow \R^{d'}$ for $d' = d \cdot \poly(1/\eps, \log d)$ such that for any $x,y \in H$,
     \[ \|x-y\|_1 \le \|f(x) - f(y)\|_2^2 \le (1+\eps)\|x-y\|_1.\]
     Furthermore, the mapping takes $O(d')$ time per point. 
\end{theorem}

The second result is the classic JL lemma.

\begin{theorem}[\cite{johnson1984extensions}]
     Let $H \subset \R^d$ be a point set. There is a mapping $f: H \rightarrow \R^{d'}$ for $d' = O(\log(|H|)/\eps^2)$ such that for any $x,y \in H$,
     \[ \|x-y\|_2^2 \le \|f(x) - f(y)\|_2^2 \le (1+\eps)\|x-y\|_2^2.\]
     Furthermore, the mapping takes $O(dd')$ time per point. 
\end{theorem}

Combining these two results gives the following corollary which allows us to view the hypothesis as vectors in $\ell_2^2$ in dimension $O(\log(|H|)/\eps^2)$ at the cost of distorting the original $\ell_1$ distances by multiplicative $1+\pm \eps$ factors.

\begin{corollary}\label{cor:dim_reduction}
      Let  $H \subset \R^d$ be a point set. There is a mapping $f: H \rightarrow \R^{d'}$ for $d' = O(\log(|H|)/\eps^2)$ such that for any $x,y \in H$,
     \[ \|x-y\|_1 \le \|f(x) - f(y)\|_2^2 \le (1+\eps)\|x-y\|_1.\]
     Furthermore, the mapping takes $d \cdot \poly(1/\eps, \log d, \log|H| )$ time per point. 
\end{corollary}

Now we turn our attention back to the Tournament Revelation problem. Recall that we need to a dynamic diameter data structure. Towards this, we first recall an algorithm of \cite{indyk2003better} for answering furthest neighbor queries, which one can think of as a query version of the diameter problem, where we ask for the furthest point (in a given set of points) to a given query point.

\begin{theorem}[\cite{indyk2003better}]\label{thm:indyk}
    Let $\alpha \in (0, 1)$. Given a set of $n$ points in $\ell_2^{2}$ in dimension $k$, there exists a datastructure which requires polynomial in $n$ and $k$ preprocessing time and can answer $1-\alpha$-approximate furthest neighbor queries and perform insertions or deletions in time $\tilde{O}(k n^{1-\alpha}/\alpha)$.
\end{theorem}

We also need the following reduction from \cite{eppstein1995dynamic} which gives a general conversion from a data structure answering furthest neighbor queries to one which which maintains a furthest neighbor \emph{within} a set of points undergoing deletions.\footnote{The method of \cite{eppstein1995dynamic} is more general, but we only state the guarantees needed for our purposes.}

\begin{theorem}[\cite{eppstein1995dynamic}]\label{thm:eppstein}
Let $d( p, q)$ be a distance function for which some data structure allows us to perform $1-\alpha$ furthest neighbor queries, and insert and delete points, in time $T(n, \alpha)$
per operation, and let $T(n,\eps)$ be monotonic in $n$ and satisfy $T(3n, \alpha) = O(T(n, \alpha))$.Then we can maintain a $1+\eps$ approximation to the diameter, as well as a pair of points realizing that distance, in amortized time $\tilde{O}(T(n, \alpha))$ per insertion, and $\tilde{O}(T(n, \alpha))$ per deletion.\footnote{\cite{eppstein1995dynamic} only discusses the case of exact queries ($\alpha = 0$), but one can easily verify that the method can also handle the case where the underlying data structure outputs approximate queries.}
\end{theorem}

Instantiating Theorem \ref{thm:eppstein} on the data structure of Theorem \ref{thm:indyk} gives us the following corollary:

\begin{corollary}\label{cor:diameter}
      Let $\alpha \in (0, 1)$. Given a set of $n$ points in $\ell_2^{2}$ in dimension $k$, there exists a datastructure which requires polynomial in $n$ and $k$ preprocessing time and can answer $1-\alpha$-approximate diameter queries and supports deletions in amortized $\tilde{O}(kn^{1-\alpha}/\alpha)$ time per operation. 
\end{corollary}

Combining everything proves the following result:

\begin{proof}[Proof of \Cref{cor:preprocessing}]
We simply instantiate the data structure of \Cref{cor:diameter} on the $n$ input hypothesis, viewed as vectors in $\ell_2^{2}$ in dimension $k = O(\log(n)/\eps^2)$, via  \Cref{cor:dim_reduction}. We can simulate the Tournament Revelation Problem of Definition \ref{def:tournament_problem} by retrieving a $1-\eps$ approximate diameter of the remaining hypothesis after every comparison step of Line 5 of \Cref{algo:mahalanabis}. Thus, the total running time across $n$ operations is $\tilde{O}(k \cdot n^{2-\eps}/\eps)$, coming from the $n$ applications of  $1-\eps$ approximate diameter queries, and an overhead of $O(\log(n)/\eps^2)$ per operation, coming from the time to simulate every comparison of line 5 of \Cref{algo:mahalanabis} by computing approximate semi-distances.

It only remains to prove the approximation factor. This is not immediate since the tournament revelation problem requires us to output the exact diameter at every stage. However, it is straightforward to check that obtaining approximate diameter pairs suffices. Indeed, Corollary 11 in \cite{mahalanabis2007density} already shows that if we instead only have access to $1\pm \eps$ approximation to the $\dtv$ distances, instead of exact values when constructing the list $L$ of \Cref{algo:mahalanabis}, then simulating the algorithm as is leads to an approximation factor of $3+\eps$, as desired.
\end{proof}

\section*{Acknowledgements}
The authors thank Sushruth Reddy for useful discussions leading to the algorithm of \Cref{sec:expected}, and Shay Moran for valuable discussion regarding \Cref{q:improper} and \Cref{sec:expected-lb}.

Anders Aamand was supported by the VILLUM Foundation grant 54451. Justin Y. Chen was supported by an NSF Graduate Research Fellowship under Grant No. 17453. Part of this work was done when Maryam Aliakbarpour and Justin Y. Chen were visiting the Simons Institute for  Theory of Computing as part of the program on Sublinear Algorithms.

\bibliographystyle{alpha}
\bibliography{bib}

\appendix

\section{Fast Quantile-Based Algorithm}
\label{sec:quantile}

\begin{algorithm}[h]
\caption{Faster Hypothesis Selection}\label{algo:fastselect}
\begin{algorithmic}[1]
\State Active set $S_0 \gets [n]$
\State $\ell \gets 0$
\While{$|S_\ell| > 0$}
    \State $\ell \gets \ell + 1$
    \For{$i \in [n]$}
        \State Sample $O(\log(n)/\delta)$ indices $R$ from  $S_{\ell-1}$ uniformly at random with replacement
        \State Choose minimal $a_i$ s.t. $\abs{\crb{j \in R : \hsemi{i}{j} \geq a_i}} \leq \ceil{2\delta |R|}$ 
        \Comment{$H_i$ kicks less than a $2\delta$ fraction above $a_i$}
    \EndFor
    \State $t_\ell \gets \max a_i$ and $i_\ell \gets \argmax a_i$
    \State $S_\ell \gets \crb{j \in S_{\ell - 1} : \hsemi{i_\ell}{j} < t_\ell}$
\EndWhile
\State $\ell' \gets \argmin t_\ell$
\State \Return $H_i$ where $i$ is chosen u.a.r.\ from $S_{\ell'-1}$
\end{algorithmic}
\end{algorithm}

For simplicity, we will assume that $\semi{i}{j}$ take on distinct values for all $i, j$.

\begin{theorem}\label{thm:fastselect}
Consider any $\eps > 0$ and $\delta > 1/n$.
\Cref{algo:fastselect} returns a hypothesis $H_i$ such that
\[
    \dtv(P, H_i) \leq 3 \cdot \OPT  + O(\eps)
\]
with probability $1 - O(\delta)$.
The algorithm uses $O\p{\frac{\log n}{\eps^2}}$ samples and runs in time $O\p{\frac{n \log^2 n}{\eps^2 \delta^2}}$.
\end{theorem}

\begin{proof}
We separately handle correctness and runtime. The sample complexity is fixed by definition in the algorithm.

\paragraph{Correctness}

We will first show that $t_{\ell'} \leq \OPT + O(\eps)$.
It suffices to show that there exists some $t_\ell \leq \OPT + O(\eps)$ as $\ell'$ is chosen to minimize $t_{\ell'}$.
Consider the round $\ell$ where $H_{i^*}$ is kicked out of the active set: $i^* \in S_{\ell-1} \setminus S_\ell$.
Note that such a round must exist as the algorithm only terminates once the active set is empty.
If $H_{i^*}$ is kicked out, it must be the case that $\hsemi{i_{\ell}}{i^*} \geq t_\ell$.
By \Cref{prop:semidist-underest}, $t_\ell \leq \OPT$.
The statement follows from \Cref{lem:sample-scheffe}.

Consider the empirical quantile $a_i$ given $m=O(\log(n)/\delta)$ random indices $R$ for a fixed $i, \ell$.
Let $b_i, c_i$ be the true $(1-3\delta)$ and $(1-\delta)$-quantiles of $\hsemi{i}{j}$ across all $j \in S_{\ell-1}$, respectively.
Let $X$ be the number of sampled elements indices where $\hsemi{i}{j} \geq b_i$, so $\E{X} = 3\delta m$.
Then, by a standard Chernoff bound,
\begin{equation*}
    \Pr{X \leq 2\delta m} \leq \exp(-\Omega(\delta m)).
\end{equation*}
Similarly, if we let $Y$ be the number of sampled indices where $\hsemi{i}{j} \geq c_i$, then
\begin{equation*}
    \Pr{Y \geq 2\delta m} \leq \exp(-\Omega(\delta m)).
\end{equation*}
By our choice of $m$, with high probability over $\poly(n)$ iterations of the algorithm, $a_i$ will be such that $\abs{j \in S_{\ell-1}: \hsemi{i}{j} \geq a_i} \in \sqb{\delta |S_{\ell-1}|, 3\delta |S_{\ell-1}|}$

Combined with the prior claim, for a $1-O(\delta)$ fraction of the $j \in S_{\ell'-1}$, $\semi{i^*}{j} \leq t_{\ell'} + \eps \leq \OPT + O(\eps)$.
Otherwise, there would be a larger threshold at round $\ell$.
Therefore, a $1-O(\delta)$ fraction of the hypotheses in $S_{\ell'-1}$ are valid approximate solutions.

\paragraph{Runtime}
The total number of iterations is $O\p{\frac{\log n}{\delta}}$ as the size of the active set decreases by a factor of $1-\Omega(\delta)$ in each round.
Within each round, computing $a_i$ requires $O\p{\frac{\log^2(n)}{\eps^2 \delta}}$ time.
As this computation must be done for all $i \in [n]$, the total runtime of the algorithm is
\begin{equation*}
    O\p{\frac{n \log^3 n}{\eps^2 \delta^2}}.
\end{equation*}
\end{proof}

\section{Any Constant Approximation Lower Bound}\label{sec:anyapprox-lb}

We briefly remark that $\Omega(\log(n)/\eps^2)$ samples are necessary to achieve any multiplicative factor guarantee for hypothesis selection, even allowing for outputting an improper distribution (i.e. even if we are allowed to output a distribution not among the known set of $n$ hypothesis). This follows from standard lower bounds from learning discrete distributions and has already been noted to imply lower bounds for hypothesis selection, e.g. see \cite{ChanDSS14}, but we present the full details below for completeness.

First let's define the following `hard' set $\mathcal{H}_k = \{h_1, \cdots, h_{2^k}\}$ of hypotheses. We set $k$ to be the domain size (suppose it is even for simplicity) and let $\eps \in (0, 1)$ be sufficiently small. Split the domain into $k/2$ pairs of size $2$ of consecutive elements. For every pair $\{i, i+1\}$, $\mathcal{H}_k$ consists of all distributions that either have probability mass $\left( \frac{1+\eps}k, \frac{1-\eps}k \right)$ or $\left( \frac{1-\eps}k, \frac{1+\eps}k \right)$ on $i, i+1$ respectively. Thus $|\mathcal{H}_k| = 2^{k/2}$.

The following is a folklore result (see~\cite{canonne2020}).

\begin{theorem}\label{thm:lb_learning_discrete}
    Let $\eps \in (0, 1/2)$. Let $\mathcal{A}$ be an algorithm that draws $m$ samples from a discrete distribution $P \in \mathcal{H}_k$ and with probability at least $9/10$ outputs a
distribution $\hat{P}$ at total variation distance at most $\eps$ from $P$. Then $m = \Omega(k/\eps^2)$.
\end{theorem}

\begin{theorem}\label{thm:C-lower-bound}
    Let $C > 0$ and $\eps \in (0, 1)$ be sufficiently small. Let $\mathcal{A}$ be an algorithm that draws $m$ samples from a discrete distribution $P \in H_k$ and with probability at least $9/10$ outputs a distribution $\hat{P}$ (possibly not in $H_k$) such that 
    \[ \dtv\left(P, \hat{P} \right) \le C \cdot \min_i \dtv(P, h_i) + \eps. \]
    Then $m = \Omega(k/\eps^2) = \Omega(\log(|\mathcal{H}_k|)/\eps^2)$.
\end{theorem}
\begin{proof}
We note that $ \min_i \dtv(P, h_i) = 0$ since $P \in \mathcal{H}_k$. Thus,
\[ \dtv(P, \hat{P}) \le C \cdot \min_i \dtv(P, h_i) + \eps = \eps. \]
The lower bound now follows from \Cref{thm:lb_learning_discrete}.
\end{proof}

\section{Omitted Proofs}\label{sec:omitted_proofs}

\begin{proof}[Proof of \Cref{prop:semidist-underest} ]
$\semi{i}{j}$ is equal to the absolute difference in mass assigned by $H_j$ and $P$ on some subset of the domain. The total variation distance is defined to be the maximum over such quantities.
\end{proof}

\begin{proof}[Proof of \Cref{prop:semidist-approx}] The first statement can be derived from the triangle inequality. 
\begin{align*}
    \dtv(P, H_j)
    &\leq \dtv(P, H_i) + \dtv(H_j, H_i) 
    \tag{Triangle inequality} \\
    &= \dtv(P, H_i) + \abs{H_j(\Sij) - H_i(\Sij)} \\
    &\leq \dtv(P, H_i) + \abs{H_j(\Sij) - P(\Sij)} + \abs{H_i(\Sij) - P(\Sij)}
    \tag{Triangle inequality} \\
    &= \dtv(P, H_i) + \semi{i}{j} + \semi{j}{i}\,.
\end{align*}

The second statement is derived by setting $i$ equal to $i^*$, and noting: $\semi{j}{i^*} \leq \dtv(P, H_{i^*}) \leq \OPT\,.$
\end{proof}

\begin{proof}[Proof of \Cref{lem:sample-scheffe}]
Recall our computational model: we consider taking a single sample from $P$ or any distribution in $\HH$ as well as a single query to the pdf of any distribution in $\HH$ to take $O(1)$ time. 

First, assume for now that we know the value $H_j(\Sij)$. Let $X$ be a random variable for the number of sampled elements from $P$ which land in $\Sij$, so that $\hsemi{i}{j} = \abs{H_j(\Sij) - X/s}$. Note that $X$ can be calculated by querying the pdf of $H_i$ and $H_j$ on every sample from $P$ and checking the ordering of the returned values (see \Cref{def:scheffe_set}). Note that $\E{X} = s P(\Sij)$, so $\abs{\semi{i}{j} - \hsemi{i}{j}} \leq \abs{X - \E{X}}/s$.
The result follows from a standard Hoeffding bound on the deviation of the sum of i.i.d.\ indicator random variables.
By repeating this procedure for all pairs of distributions and reusing samples, each with a failure probability of $\delta/n^2$, a union bound guarantees an overall success probability of at least $1-\delta$.

Now we remove the assumption that we know $H_j(\Sij)$ by approximating it up to additive error $\eps$ in our computational model (where note that sampling from any distribution in the known set $\HH$ takes $O(1)$ time). Indeed, take $s' = O\left(\frac{ \log(n/\delta)}{\epsilon^2}\right)$ samples from $H_j$ (which takes $O(s')$ time). By querying the density functions of $H_j$ and $H_i$ (again taking $O(s')$ time), we obtain a random variable $X'$ satisfying $\E{X'} = s'H_j(\Sij)$. Again, the standard Hoeffding bound on the deviation of a sum of i.i.d. bounded random variables implies that $|X'/s' - H_j(\Sij)| \le \eps$ with probability $\delta/n^2$. Thus a union bound implies we can estimate all quadratically many $H_j(\Sij)$ values up to additive error $\eps$ with probability $1-\delta$ in total time $O(n^2 s')$. By the triangle inequality, substituting our approximation for $H_j(\Sij)$ only changes our estimate for $\hsemi{i}{j}$ by an additive $\eps$ factor, and so the claimed approximation bound follows by adjusting $\eps$ and $\delta$ values by constant factors. Furthermore, the time bound follows from our computational model where sampling from $P$ or any distribution in $\HH$ and querying the pdf of any distribution in $\HH$ all take $O(1)$ time. 
\end{proof}

\begin{proof}[Proof of \Cref{claim:two_dist_ev}]
Assume without loss of generality that $i^*=1$.
Then by \Cref{prop:semidist-underest} and \Cref{prop:semidist-approx},
\begin{align*}
    \E[]{\dtv(P, H)}
    &= \frac{\semi{1}{2}}{\semi{1}{2} + \semi{2}{1}} \cdot \dtv(P, H_1)
    + \frac{\semi{2}{1}}{\semi{1}{2} + \semi{2}{1}} \cdot \dtv(P, H_2) \\
    &\leq \frac{\semi{1}{2}}{\semi{1}{2} + \semi{2}{1}} \cdot \OPT
    + \frac{\semi{2}{1}}{\semi{1}{2} + \semi{2}{1}} \cdot \p{\OPT + \semi{2}{1}  + \semi{1}{2}} \\
    &= \OPT + \semi{2}{1}\leq 2\OPT. 
\end{align*}
\end{proof}

\end{document}